\documentclass[american,12pt]{article}
\usepackage[T1]{fontenc}
\usepackage[utf8]{inputenc}
\usepackage{url}
\usepackage{bm}
\usepackage{amsmath}
\usepackage{amsthm}
\usepackage{amssymb}
\usepackage{graphicx}
\usepackage{geometry}
\usepackage{setspace}
\usepackage[authoryear]{natbib}
\makeatletter

\providecommand{\tabularnewline}{\\}

\newcommand{\lyxaddress}[1]{
	\par {\raggedright #1
	\vspace{1.4em}
	\noindent\par}
}
\theoremstyle{plain}
\newtheorem{thm}{\protect\theoremname}
\theoremstyle{plain}
\newtheorem{prop}[thm]{\protect\propositionname}
\theoremstyle{plain}
\newtheorem{lem}[thm]{\protect\lemmaname}
\theoremstyle{remark}
\newtheorem{rem}[thm]{\protect\remarkname}

\usepackage{hyperref}

\makeatother

\usepackage{babel}
\providecommand{\lemmaname}{Lemma}
\providecommand{\propositionname}{Proposition}
\providecommand{\remarkname}{Remark}
\providecommand{\theoremname}{Theorem}

\begin{document}
\title{Statistical process control via $p$-values}
\author{Hien Duy Nguyen$^{1,2}$ and Dan Wang$^{2,3}$}
\maketitle

\lyxaddress{$^{1}$Institute of Mathematics for Industry, Kyushu University,
Japan. $^{2}$School of Computing, Engineering, and Mathematical Sciences,
La Trobe University, Australia. $^{3}$Department of Statistics, School
of Mathematics, Northwest University, Xi’an, China}
\begin{abstract}
We study statistical process control (SPC) through charting of $p$-values.
When in control (IC), any valid sequence $(P_{t})_{t}$ is super-uniform,
a requirement that can hold in nonparametric and two-phase designs
without parametric modelling of the monitored process. Within this
framework, we analyse the Shewhart rule that signals when $P_{t}\le\alpha$.
Under super-uniformity alone, and with no assumptions on temporal
dependence, we derive universal IC lower bounds for the average run
length (ARL) and for the expected time to the $k$th false alarm ($k$-ARL).
When conditional super-uniformity holds, these bounds sharpen to the
familiar $\alpha^{-1}$ and $k\alpha^{-1}$ rates, giving simple,
distribution-free calibration for $p$-value charts. Beyond thresholding,
we use merging functions for dependent $p$-values to build EWMA-like
schemes that output, at each time $t$, a valid $p$-value for the
hypothesis that the process has remained IC up to $t$, enabling smoothing
without ad hoc control limits. We also study uniform EWMA processes,
giving explicit distribution formulas and left-tail guarantees. Finally,
we propose a modular approach to directional and coordinate localisation
in multivariate SPC via closed testing, controlling the family-wise
error rate at the time of alarm. Numerical examples illustrate the
utility and variety of our approach.
\end{abstract}

\section{Introduction}

Statistical process control (SPC) comprises an important class of
methods for monitoring the behaviour of manufacturing, engineering,
and other production processes. Since their introduction by \citet{Shewhart1931},
the SPC literature has grown significantly, with works such as \citet{Zwetsloot:2021aa},
\citet{waqas2024controlcharts}, \citet{ahmed2025highdimensional},
and \citet{Tsung2025} providing good recent references on the current
state of the art. Good introductions to the SPC literature appear
in the texts of \citet{xie2002statistical}, \citet{qiu2013introduction},
and \citet{chakraborti2019nonparametric}.

Typically, when conducting SPC, one seeks to monitor a process whose
in-control state can be characterised as generating data from a null
distribution, defined by a null hypothesis $\mathrm{H}_{0}$. One
then wishes to detect whether the process becomes out-of-control (OOC),
characterised as a deviation away from $\mathrm{H}_{0}$. In the discrete-time
setting, at time point $t\in\mathbb{N}$, the analyst assesses deviation
from the in-control (IC) state by sampling $n_{t}\in\mathbb{N}$ observations
from the data-generating process (DGP) and computing a test statistic
$Z_{t}$ whose distributional characteristics are known under $\mathrm{H}_{0}$.
The analyst then monitors the sequence of test statistics $\left(Z_{t}\right)_{t}$
and raises an alarm whenever the test of $\mathrm{H}_{0}$ using $Z_{t}$
is rejected at some specified level $\alpha>0$. This choice controls
the instantaneous false-alarm rate (FAR) of the monitoring process
and the IC average run length, typically via a control limit for $Z_{t}$
that can be characterised as a hypothesis-testing critical value depending
on $\alpha$. This is the essence of the general scheme introduced
by \citet{Shewhart1931}. A more sophisticated analyst can further
employ time-point averaging methods, such as the cumulative sum (CUSUM)
and exponentially weighted moving average (EWMA) charts introduced
by \citet{page1954} and \citet{roberts1959}, respectively. The null
distributions of these chart statistics can be obtained in specific
parametric situations, such as under assumptions of normality (see,
e.g., \citealp{roberts1959}, \citealp{LucasSaccucci1990}, \citealp{LowryWoodallChampRigdon1992},
\citealp{LuReynolds1999}, \citealp{yeh2003multivariate}, and \citealp{Knoth2005_EWMA_S2_ARL}),
but can be difficult to deduce in general.

For any hypothesis $\mathrm{H}_{0}$, one can typically construct
a sequence of $p$-values $\left(P_{t}\right)_{t}$, based on samples
of sizes $\left(n_{t}\right)_{t}$, which can be used as test statistics
for $\mathrm{H}_{0}$ in place of $\left(Z_{t}\right)_{t}$. Here,
we define the $p$-value as a random variable $P$ taking values in
the unit interval such that, under the null $\mathrm{H}_{0}$, the
probability that $P\le\alpha$ is less than or equal to $\alpha$
for each $\alpha>0$, following the general definition of \citet{Lehmann2022}.
The use of $p$-values in SPC has been explored in the works of \citet{GriggSpiegelhalter2008},
who derive and make use of the steady-state distribution of CUSUM
statistics to compute $p$-values for monitoring; \citet{GriggSpiegelhalterJones2009}
and \citet{LiTsung2009}, who consider the use of $p$-values and
false discovery rate control to simultaneously monitor multiple processes;
and \citet{LiQiuChatterjeeWang2013}, who suggest the charting of
$p$-values as a unified framework for conducting SPC, particularly
in the presence of variable sampling intervals (VSI) and variable
sample sizes (VSS) in each time period.

In this work, we continue to develop the use of the $p$-value as
a charting tool for SPC. In particular, we demonstrate that $p$-values
are applicable in both one- and two-phase monitoring programmes and,
like \citet{LiQiuChatterjeeWang2013}, we discuss how $p$-values
are applicable in VSI and VSS settings. Without any assumption on
stochastic dependence within the sequence $\left(P_{t}\right)_{t}$,
we prove that the $p$-value chart that raises an alarm if $P_{t}\le\alpha$
has an average run length (ARL) that is always bounded below by $\left(1+\alpha^{-1}\right)/2$,
and under a conditional super-uniformity assumption, we prove that
this can be improved to the typical $\alpha^{-1}$ bound of parametric
methods. To extend these observations, motivated by the notion of
$k$-FWER (generalised family-wise error rate; the case $k=1$ corresponds
to the usual family-wise error rate (FWER); \citealp{LehmannRomano2005gFWER}),
we also study the $k$-ARL, defined as the average time until first
observing $k$ false alarms under the null hypothesis, which generalises
the usual ARL (corresponding to the $k=1$ case). Without dependence
assumptions, we prove that $k$-ARL is bounded below by $\left(1+k\alpha^{-1}\right)/2$,
which again improves to $k\alpha^{-1}$ under conditional super-uniformity
assumptions. These $k$-ARL bounds allow the $p$-value charts to
be correctly calibrated for all DGPs where $\mathrm{H}_{0}$ holds
and mitigate against the time-consuming task of simulation-based ARL
calibration, which only holds under the specific DGP assumptions that
underlie the simulation study.

Further to these results, we demonstrate that, like well-behaved test
statistics, one can construct EWMA-like schemes $\left(Q_{t}\right)_{t}$
based on a $p$-value sequence $\left(P_{t}\right)_{t}$, such that
each $Q_{t}$ is a $p$-value for the hypothesis that $\mathrm{H}_{0}$
is true for all time points up to time $t$, against the complementary
hypothesis. We demonstrate that this is possible using recent results
for combining $p$-values \citep{Vovk2020,Vovk2022}. Thus, given
any sequence of test statistics $\left(Z_{t}\right)_{t}$ whose null
distributions admit $p$-values, one can always generate EWMA-like
schemes for SPC. We further consider and follow on from the work on
EWMA charts based on uniform test-statistic constructions of \citet{yeh2003multivariate}
and obtain novel results regarding the weighted averaging of independent
and identically distributed random variables. This permits the general
and generic construction of EWMA versions of one-phase parametric
Shewhart-type charts with ARL control that does not depend on simulated
calibration or specific test statistic distributions.

Using a $p$-value-based closed-testing procedure (e.g., \citealp{MarcusPeritzGabriel1976}),
we show how $p$-value charts can detect any departure of the multivariate
parameter from a null baseline as well as localise the affected coordinates
and directions of deviation, in a principled way, with simultaneous
error control. This provides an alternative to the test-statistic-based
approaches to the same problem of \citet{Hawkins1991RegressionAdjusted},
\citet{MasonTracyYoung1995MYT}, \citet{TanShi2012BayesianMeanShifts},
\citet{XuShuLiWang2023Technometrics}, \citet{XuDeng2023CIE}, and
\citet{XiongXu2025JSCS}, for example. We then conclude with demonstrations
of the validity and applicability of our approaches via examples and
numerical simulations, where we present novel methods for monitoring
the stability of the distribution of a process using $p$-values from
Kolmogorov--Smirnov tests, and provide a nonparametric method for
localising the direction and coordinate of change via Mann--Whitney
$U$-tests (see \citealp{GibbonsChakraborti2010} regarding nonparametric
tests).

The remainder of the manuscript proceeds as follows. In Section 2,
we present our general technical framework, together with our ARL
results and EWMA constructions. Section 3 then contains the application
of $p$-values to the problem of localised detection for multivariate
charts. Examples and numerical simulations are presented in Section
4, and Section 5 contains concluding remarks. Proofs of theorems and
additional technical and numerical results appear in the Appendix. 

\section{Main results}

\label{sec:Main-results}

\subsection{Technical setup}

Let $\left(\Omega,\mathfrak{A},\mathrm{P}\right)$ be a probability
space on which all of our random objects are supported, with generic
element $\omega$ and expectation symbol $\mathrm{E}$. We suppose
that the fundamental object that we wish to monitor via an SPC process
is the random object $X:\Omega\to\mathbb{X}$, where $\left(\mathbb{X},\mathfrak{B}_{\mathbb{X}}\right)$
is a metric space equipped with its Borel $\sigma$-algebra. For each
$t\in\mathbb{N}$, let $X_{t}$ denote a time-$t$ replicate of $X$
and let $\mathrm{P}_{X,t}\in{\cal M}_{\mathbb{X}}$ denote the push-forward
measure of $X_{t}$ at time $t$, where ${\cal M}_{\mathbb{X}}$ is
the set of Borel probability measures on $\mathbb{X}$. We will say
that the process is IC at time $t$ if $\mathrm{P}_{X,t}$ satisfies
the null hypothesis $\mathrm{H}_{0}$, where the null hypothesis is
some property of a measure in ${\cal M}_{\mathbb{X}}$ (e.g., if $\mathbb{X}=\mathbb{R}$,
$\mathrm{H}_{0}$ may be that $\mathrm{E}X_{t}=0$).

We consider two types of SPC designs: the one-phase and the two-phase
designs. In the one-phase design, at each time $t$, we realise an
IID sample $\mathbf{X}_{t}=\left(X_{t,i}\right)_{i}$ of replicates
of $X_{t}$ with probability measure $\mathrm{P}_{X,t}$, for $i\in\left[n_{t}\right]=\left[1,n_{t}\right]\cap\mathbb{N}$.
Using $\left(\mathbf{X}_{t}\right)_{t}$, we then compute the $p$-value
$P_{t}=p_{t}\left(\mathbf{X}_{1},\dots,\mathbf{X}_{t}\right)$ for
some functions $p_{t}:\mathbb{X}^{\sum_{s=1}^{t}n_{s}}\to\left[0,1\right]$,
and say that $\left(P_{t}\right)_{t}$ satisfies the super-uniformity
property (under $\mathrm{H}_{0}$) if, for each $t$, 
\begin{equation}
\mathrm{P}_{0,t}\left(P_{t}\le\alpha\right)\le\alpha\text{, for every }\alpha\in\left[0,1\right]\text{ and every }\mathrm{P}_{0,t}\in{\cal M}_{0,t}\text{.}\label{eq:superuniformity}
\end{equation}
Here, ${\cal M}_{0,t}$ denotes the set of probability measures $\mathrm{P}_{0,t}$
on $\left(\Omega,\mathfrak{A}\right)$ such that the push-forward
measure of $X_{s}$ under $\mathrm{P}_{0,t}$ satisfies $\mathrm{H}_{0}$,
for every $s\in\left[t\right]$.

Alternatively, in a two-phase design we first realise a so-called
Phase~I sample $\mathbf{X}_{0}=\left(X_{0,i}\right)_{i}$ of IID
replicates of $X_{0}$ with measure $\mathrm{P}_{X,0}$, where $i\in\left[n_{0}\right]$.
The null hypothesis $\mathrm{H}_{0}$ is then defined relative to
the measure $\mathrm{P}_{X,0}$; for example, we may consider the
equality in distribution hypothesis that $\mathrm{P}_{X,t}=\mathrm{P}_{X,0}$.
Like the one-phase design, we then realise IID samples $\left(X_{t,i}\right)_{i}$
at each time $t$ of size $n_{t}$, arising from the measure $\mathrm{P}_{X,t}$.
Since $\mathrm{H}_{0}$ at each $t$ is a hypothesis concerning $\mathrm{P}_{X,0}$
along with the measures $\mathrm{P}_{X,t}$, the $p$-value for the
test must be taken as a function of both samples $\mathbf{X}_{0}$
and $\mathbf{X}_{1},\dots,\mathbf{X}_{t}$: $P_{t}=p_{t}\left(\mathbf{X}_{0},\dots,\mathbf{X}_{t}\right)$
for some $p_{t}:\mathbb{X}^{\sum_{s=0}^{t}n_{s}}\to\left[0,1\right]$,
and again is taken to satisfy the super-uniformity property (\ref{eq:superuniformity}).

In either case, we obtain a sequence of $p$-values $\left(P_{t}\right)_{t}$
that test $\mathrm{H}_{0}$ at each time point $t$, where we do not
assume any particular temporal dependence structure between the $p$-values.
The sequence constitutes a control chart and an alarm is triggered
at time $t$ if $P_{t}<\alpha$ for some control limit $\alpha$ chosen
by the analyst. Here, the control limit is the usual size of the hypothesis
test for $\mathrm{H}_{0}$, and the super-uniformity property ensures
that, under $\mathrm{H}_{0}$, the probability of a false alarm at
time $t$ is at most $\alpha$. 

\subsection{Controlling the average run length}

The ARL is defined as the expected time to first alarm, assuming that
$\mathrm{P}_{X,t}$ satisfies $\mathrm{H}_{0}$ for each time $t$.
Let $R:\Omega\to\mathbb{N}$ denote the run length. For each $m\in\mathbb{N}$,
we can write 
\[
\mathrm{P}_{0}\left(R\le m\right)=\mathrm{P}_{0}\left(\bigcup_{t=1}^{m}\left\{ P_{t}\le\alpha\right\} \right)\le\sum_{t=1}^{m}\mathrm{P}_{0}\left(P_{t}\le\alpha\right)\le m\alpha,
\]
by super-uniformity, where $\mathrm{P}_{0}$ is any measure on $\left(\Omega,\mathfrak{A}\right)$
in ${\cal M}_{0,t}$ for all $t\in\mathbb{N}$. We thus have $\mathrm{P}_{0}\left(R>m\right)\ge1-m\alpha$
for every $m\ge0$, and thus for each $m\ge1$, 
\[
\mathrm{P}_{0}\left(R\ge m\right)=\mathrm{P}_{0}\left(R>m-1\right)\ge1-\left(m-1\right)\alpha.
\]
Using the survival function expression for the expectation (cf. \citealp[Thm. 12.1, p. 75]{Gut2013Probability}),
we can express the ARL as 
\begin{equation}
\mathrm{ARL}=\mathrm{E}_{0}R=\sum_{m=1}^{\infty}\mathrm{P}_{0}\left(R\ge m\right)\ge\sum_{m=1}^{\infty}\bigl[1-\left(m-1\right)\alpha\bigr]_{+},\label{eq:ARL-expression}
\end{equation}
where $\mathrm{E}_{0}$ is the expectation operator for $\mathrm{P}_{0}$
and $\left[a\right]_{+}=\max\{a,0\}$. 
\begin{prop}
\label{prop:ARLnoconditions} For each $\alpha\in\left(0,1\right]$,
if $\left(P_{t}\right)_{t}$ satisfies (\ref{eq:superuniformity}),
then 
\[
\mathrm{ARL}\ge\frac{1}{2\alpha}+\frac{1}{2}\text{.}
\]
\end{prop}

\begin{proof}
All proofs appear in the Appendix. 
\end{proof}
The bound above can be improved given some assumptions on conditional
super-uniformity of the $p$-values. In particular, suppose that $\left({\cal F}_{t}\right)_{t}$
is a filtration on the space $\left(\Omega,\mathfrak{A}\right)$,
whereupon the following assumption holds: 
\begin{equation}
\mathrm{P}_{0}\left(P_{t}\le\alpha\mid{\cal F}_{t-1}\right)\le\alpha\text{, a.s., for every }\alpha\in\left[0,1\right]\text{.}\label{eq:Conditional_Super_Uni}
\end{equation}

\begin{prop}
\label{prop:ARL_with_conditional_indep} Suppose that $\left(P_{t}\right)_{t}$
satisfies (\ref{eq:Conditional_Super_Uni}) for some filtration $\left({\cal F}_{t}\right)_{t}$.
Then, for each $\alpha\in\left(0,1\right]$, 
\[
\mathrm{ARL}\ge\frac{1}{\alpha}\text{.}
\]
\end{prop}

With ${\cal F}_{t}=\sigma\left(\mathbf{X}_{1},\dots,\mathbf{X}_{t}\right)$,
observe that the conditional super-uniformity assumption (\ref{eq:Conditional_Super_Uni})
is satisfied in the one-phase design whenever the $p$-value at time
$t$ depends only on the time-$t$ sample $\mathbf{X}_{t}$ and $\mathbf{X}_{t}$
is independent of ${\cal F}_{t-1}$ under $\mathrm{P}_{0}$, i.e.,
$P_{t}=p_{t}\left(\mathbf{X}_{t}\right)$. Similarly, for a two-phase
design, with ${\cal F}_{t}=\sigma\left(\mathbf{X}_{0},\mathbf{X}_{1},\dots,\mathbf{X}_{t}\right)$,
the condition is satisfied whenever the time-$t$ $p$-value depends
only on the Phase~I sample $\mathbf{X}_{0}$ and the time-$t$ sample
$\mathbf{X}_{t}$ and $\mathbf{X}_{t}$ is conditionally independent
of $\left(\mathbf{X}_{1},\dots,\mathbf{X}_{t-1}\right)$ given $\mathbf{X}_{0}$
under $\mathrm{P}_{0}$, i.e., $P_{t}=p_{t}\left(\mathbf{X}_{0};\mathbf{X}_{t}\right)$.
These are among the most common cases for $p$-value constructions,
and thus the assumption can be considered reasonable. In the sequel,
we will show that certain EWMA-like schemes also satisfy condition
(\ref{eq:Conditional_Super_Uni}).

We note that the bound of Proposition \ref{prop:ARL_with_conditional_indep}
is, in fact, tight. This can be seen by considering an IID sequence
$\left(P_{t}\right)_{t}$ adapted to $\left({\cal F}_{t}\right)_{t}$,
where $P_{t}\sim\mathrm{Unif}\left(0,1\right)$ under $\mathrm{H}_{0}$
and therefore $\mathrm{P}_{0}\left(P_{t}\le\alpha\,\big|\,{\cal F}_{t-1}\right)=\alpha$
for each $\alpha\in\left[0,1\right]$ and $t\in\mathbb{N}$, thus
satisfying (\ref{eq:Conditional_Super_Uni}). It then holds that $\mathrm{P}_{0}\left(R>t\right)=\left(1-\alpha\right)^{t}$
and therefore $\mathrm{E}_{0}R=\alpha^{-1}$. This is the usual computation
of the ARL for parametric Shewhart-type charts (see, e.g., \citealp[Sec. 3.6.1]{chakraborti2019nonparametric}). 

\subsection{Controlling the $k$-ARL}

Like the ARL, we define the $k$-ARL as the expected time until $k$
alarms are raised, assuming that $\mathrm{P}_{X,t}$ satisfies $\mathrm{H}_{0}$
for each time $t$, and let $R_{k}:\Omega\to\mathbb{N}$ denote the
time until $k$ alarms are raised. That is, the $k$-ARL is equal
to $\mathrm{E}_{0}R_{k}$. Note that the usual ARL and run length
are the $1$-ARL and $R_{1}$, respectively. Under only super-uniformity,
we have the following result. 
\begin{prop}
\label{prop:kARL_superuniform} For each $\alpha\in\left(0,1\right]$,
if $\left(P_{t}\right)_{t}$ satisfies (\ref{eq:superuniformity}),
then, for any $k\ge1$, 
\[
\mathrm{E}_{0}R_{k}\ge\left(\nu+1\right)\left(1-\frac{\alpha\nu}{2k}\right)\ge\frac{k}{2\alpha}+\frac{1}{2}\text{,}
\]
where $\nu=\left\lfloor k/\alpha\right\rfloor $. 
\end{prop}

Under the stronger assumption of Proposition~\ref{prop:ARL_with_conditional_indep},
we also have an analogous result for controlling the $k$-ARL. Unfortunately,
we do not know an elementary proof of this result, as we did for Proposition
\ref{prop:ARL_with_conditional_indep}, and must instead resort to
martingale theory. 
\begin{prop}
\label{prop:Conditional_Super_Unif_kARL} Suppose that $\left(P_{t}\right)_{t}$
satisfies (\ref{eq:Conditional_Super_Uni}) for some filtration $\left({\cal F}_{t}\right)_{t}$.
Then, for each $\alpha\in\left(0,1\right]$ and $k\ge1$, 
\[
\mathrm{E}_{0}R_{k}\ge\frac{k}{\alpha}\text{.}
\]
\end{prop}

We note that we are not the first to study the time to $k$ alarms,
$R_{k}$, although we cannot find specific instances of the $k$-ARL
being studied in the $p$-value context in the literature. Prior work
includes the study of $k>1$ failure rules for negative binomial charts
\citep{Albers2010_NBcharts_JSPI,Albers2011_Overdispersion_Metrika},
and conforming-count charts (CCC-$r$) \citep{XieGohLu1998_CIM_CIE,OhtaKusukawaRahim2001_CCCr_QREI,JoekesSmrekarRighetti2016_CCCr_QE,Chen2025CCCr}.
In these works, the control of the $k$-ARL is considered parametrically
under a negative binomial model, in which conformity of the process
being measured is treated as an IID Bernoulli event. In contrast,
our approach is model-free and relies only on the availability of
super-uniform $p$-values that satisfy either (\ref{eq:superuniformity})
or (\ref{eq:Conditional_Super_Uni}).

Finally, we note that the bound from Proposition \ref{prop:Conditional_Super_Unif_kARL}
is sharp by the same reasoning as that of Proposition \ref{prop:ARL_with_conditional_indep}.
Namely, suppose that we take $P_{t}\sim\mathrm{Unif}\left(0,1\right)$
IID and independent of ${\cal F}_{t-1}$. Let $I_{t}$ denote the
indicator of the event $\left\{ P_{t}\le\alpha\right\} $. Then, $\left(I_{t}\right)_{t}$
are IID $\mathrm{Bern}\left(\alpha\right)$. If we define $G_{1}=\inf\left\{ t\ge1:I_{t}=1\right\} $
and $G_{j}=\inf\left\{ t\ge1:\sum_{s=1}^{t}I_{s}=j\right\} -\inf\left\{ t\ge1:\sum_{s=1}^{t}I_{s}=j-1\right\} $,
for each $j\ge2$, then $G_{j}\sim\mathrm{Geom}\left(\alpha\right)$
for each $j\in\mathbb{N}$, where $G_{j}$ is the inter-alarm waiting
time and $\mathrm{E}_{0}G_{j}=1/\alpha$. Then, the time to $k$ alarms
is simply the sum of the first $k$ inter-alarm times: $R_{k}=\sum_{j=1}^{k}G_{j}$,
and by linearity, we have $\mathrm{E}_{0}R_{k}=k/\alpha$. 

\subsection{Exponentially weighted moving average charts}

To construct our exponentially weighted moving average charts, we
rely on results from \citet{Vovk2020} that permit us to average arbitrarily
dependent $p$-values. For $p$-values $P_{1},\dots,P_{m}$ satisfying
(\ref{eq:superuniformity}), weights $\bm{w}=\left(w_{1},\dots,w_{m}\right)\in\left[0,1\right]^{m}$
such that $\sum_{t=1}^{m}w_{t}=1$, and $r>-1$ with $r\neq0$, we
define the weighted generalised mean 
\[
\left(p_{1},\dots,p_{m}\right)\mapsto\mathrm{M}_{r,\bm{w}}\left(p_{1},\dots,p_{m}\right)=\left(\sum_{t=1}^{m}w_{t}p_{t}^{r}\right)^{1/r}\text{.}
\]
We say that the function $q:\left[0,1\right]^{m}\to\mathbb{R}_{\ge0}$
is a valid merging function if $\mathrm{P}_{0}\left(q\left(P_{1},\dots,P_{m}\right)\le\alpha\right)\le\alpha$
for each $\alpha\in\left[0,1\right]$. The following result is a consequence
of \citet{Vovk2020}. 
\begin{lem}
\label{lem:p-value=00003D00003D00003D00003D000020averaging=00003D00003D00003D00003D000020Lemma}Let
$r>-1$ with $r\neq0$ and $w_{\mathrm{max}}=\max_{t\in\left[m\right]}w_{t}$. 
\begin{enumerate}
\item For $r>-1$, the map 
\[
q=\left(1+r\right)^{1/r}\mathrm{M}_{r,\bm{w}}
\]
is a valid merging function. 
\item For $r\in\left[1,\infty\right)$, the map 
\[
q=\bigl(\min\{1+r,w_{\mathrm{max}}^{-1}\}\bigr)^{1/r}\mathrm{M}_{r,\bm{w}}
\]
is a valid merging function. 
\end{enumerate}
\end{lem}

For a generic sequence of test statistics $\left(Z_{t}\right)_{t}$
and a learning rate $\lambda\in\left(0,1\right)$, an EWMA chart $\left(\tilde{Z}_{\lambda,t}\right)_{t}$
is typically defined via the one-step expression: 
\[
\tilde{Z}_{\lambda,t}=\lambda Z_{t}+\left(1-\lambda\right)\tilde{Z}_{\lambda,t-1}\text{,}
\]
for each $t\in\mathbb{N}$, given some initial $\tilde{Z}_{\lambda,0}$.
Unfortunately, such a scheme is not implementable for $p$-values,
in the sense that if $\left(P_{t}\right)_{t}$ satisfies (\ref{eq:superuniformity})
and 
\[
\tilde{P}_{\lambda,t}=\lambda P_{t}+\left(1-\lambda\right)\tilde{P}_{\lambda,t-1}\text{,}
\]
then $\left(\tilde{P}_{\lambda,t}\right)_{t}$ need not be a sequence
of super-uniform $p$-values under $\mathrm{P}_{0}$. We thus require
a modified approach, making use of Lemma \ref{lem:p-value=00003D00003D00003D00003D000020averaging=00003D00003D00003D00003D000020Lemma}.

\subsubsection{EWMA-like charts via $p$-value averaging }

\label{subsec:EWMA-like-charts-p-avg}

Starting with $\left(P_{t}\right)_{t\in\mathbb{N}}$ satisfying (\ref{eq:superuniformity}),
for each $\lambda\in\left(0,1\right)$ and $r>-1$ with $r\neq0$
we define the EWMA recursion directly from the first $p$-value by
\[
S_{\lambda,1}^{\left(r\right)}=P_{1}^{r},\text{ and }S_{\lambda,t}^{\left(r\right)}=\lambda P_{t}^{r}+\left(1-\lambda\right)S_{\lambda,t-1}^{\left(r\right)},\text{ for }t\ge2.
\]
By induction one checks that, for each $t\in\mathbb{N}$, 
\begin{equation}
S_{\lambda,t}^{\left(r\right)}=\sum_{s=1}^{t}w_{t,s}P_{s}^{r},\text{ }w_{t,1}=(1-\lambda)^{t-1},\text{ and }w_{t,s}=\lambda(1-\lambda)^{t-s},\;\mathrm{for\ }2\le s\le t,\label{eq:EWMA-no-P0-weights}
\end{equation}
so that $w_{t,s}\ge0$ and $\sum_{s=1}^{t}w_{t,s}=1$.

Let 
\[
w_{t,\max}=\max_{1\le s\le t}w_{t,s}=\max\left\{ \lambda,(1-\lambda)^{t-1}\right\} .
\]
For each $\lambda\in(0,1)$, $r>-1$ with $r\neq0$, and $t\in\mathbb{N}$
we now set 
\begin{equation}
a_{\lambda,t}^{\left(r\right)}=\begin{cases}
\min\left\{ 1+r,\frac{1}{w_{t,\max}}\right\} ^{1/r} & \mathrm{if\ }r\ge1,\\[0.3em]
\left(1+r\right)^{1/r} & \mathrm{if\ }r\in\left(-1,1\right)\setminus\{0\},
\end{cases}\text{ and }Q_{\lambda,t}^{\left(r\right)}=a_{\lambda,t}^{\left(r\right)}\left[S_{\lambda,t}^{\left(r\right)}\right]^{1/r}.\label{eq:Q-def}
\end{equation}
The fact that $\left(Q_{\lambda,t}^{\left(r\right)}\right)_{t}$ is
EWMA-like can be seen most clearly when $r=1$, in which case 
\[
Q_{\lambda,1}^{\left(1\right)}=P_{1}
\]
and, for each $t\ge2$, 
\begin{equation}
Q_{\lambda,t}^{\left(1\right)}=a_{\lambda,t}^{\left(1\right)}\left\{ \lambda P_{t}+\left(1-\lambda\right)\frac{Q_{\lambda,t-1}^{\left(1\right)}}{a_{\lambda,t-1}^{\left(1\right)}}\right\} ,\label{eq:Q-r1-recursion}
\end{equation}
so that $Q_{\lambda,t}^{(1)}$ is a rescaled EWMA of the $p$-values
$P_{t}$. 
\begin{prop}
\label{prop:EWMA-like-charts} For each $\lambda\in\left(0,1\right)$
and $r>-1$ with $r\neq0$, if $\left(P_{t}\right)_{t}$ satisfies
(\ref{eq:superuniformity}), then the sequence $\left(Q_{\lambda,t}^{\left(r\right)}\right)_{t\in\mathbb{N}}$
defined in \emph{(\ref{eq:EWMA-no-P0-weights})--(\ref{eq:Q-def})}
is super-uniform under $\mathrm{P}_{0}$, i.e. for each $\alpha\in\left[0,1\right]$,
\[
\mathrm{P}_{0}\left(Q_{\lambda,t}^{\left(r\right)}\le\alpha\right)\le\alpha\quad\mathrm{for\ all\ }t\in\mathbb{N}.
\]
\end{prop}

A notable inconvenience in applying Proposition \ref{prop:EWMA-like-charts}
when $r\ge1$ is that the constants $\left(a_{\lambda,t}^{\left(r\right)}\right)_{t}$
are not time-independent. At some loss of sharpness in the super-uniform
inequality, we can replace $a_{\lambda,t}^{\left(r\right)}$ by a
time-independent bound $a_{\lambda}^{\left(r\right)}$ that works
for all $t$. Since $w_{t,\max}\ge\lambda$ for every $t\ge1$, we
have $w_{t,\max}^{-1}\le\lambda^{-1}$ and therefore 
\[
\min\{1+r,w_{t,\max}^{-1}\}\le\min\{1+r,\lambda^{-1}\}.
\]
Define 
\begin{equation}
a_{\lambda}^{\left(r\right)}=\begin{cases}
\min\left\{ 1+r,\frac{1}{\lambda}\right\} ^{1/r} & \mathrm{if\ }r\ge1,\\[0.3em]
\left(1+r\right)^{1/r} & \mathrm{if\ }r\in\left(-1,1\right)\setminus\{0\},
\end{cases}\text{ and }\tilde{Q}_{\lambda,t}^{\left(r\right)}=a_{\lambda}^{\left(r\right)}\left[S_{\lambda,t}^{\left(r\right)}\right]^{1/r}.\label{eq:Q-tilde-def}
\end{equation}
Then $a_{\lambda}^{\left(r\right)}\ge a_{\lambda,t}^{\left(r\right)}$
for every $t$, so 
\[
\tilde{Q}_{\lambda,t}^{\left(r\right)}=\frac{a_{\lambda}^{\left(r\right)}}{a_{\lambda,t}^{\left(r\right)}}Q_{\lambda,t}^{\left(r\right)}\ge Q_{\lambda,t}^{\left(r\right)}.
\]
Hence $\left\{ \tilde{Q}_{\lambda,t}^{\left(r\right)}\le\alpha\right\} \subseteq\left\{ Q_{\lambda,t}^{\left(r\right)}\le\alpha\right\} $
for all $\alpha\in[0,1]$, and the super-uniformity of $Q_{\lambda,t}^{\left(r\right)}$
yields 
\[
\mathrm{P}_{0}\left(\tilde{Q}_{\lambda,t}^{\left(r\right)}\le\alpha\right)\le\mathrm{P}_{0}\left(Q_{\lambda,t}^{\left(r\right)}\le\alpha\right)\le\alpha,\quad\mathrm{for\ all\ }t\in\mathbb{N}.
\]

Another unfortunate inconvenience is that the EWMA-like $p$-values
$\left(Q_{\lambda,t}^{\left(r\right)}\right)_{t}$ and their modification
$\left(\tilde{Q}_{\lambda,t}^{\left(r\right)}\right)_{t}$ need not
maintain \emph{conditional} super-uniformity when the original sequence
$\left(P_{t}\right)_{t}$ satisfies the conditional super-uniformity
condition (\ref{eq:Conditional_Super_Uni}). Fortunately, conditional
super-uniformity is preserved for the modification 
\begin{equation}
\bar{Q}_{\lambda,t}^{\left(r\right)}=\lambda^{-1/r}\left[S_{\lambda,t}^{\left(r\right)}\right]^{1/r},\quad t\ge1,\label{eq:Q-bar-def}
\end{equation}
for each $\lambda\in\left(0,1\right)$, when $r\ge1$. In the special
case $r=1$ we have 
\[
\bar{Q}_{\lambda,1}^{\left(1\right)}=\lambda^{-1}P_{1}\quad\mathrm{and}\quad\bar{Q}_{\lambda,t}^{\left(1\right)}=P_{t}+\left(1-\lambda\right)\bar{Q}_{\lambda,t-1}^{\left(1\right)},\quad t\ge2,
\]
so that $\lambda\bar{Q}_{\lambda,t}^{\left(1\right)}$ is the usual
EWMA of the raw $p$-values $P_{t}$ with weight $\lambda$ on the
current observation. 
\begin{prop}
\label{prop:EWMA-like-charts-conditional} For each $\lambda\in\left(0,1\right)$
and $r\ge1$, suppose $\left(P_{t}\right)_{t}$ satisfies (\ref{eq:Conditional_Super_Uni})
for the filtration $\left({\cal F}_{t}\right)_{t}$. Define $S_{\lambda,t}^{\left(r\right)}$
as above and $\bar{Q}_{\lambda,t}^{\left(r\right)}$ by (\ref{eq:Q-bar-def}).
Then $\left(\bar{Q}_{\lambda,t}^{\left(r\right)}\right)_{t}$ is conditionally
super-uniform under $\mathrm{P}_{0}$, i.e., for each $\alpha\in\left[0,1\right]$,
\[
\mathrm{P}_{0}\left(\bar{Q}_{\lambda,t}^{\left(r\right)}\le\alpha\,\big|\,{\cal F}_{t-1}\right)\le\alpha\text{, }\mathrm{a.s.,\ for\ all\ }t\ge1.
\]
\end{prop}

\begin{rem}
\label{rem:EWMA-lambda-large-small} For $r\ge1$, when $\lambda\ge1/2$,
it holds that $1/\lambda\le2\le1+r$, so $a_{\lambda}^{\left(r\right)}=\lambda^{-1/r}$
and therefore $\tilde{Q}_{\lambda,t}^{\left(r\right)}=\bar{Q}_{\lambda,t}^{\left(r\right)}$
for all $t\ge1$. Moreover, for $t\ge2$ we have $(1-\lambda)^{t-1}\le1-\lambda\le\lambda$,
so $w_{t,\max}=\lambda$, hence $a_{\lambda,t}^{\left(r\right)}=\lambda^{-1/r}$
and 
\[
Q_{\lambda,t}^{\left(r\right)}=\tilde{Q}_{\lambda,t}^{\left(r\right)}=\bar{Q}_{\lambda,t}^{\left(r\right)},\text{ for }t\ge2.
\]
At time $t=1$ we have $w_{1,\max}=1$, so $a_{\lambda,1}^{\left(r\right)}=1$
and $Q_{\lambda,1}^{\left(r\right)}=P_{1}$ while $\tilde{Q}_{\lambda,1}^{\left(r\right)}=\bar{Q}_{\lambda,1}^{\left(r\right)}=\lambda^{-1/r}P_{1}$.
Thus, for $\lambda\ge1/2$, the only discrepancy between $Q_{\lambda,t}^{\left(r\right)}$
and the conditionally valid $\bar{Q}_{\lambda,t}^{\left(r\right)}$
occurs at the very first time point.

When $\lambda<1/2$, the quantities $Q_{\lambda,t}^{\left(r\right)}$
and $\tilde{Q}_{\lambda,t}^{\left(r\right)}$ can differ for early
time points, because the maximal EWMA weight $w_{t,\max}=\max\{\lambda,(1-\lambda)^{t-1}\}$
may be larger than $\lambda$. However, since $(1-\lambda)^{t-1}$
decays geometrically in $t$, there exists a finite index $t_{\max}(\lambda,r)$
such that $Q_{\lambda,t}^{\left(r\right)}=\tilde{Q}_{\lambda,t}^{\left(r\right)}$
for all $t>t_{\max}(\lambda,r)$. Thus, for large $t$ and $\lambda<1/2$
it suffices to work with the simpler, time-homogeneous EWMA-like chart
$\left(\tilde{Q}_{\lambda,t}^{\left(r\right)}\right)_{t}$. 
\end{rem}

An alternative construction of EWMA-like charts for $p$-value statistics
is via the method of $e$-values. We provide such a derivation in
the Appendix for completeness but note that the constructed objects
did not perform as well in comparison to the averaging method above
in external numerical assessments. 

\subsubsection{The distribution of the EWMA process for independent uniform random
variables}

\label{subsec:Uniform_EWMA_Theorems}

We conclude this section with an extension of the exploration of the
EWMA process for IID uniform random variables, considered by \citet{yeh2003multivariate}.
In particular, \citet{yeh2003multivariate} consider a sequence of
IID test statistics $\left(Z_{t}\right)_{t}$, whose distribution
under $\mathrm{H}_{0}$ is equal to that of $Z$. Using the probability
integral transformation (i.e., the cumulative distribution function
(CDF) of $Z$, $\mathrm{F}_{Z}$), the authors generate the sequence
$\left(U_{t}\right)_{t}$, where, for each $t$, $U_{t}=\mathrm{F}_{Z}\left(Z_{t}\right)$
has distribution $\mathrm{Unif}\left(0,1\right)$ under $\mathrm{H}_{0}$
(when $\mathrm{F}_{Z}$ is continuous). The authors then consider
the transformation $V_{t}=U_{t}-1/2$, and construct the corresponding
EWMA process $\left(\tilde{V}_{\lambda,t}\right)_{t}$, defined by
\[
\tilde{V}_{\lambda,t}=\lambda V_{t}+\left(1-\lambda\right)\tilde{V}_{\lambda,t-1}\text{,}
\]
for $\lambda\in\left(0,1\right)$ and $t\ge1$, with $\tilde{V}_{\lambda,0}=0$.
They proceed to produce mean and variance formulas for the process
and obtain ARL estimates for various control rules via Monte Carlo
simulation.

We follow on from their work by considering IID sequences $\left(U_{t}\right)_{t}$
with $U_{t}\sim\mathrm{Unif}\left(0,1\right)$ (we may regard $P_{t}=U_{t}$
as a $p$-value under $\mathrm{H}_{0}$), and construct the EWMA process
$\left(\tilde{U}_{\lambda,t}\right)_{t}$, defined by 
\[
\tilde{U}_{\lambda,t}=\lambda U_{t}+\left(1-\lambda\right)\tilde{U}_{\lambda,t-1}\text{,}
\]
for $\lambda\in\left(0,1\right)$ and $t\ge1$, with $\tilde{U}_{\lambda,0}=u_{0}\in\left[0,1\right]$.
Since $V_{t}=U_{t}-1/2$ and the EWMA recursion is linear, the two
EWMAs are related by 
\[
\tilde{U}_{\lambda,t}=\tilde{V}_{\lambda,t}+\frac{1}{2}+\left(1-\lambda\right)^{t}\Big(u_{0}-\frac{1}{2}\Big)\text{, for }t\ge1,
\]
where $\left(\tilde{V}_{\lambda,t}\right)_{t}$ denotes the EWMA of
$V_{t}$ with $\tilde{V}_{\lambda,0}=0$. In particular, when $u_{0}=1/2$,
\[
\tilde{U}_{\lambda,t}=\tilde{V}_{\lambda,t}+\frac{1}{2}\text{, for all }t\ge1,
\]
so $\left(\tilde{U}_{\lambda,t}\right)_{t}$ is a mean-shift of $\left(\tilde{V}_{\lambda,t}\right)_{t}$
and thus characterises the same underlying stochastic object, up to
an additive constant. We now present an expression for the probability
density function (PDF) of $\tilde{U}_{\lambda,t}$, for each $t\ge1$. 
\begin{prop}
\label{prop:PDF_of_Ut}Let $a_{t,s}=\lambda\left(1-\lambda\right)^{t-s}$
and $a_{\mathbb{S}}=\sum_{s\in\mathbb{S}}a_{t,s}$, for each $s\in\left[t\right]$
and $\mathbb{S}\subseteq\left[t\right]$. Then, for each $u_{0}\in\left[0,1\right]$
and $t\ge1$, the PDF and CDF of $\tilde{U}_{\lambda,t}$ have the
respective forms: 
\[
u\mapsto f_{t,u_{0}}\left(u\right)=\frac{1}{\left(t-1\right)!\prod_{s=1}^{t}a_{t,s}}\sum_{\mathbb{S}\subseteq\left[t\right]}\left(-1\right)^{\left|\mathbb{S}\right|}\left[u-\left(1-\lambda\right)^{t}u_{0}-a_{\mathbb{S}}\right]_{+}^{t-1}\text{,}
\]
\[
u\mapsto\mathrm{F}_{t,u_{0}}\left(u\right)=\frac{1}{t!\prod_{s=1}^{t}a_{t,s}}\sum_{\mathbb{S}\subseteq\left[t\right]}\left(-1\right)^{\left|\mathbb{S}\right|}\left[u-\left(1-\lambda\right)^{t}u_{0}-a_{\mathbb{S}}\right]_{+}^{t}\text{.}
\]
\end{prop}

From the linearity of expectation applied to $\tilde{U}_{\lambda,t}$,
and the fact that $\mathrm{E}_{0}U_{t}=1/2$ and $\mathrm{var}_{0}\left(U_{t}\right)=1/12$
(where $\mathrm{var}_{0}$ is the variance under $\mathrm{H}_{0}$),
we have 
\[
\mathrm{E}_{0}\tilde{U}_{\lambda,t}=\left(1-\lambda\right)^{t}u_{0}+\sum_{s=1}^{t}a_{t,s}\,\mathrm{E}_{0}U_{s}=\frac{1}{2}+\left(1-\lambda\right)^{t}\left(u_{0}-\frac{1}{2}\right)\text{,}
\]
and, by independence of $\left(U_{t}\right)_{t}$, 
\[
\mathrm{var}_{0}\left(\tilde{U}_{\lambda,t}\right)=\sum_{s=1}^{t}a_{t,s}^{2}\,\mathrm{var}_{0}\left(U_{s}\right)=\frac{1}{12}\frac{\lambda\left[1-\left(1-\lambda\right)^{2t}\right]}{2-\lambda}\text{.}
\]
These expressions match those obtained by \citet{yeh2003multivariate}.

As we stated at the start of the section, the EWMA $p$-values $\left(\tilde{P}_{\lambda,t}\right)_{t}$
need not be super-uniform under $\mathrm{H}_{0}$, in general. We
validate this claim by showing (i) if $u_{0}<1$, then $\tilde{U}_{\lambda,t}$
is not super-uniform for every $t\ge1$ and $\lambda\in(0,1)$; and
(ii) when $u_{0}=1$, the only globally super-uniform case is $t=1$,
whereas for $t\ge2$ super-uniformity fails.

For (i), note that since $a_{t,s}U_{s}:\Omega\to[0,a_{t,s}]$ for
each $s$, 
\[
\sum_{s=1}^{t}a_{t,s}U_{s}:\Omega\to\Bigl[0,\sum_{s=1}^{t}a_{t,s}\Bigr]=\bigl[0,\,1-(1-\lambda)^{t}\bigr].
\]
Because $\tilde{U}_{\lambda,t}=(1-\lambda)^{t}u_{0}+\sum_{s=1}^{t}a_{t,s}U_{s}$,
\[
\tilde{U}_{\lambda,t}:\Omega\to\bigl[(1-\lambda)^{t}u_{0},1-(1-\lambda)^{t}(1-u_{0})\bigr].
\]
Hence, if $\alpha^{*}$ satisfies $1-(1-\lambda)^{t}(1-u_{0})<\alpha^{*}<1$,
then $\mathrm{P}_{0}\left(\tilde{U}_{\lambda,t}\le\alpha^{*}\right)=1>\alpha^{*}$,
so $\tilde{U}_{\lambda,t}$ is not super-uniform with respect to $\mathrm{P}_{0}$.
In particular, if $u_{0}<1$, such an $\alpha^{*}$ always exists.

For (ii), take $u_{0}=1$, so $\mathrm{supp}(\tilde{U}_{\lambda,t})=[(1-\lambda)^{t},1]$.
When $t=1$, $\tilde{U}_{\lambda,1}\sim\mathrm{Unif}\left(1-\lambda,1\right)$
and is globally super-uniform. For $t\ge2$, define the normalised
convex combination 
\[
\bar{U}_{\lambda,t}=\frac{\tilde{U}_{\lambda,t}-(1-\lambda)^{t}}{1-(1-\lambda)^{t}}=\sum_{s=1}^{t}w_{s}U_{s}\text{, }w_{s}=\frac{a_{t,s}}{1-(1-\lambda)^{t}}\text{, and }\sum_{s=1}^{t}w_{s}=1.
\]
Let $\mathrm{G}_{\lambda,t}(\bar{u})=\mathrm{P}_{0}\left(\bar{U}_{\lambda,t}\le\bar{u}\right)$
for $\bar{u}\in[0,1]$. Then 
\[
\mathrm{F}_{t,1}(\alpha)=\mathrm{G}_{\lambda,t}\left(\frac{\alpha-(1-\lambda)^{t}}{1-(1-\lambda)^{t}}\right).
\]
For $t\ge2$, the two largest weights are 
\[
w_{t}=\frac{\lambda}{1-(1-\lambda)^{t}}\text{ and }w_{t-1}=\frac{\lambda(1-\lambda)}{1-(1-\lambda)^{t}}.
\]
Set 
\[
c=\frac{1}{w_{t}w_{t-1}}=\dfrac{\bigl(1-(1-\lambda)^{t}\bigr)^{2}}{\lambda^{2}(1-\lambda)}\text{.}
\]
If $\epsilon\le\min\{w_{t},w_{t-1}\}$, then the event $\{\bar{U}_{\lambda,t}>1-\epsilon\}$
implies both $\{w_{t}(1-U_{t})<\epsilon\}$ and $\{w_{t-1}(1-U_{t-1})<\epsilon\}$.
By independence and the $\mathrm{Unif}\left(0,1\right)$ CDF, 
\[
\mathrm{P}_{0}\left(\bar{U}_{\lambda,t}>1-\epsilon\right)\le\frac{\epsilon}{w_{t}}\cdot\frac{\epsilon}{w_{t-1}}=c\epsilon^{2},
\]
so 
\[
\mathrm{G}_{\lambda,t}(1-\epsilon)=1-\mathrm{P}_{0}\left(\bar{U}_{\lambda,t}>1-\epsilon\right)\ge1-c\epsilon^{2}.
\]

Choose 
\[
\epsilon^{*}=(1-\eta)\min\left\{ \frac{w_{t}}{2},\frac{w_{t-1}}{2},\frac{1-(1-\lambda)^{t}}{2c}\right\} \text{, for some }\eta\in(0,1).
\]
Then $\epsilon^{*}\le\min\{w_{t},w_{t-1}\}$ and $c\epsilon^{*2}<\left\{ 1-(1-\lambda)^{t}\right\} \epsilon^{*}/2$,
yielding the chain of inequalities: 
\[
\mathrm{G}_{\lambda,t}(1-\epsilon^{*})\ge1-c\epsilon^{*2}>1-\frac{1-(1-\lambda)^{t}}{2}\epsilon^{*}>1-\bigl(1-(1-\lambda)^{t}\bigr)\epsilon^{*}.
\]
Setting $\alpha^{*}=1-\bigl(1-(1-\lambda)^{t}\bigr)\epsilon^{*}$,
we conclude 
\[
\mathrm{F}_{t,1}(\alpha^{*})=\mathrm{G}_{\lambda,t}(1-\epsilon^{*})>\alpha^{*},
\]
which violates the super-uniformity inequality $\mathrm{P}_{0}\left(\tilde{U}_{\lambda,t}\le\alpha\right)\le\alpha$
for all $\alpha\in[0,1]$. Hence, for $u_{0}=1$ and $t\ge2$, $\tilde{U}_{\lambda,t}$
is not super-uniform.

We conclude this section with a positive outcome: although the process
$\left(\tilde{U}_{\lambda,t}\right)_{t}$ is not globally super-uniform
under $\mathrm{H}_{0}$, it is left-tail super-uniform whenever $u_{0}\ge1/2$,
i.e., 
\[
\mathrm{P}_{0}\bigl(\tilde{U}_{\lambda,t}\le\alpha\bigr)\le\alpha\text{ for every }\alpha\in[0,1/2].
\]
This enables us to still use $\left(\tilde{U}_{\lambda,t}\right)_{t}$
with the alarm rule $\tilde{U}_{\lambda,t}\le\alpha$ and still enjoy
the guarantee of Proposition \ref{prop:ARLnoconditions} (whose conclusion
is valid for any fixed $\alpha\in[0,1]$ for which $\mathrm{P}_{0}\bigl(\tilde{U}_{\lambda,t}\le\alpha\bigr)\le\alpha$
for all $t\ge1$). 
\begin{prop}
\label{prop:Ut_left_super_unif} For each $\alpha\in[0,1/2]$, $u_{0}\ge1/2$,
$\lambda\in(0,1)$, and $t\ge1$, 
\[
\mathrm{P}_{0}\bigl(\tilde{U}_{\lambda,t}\le\alpha\bigr)\le\alpha\text{.}
\]
\end{prop}

\section{Directional and coordinate localisation via $p$-value charts}

We now specialise the framework to a setting where the monitored process
is indexed by a $d$-dimensional parameter $\theta=(\theta_{1},\dots,\theta_{d})$
with a fixed IC baseline $\theta_{0}=(\theta_{0,1},\dots,\theta_{0,d})\in\mathbb{R}^{d}$.
In a one-phase design, we may consider the null measure under $\mathrm{H}_{0}$,
$\mathrm{P}_{X,0}=\mathrm{P}_{X,\theta_{0}}$, and test against the
alternative that $\mathrm{P}_{X,t}=\mathrm{P}_{X,\theta}$ for some
$\theta\ne\theta_{0}$. In a two-phase design, for $X\in\mathbb{R}^{d}$
one may assess the null $\mathrm{H}_{0}$ that $\mathrm{P}_{X,t}=\mathrm{P}_{X,0}$
by testing, for example, the equality in mean $\mathrm{E}X_{t}=\mathrm{E}X_{0}$,
where $\Delta=\mathrm{E}X_{t}-\mathrm{E}X_{0}$ satisfies $\Delta=0$
under $\mathrm{H}_{0}$. 

At the point of alarm, our inferential objective is twofold: firstly,
to identify affected coordinates $j\in[d]$, for which $\theta_{j}\ne\theta_{0,j}$.
Secondly, for each affected coordinate $j$, to determine the direction
of change, i.e., whether $\theta_{j}<\theta_{0,j}$ or $\theta_{j}>\theta_{0,j}$.
We aim to make these conclusions in a principled manner with simultaneous
control, in the sense that the probability that at least one of our
dimensional or directional conclusions is incorrect is bounded above
by a pre-specified level $\alpha$.

To this end, for each coordinate $j\in[d]$ and direction $\square\in\{\le,\ge\}$,
suppose we have a sequence of directional $p$-values $\bigl(P_{t}^{(j,\square)}\bigr)_{t\in\mathbb{N}}$
that test the corresponding one-sided null $\mathrm{H}_{0}^{(j,\square)}$
(e.g., $\theta_{0,j}\ge\theta_{j}$ when $\square=\ge$, or $\theta_{0,j}\le\theta_{j}$
when $\square=\le$). We assume that each sequence is valid in the
super-uniform sense used in Section \ref{sec:Main-results}; i.e.,
for every $t\in\mathbb{N}$ and $\alpha\in[0,1]$, 
\[
\mathrm{P}_{0,t}^{(j,\square)}\bigl(P_{t}^{(j,\square)}\le\alpha\bigr)\le\alpha,
\]
whenever $\mathrm{P}_{0,t}^{(j,\square)}$ is a probability measure
on $(\Omega,\mathfrak{A})$ such that the push-forward measures $\mathrm{P}_{X,s}$
satisfy $\mathrm{H}_{0}^{(j,\square)}$ for all $s\in[t]$. In particular,
in a one-phase design the $p$-values may be of the form $P_{t}^{(j,\square)}=p_{t}^{(j,\square)}(\mathbf{X}_{1},\dots,\mathbf{X}_{t})$
(or $p_{t}^{(j,\square)}(\mathbf{X}_{t})$), while in a two-phase
design they may be of the form $P_{t}^{(j,\square)}=p_{t}^{(j,\square)}(\mathbf{X}_{0};\mathbf{X}_{t})$.
When convenient, we may further assume the conditional super-uniformity
property (as in (\ref{eq:Conditional_Super_Uni})) with respect to
a filtration $({\cal F}_{t})_{t\in\mathbb{N}}$.

Given these $2d$ sequences of directional $p$-values, we construct
an overall sequence of $p$-values $\bigl(\bar{P}_{t}\bigr)_{t\in\mathbb{N}}$
satisfying the super-uniformity condition (\ref{eq:superuniformity}),
and whose induced rejections satisfy the closed-testing principle
of \citet{MarcusPeritzGabriel1976}, based on the family of tests
\[
\mathbb{H}_{0}=\bigl\{\mathrm{H}_{0}^{(j,\square)}:j\in[d],\square\in\{\le,\ge\}\bigr\}.
\]
Upon an alarm raised by the rule $\bar{P}_{t}\le\alpha$, the procedure
provides a set of rejections $\mathbb{D}_{t}\subset[d]\times\{\le,\ge\}$
such that the family-wise error rate satisfies 
\begin{equation}
\mathrm{P}_{0}\bigl(\exists(j,\square)\in\mathbb{D}_{t}:\ \theta_{0,j}\ \square\ \theta_{j}\bigr)\le\alpha.\label{eq:FWER_control-1}
\end{equation}

\subsection{$p$-value chart construction }

\label{subsec:Directional chart}

We take as input the $2d$ sequences of $p$-values $\bigl(P_{t}^{(j,\square)}\bigr)_{t\in\mathbb{N}}$
and choose our level $\alpha\in(0,1)$ to determine our acceptable
ARL. Then, at each time $t$, we proceed via the following steps: 
\begin{enumerate}
\item For each $j\in[d]$, compute $P_{t}^{j}=\min\{1,\,2\min\{P_{t}^{(j,\le)},P_{t}^{(j,\ge)}\}\}$. 
\item Compute the aggregate $p$-value via the expression $\bar{P}_{t}=\min\{1,\,d\min_{j\in[d]}P_{t}^{j}\}$
or $\bar{P}_{t}=\min\{1,\,\min\{2,d\}\,d^{-1}\sum_{j=1}^{d}P_{t}^{j}\}$. 
\end{enumerate}
We then reject the global null hypothesis $\mathrm{H}_{0}:\theta=\theta_{0}$
in favour of the alternative that $\theta\ne\theta_{0}$ and raise
an alarm at time $t$ if $\bar{P}_{t}\le\alpha$. Note that the two
choices for $\bar{P}_{t}$ correspond to the Bonferroni aggregator
and the $r=1$ case of Lemma \ref{lem:p-value=00003D00003D00003D00003D000020averaging=00003D00003D00003D00003D000020Lemma},
with uniform weights to ensure that it satisfies the super-uniformity
condition (\ref{eq:superuniformity}). One can instead compute the
aggregate $p$-value via any other merging function that satisfies
Lemma \ref{lem:p-value=00003D00003D00003D00003D000020averaging=00003D00003D00003D00003D000020Lemma}.

At the time point where we raise an alarm, we apply Holm's procedure
(\citealp{Holm1979}) on the set $\{P_{t}^{j}\}_{j\in[d]}$. For each
coordinate $j$ that is rejected by Holm's procedure, we then place
the pair $(j,\le)$ into $\mathbb{D}_{t}$ if $P_{t}^{(j,\le)}\le P_{t}^{(j,\ge)}$,
else we place the pair $(j,\ge)$ into $\mathbb{D}_{t}$. The set
of paired coordinates and directions $\mathbb{D}_{t}$ then satisfies
(\ref{eq:FWER_control-1}). We can summarise this process via the
following steps: 
\begin{enumerate}
\item Order $\{P_{t}^{j}\}_{j\in[d]}$ to obtain the order statistics $P_{t}^{(1)}\le P_{t}^{(2)}\le\dots\le P_{t}^{(d)}$,
and note the corresponding indices of the ordering: $j_{(1)},\dots,j_{(d)}$. 
\item Let $k^{*}$ be the largest $k\in[d]$ such that $P_{t}^{(i)}\le\alpha/(d-i+1)$
for all $i\in[k]$. Put $j_{(1)},\dots,j_{(k^{*})}$ into $\mathbb{J}_{t}$. 
\item For each $j\in\mathbb{J}_{t}$, if $P_{t}^{(j,\le)}\le P_{t}^{(j,\ge)}$,
then put $(j,\le)$ into $\mathbb{D}_{t}$; else, put $(j,\ge)$ into
$\mathbb{D}_{t}$. 
\end{enumerate}

\subsubsection{Family-wise error rate control}

To verify that (\ref{eq:FWER_control-1}) holds, we first recall the
closed testing principle of \citet{MarcusPeritzGabriel1976} (see
also \citealp[Thm.~3.4]{Dickhaus2014SSI}) and the equivalence of
Holm's step-down procedure to the Bonferroni closure (cf. \citealp{Holm1979}). 
\begin{thm}
\label{thm:closure} Let $\{\,\mathbb{H}_{j}:j\in[m]\,\}$ be a finite
family of hypotheses. For every nonempty $\mathbb{J}\subseteq[m]$,
write $\mathbb{H}_{\mathbb{J}}=\bigcap_{j\in\mathbb{J}}\mathbb{H}_{j}$
and let $\mathcal{M}(\mathbb{H}_{\mathbb{J}})$ denote the set of
all data--generating distributions under $\mathbb{H}_{\mathbb{J}}$.
Suppose that for each such $\mathbb{J}$ we are given a local test
$T_{\mathbb{J}}$ of $\mathbb{H}_{\mathbb{J}}$ with level $\alpha$
in the sense that 
\[
\sup_{\mathrm{Q}\in\mathcal{M}(\mathbb{H}_{\mathbb{J}})}\mathrm{Q}\big(T_{\mathbb{J}}=1\big)\ \le\ \alpha.
\]
Define the closed \emph{testing} procedure that rejects an elementary
hypothesis $\mathbb{H}_{j}$ if and only if $T_{\mathbb{J}}$ rejects
$\mathbb{H}_{\mathbb{J}}$ for every $\mathbb{J}$ such that $j\in\mathbb{J}$.
Then the procedure has strong FWER control at level $\alpha$, i.e.
\[
\mathrm{Q}\big(\exists j\in[m]\ \text{with }\mathbb{H}_{j}\ \text{true and rejected}\big)\le\alpha,
\]
for any data--generating distribution $\mathrm{Q}$. 
\end{thm}

\begin{thm}
\label{thm:Holm's}Let $\{P_{j}\}_{j\in[m]}$ be super-uniform $p$-values
for the hypotheses $\{\mathbb{H}_{j}\}_{j\in[m]}$. For each nonempty
$\mathbb{J}\subset[m]$, define the Bonferroni local test: 
\[
T_{\mathbb{J}}^{\mathrm{Bon}}=\mathbf{1}_{\{\min_{j\in\mathbb{J}}P_{j}\le\alpha/|\mathbb{J}|\}}.
\]
Then the closed testing procedure generated by $\{T_{\mathbb{J}}^{\mathrm{Bon}}\}_{\mathbb{J}\subset[m]}$
is exactly Holm's step-down procedure: i.e., if $P_{(1)}\le\dots\le P_{(m)}$
are order statistics of $\{P_{j}\}_{j\in[m]}$, with indices $j_{(1)},\dots,j_{(m)}$,
then reject $\mathbb{H}_{j_{(1)}},\dots,\mathbb{H}_{j_{(k^{*})}}$,
where $k^{*}$ is the largest $k\in[m]$ such that $P_{(i)}\le\alpha/(m-i+1)$
for all $i\in[k]$. 
\end{thm}

We now proceed to prove that (\ref{eq:FWER_control-1}) holds. For
each fixed $t\in\mathbb{N}$, coordinate $j\in[d]$ and direction
$\square\in\{\le,\ge\}$, we recall that the $p$-value $P_{t}^{(j,\square)}$
is super-uniform under the null $\mathrm{H}_{0}^{(j,\square)}$. Let
$\mathrm{P}_{0,t}^{j}$ be a probability measure on $(\Omega,\mathfrak{A})$
such that the push-forward measures $\mathrm{P}_{X,s}$ satisfy $\mathrm{H}_{0}^{j}:\ \theta_{j}=\theta_{0,j}$
for all $s\in[t]$. Then $\mathrm{P}_{0,t}^{j}(P_{t}^{j}\le\alpha)\le\alpha$,
for every $\alpha\in[0,1]$. To see this, observe that 
\begin{equation}
\mathrm{P}_{0,t}^{j}\bigl(2\min\{P_{t}^{(j,\le)},P_{t}^{(j,\ge)}\}\le\alpha\bigr)\ \le\ \mathrm{P}_{0,t}^{j}\bigl(P_{t}^{(j,\le)}\le\alpha/2\bigr)+\mathrm{P}_{0,t}^{j}\bigl(P_{t}^{(j,\ge)}\le\alpha/2\bigr)\ \le\ \alpha,\label{eq:Union_bound}
\end{equation}
via the union bound and super-uniformity of the $p$-values. The validity
of our directional procedure is then stated as follows. 
\begin{prop}
\label{prop:Directional_correct_FWER}The procedure from Section \ref{subsec:Directional chart}
satisfies (\ref{eq:FWER_control-1}) at each time $t\in\mathbb{N}$. 
\end{prop}

While preparing this manuscript, we became aware that the procedure
in Section~\ref{subsec:Directional chart} is a special case of the
simultaneous directional inference framework of \citet{HellerSolari2024SDI}.
This observation suggests that broader classes of localisation-type
charts, beyond our proposal, can be developed along the same lines. 

\section{Example applications}

\label{sec:Example-applications}

We demonstrate the correctness of our theory and the usefulness of
our results through the following examples and simulations. All \textsf{R}
scripts for these simulations are available at \url{https://github.com/hiendn/Pvalue_SPC}.

\subsection{Elementary examples}

\label{subsec:Elementary-examples}

We start by considering the one-phase design where we operate under
the global assumption that $X$ is normally distributed with law $\mathrm{N}\left(\mu,\sigma^{2}\right)$,
for some $\mu\in\mathbb{R}$ and $\sigma^{2}>0$. If we consider a
one-phase design where we take $\mathrm{H}_{0}$ to be the hypothesis
that $X\sim\mathrm{N}\left(0,1\right)$ and observe a single observation
$X_{t}$ at each time $t$, where $\left(X_{t}\right)_{t}$ is an
IID sequence, then the natural test statistic is $Z_{t}=X_{t}$. By
the probability integral transform and our observation about the tightness
of the bound $\mathrm{E}_{0}R\ge1/\alpha$, it follows that we can
exactly control the $k$-ARL for this procedure at level $k/\alpha$
by raising an alarm whenever $\left|Z_{t}\right|\ge\zeta_{1-\alpha/2}$,
where $\zeta_{u}$ is the $u$th quantile of the standard normal distribution.
However, in this scenario, the $p$-value for $\mathrm{H}_{0}$ is
exactly $P_{t}=2\left\{ 1-\Phi\left(\left|X_{t}\right|\right)\right\} $
(satisfying (\ref{eq:Conditional_Super_Uni})), where $\Phi$ is the
standard normal CDF, which yields exactly the same decision rule if
we raise an alarm whenever $P_{t}\le\alpha$, and thus in this elementary
example, the chart based on either the sequence $\left(P_{t}\right)_{t}$
or $\left(Z_{t}\right)_{t}$ will perform exactly the same.

We note that we can similarly propose an elementary two-phase design,
whereupon in Phase I, we observe a single observation $X_{0}\sim\mathrm{N}\left(\mu_{0},1\right)$
and in Phase II we observe $X_{t}\sim\mathrm{N}\left(\mu_{1},1\right)$
at each time $t$, where $\left(X_{t}\right)_{t}$ is an IID sequence.
To test the hypothesis $\mathrm{H}_{0}:\mu_{0}=\mu_{1}$ at each time
$t$, we can use the test statistic $Z_{t}=\left(X_{t}-X_{0}\right)/\sqrt{2}$
and raise an alarm whenever $\left|Z_{t}\right|\ge\zeta_{1-\alpha/2}$.
Equivalently, we can compute the $p$-value $P_{t}=2\left\{ 1-\Phi\left(\left|Z_{t}\right|\right)\right\} $
and raise an alarm under the usual rule that $P_{t}\le\alpha$. In
the one-phase IID case above, the $k$-ARL is controlled at exactly
$k/\alpha$. However, in the present two-phase baseline-with-reuse
setting, the $p$-values are (unconditionally) super-uniform but $Z_{t}$
is not independent of ${\cal F}_{t-1}$, so (\ref{eq:Conditional_Super_Uni})
need not hold and exact equality at $k/\alpha$ need not occur. To
assess the tightness of the bounds from Propositions \ref{prop:ARLnoconditions}
and \ref{prop:kARL_superuniform} for $\left(P_{t}\right)_{t}$ in
this setting, we consider a short simulation using alarm rules with
$\alpha\in\left\{ 0.01,0.05\right\} $ assessing the $k$-ARL for
$k\in\left\{ 1,5\right\} $. We conduct $100$ repetitions for each
setting and report the results in Table \ref{tab:Elementary_normal_sim}.

\begin{table}
\caption{Average $R_{k}$ versus $k$-ARL bounds from Proposition \ref{prop:kARL_superuniform}
for two-phase normal example $p$-value charts, with 100 simulation
repetitions.}
\label{tab:Elementary_normal_sim}

\centering{}%
\begin{tabular}{cccccc}
\hline 
$\alpha$  & $k$  & mean $R_{k}$  & st. err. $R_{k}$  & lower bound  & ratio\tabularnewline
\hline 
0.01  & 1  & 1206.47  & 204.01  & 50.5  & 23.89\tabularnewline
0.01  & 5  & 5255.38  & 592.48  & 250.5  & 20.98\tabularnewline
0.05  & 1  & 74.4  & 12.29  & 10.5  & 7.09\tabularnewline
0.05  & 5  & 316.37  & 33.12  & 50.5  & 6.26\tabularnewline
\hline 
\end{tabular}
\end{table}

Let us consider a slightly less elementary two-phase design where
in Phase I, we observe $X_{0}\sim\mathrm{N}\left(0,1\right)$, and
in Phase II, we observe the autoregressive data $X_{t}=\delta+\beta X_{t-1}+\varepsilon_{t}$,
where $\varepsilon_{t}\sim\mathrm{N}\left(0,\sigma^{2}\right)$, $\left|\beta\right|<1$
and $\sigma^{2}=1-\beta^{2}$ (i.e., the process $\left(X_{t}\right)_{t}$
is stationary with marginal law $X_{t}\sim\mathrm{N}\left(\delta/\left(1-\beta\right),1\right)$).
Here, $\beta\ne0$ is unknown. In this situation $X_{t}$ is not independent
of the history ${\cal F}_{t-1}=\sigma\left(X_{0},\dots,X_{t-1}\right)$.
Naturally, to test the null $\mathrm{H}_{0}:\delta=0$, we can use
the marginal test statistics $Z_{t}=X_{t}$, which generate the $p$-values
$P_{t}^{\prime}=2\left\{ 1-\Phi\left(\left|X_{t}\right|\right)\right\} $,
which satisfy (\ref{eq:superuniformity}) but not (\ref{eq:Conditional_Super_Uni}),
and thus only the bounds of Propositions \ref{prop:ARLnoconditions}
and \ref{prop:kARL_superuniform} apply.

Alternatively, consider that when $\beta$ is known to equal $b\in\left(-1,1\right)$,
we can take $Z_{t}\left(b\right)=\left\{ X_{t}-bX_{t-1}\right\} /\sqrt{1-b^{2}}$,
and compute its $p$-value as $P_{t}\left(b\right)=2\left\{ 1-\Phi\left(\left|Z_{t}\left(b\right)\right|\right)\right\} $.
Indeed, if $\beta=b$, then under $\mathrm{H}_{0}$, conditional on
${\cal F}_{t-1}$, $X_{t}\sim\mathrm{N}\left(bX_{t-1},1-b^{2}\right)$.
However, since $\beta$ is not known, we can instead compute $P_{t}^{*}=\sup_{b\in\left(-1,1\right)}P_{t}\left(b\right)$,
and note that since $P_{t}^{*}\ge P_{t}\left(\beta\right)$ and $P_{t}\left(\beta\right)\,|\,{\cal F}_{t-1}\sim\mathrm{Unif}\left(0,1\right)$,
it holds that $\mathrm{P}_{0}\left(P_{t}^{*}\le\alpha\,\big|\,{\cal F}_{t-1}\right)\le\alpha$,
for each $\alpha\in\left[0,1\right]$, and thus satisfies (\ref{eq:Conditional_Super_Uni})
(cf. \citealp{BergerBoos1994}). In fact, we have the closed form
$P_{t}^{*}=2\left\{ 1-\Phi\left(\sqrt{\left[X_{t}^{2}-X_{t-1}^{2}\right]_{+}}\right)\right\} $,
since, over $\left(-1,1\right)$, the function $g\left(b\right)=\left(y-bx\right)^{2}/\left(1-b^{2}\right)$
achieves its minimum at $b=y/x$ when $\left|y\right|\le\left|x\right|$,
yielding value $0$, and at the unique stationary point $b=x/y$ otherwise,
yielding value $y^{2}-x^{2}$; thus $\inf_{b\in\left(-1,1\right)}\sqrt{g\left(b\right)}=\sqrt{\left[y^{2}-x^{2}\right]_{+}}$.
We conclude that the sequence $\left(P_{t}^{*}\right)_{t}$ satisfies
(\ref{eq:Conditional_Super_Uni}) but is not generated by a probability
integral transformation and thus the bounds of Propositions \ref{prop:ARL_with_conditional_indep}
and \ref{prop:Conditional_Super_Unif_kARL} apply but may not be tight.

To assess the $k$-ARL performance of the two charts $\left(P_{t}^{\prime}\right)_{t}$
and $\left(P_{t}^{*}\right)_{t}$ and the tightness of the theoretical
bounds, we simulate the $k$-ARL values of both charts for the scenario
where $\beta\in\left\{ 0.1,0.5\right\} $, under the null hypothesis
(i.e., $\delta=0$), for alarm rules based on $\alpha\in\left\{ 0.01,0.05\right\} $
and $k\in\left\{ 1,5\right\} $. We conduct $100$ repetitions of
each scenario and report the results in Table \ref{tab:AR1_Sim}. 

\begin{table}
\caption{Average $R_{k}$ versus $k$-ARL bounds from Propositions \ref{prop:kARL_superuniform}
and \ref{prop:Conditional_Super_Unif_kARL} for autoregressive example
$p$-value charts $\left(P_{t}^{\prime}\right)_{t}$ and $\left(P_{t}^{*}\right)_{t}$,
respectively, with 100 simulation repetitions.}\label{tab:AR1_Sim}

\centering{}%
\begin{tabular}{cccccccc}
\hline 
Chart & $\beta$ & $\alpha$ & $k$ & mean $R_{k}$ & st. err. $R_{k}$ & lower bound & ratio\tabularnewline
\hline 
$\left(P_{t}^{\prime}\right)_{t}$ & 0.1 & 0.01 & 1 & 84.93 & 8.17 & 50.5 & 1.68\tabularnewline
 & 0.1 & 0.01 & 5 & 453.21 & 21.20 & 250.5 & 1.81\tabularnewline
 & 0.1 & 0.05 & 1 & 22.6 & 2.03 & 10.5 & 2.15\tabularnewline
 & 0.1 & 0.05 & 5 & 93.61 & 3.63 & 50.5 & 1.85\tabularnewline
 & 0.5 & 0.01 & 1 & 119.79 & 13.65 & 50.5 & 2.37\tabularnewline
 & 0.5 & 0.01 & 5 & 502.01 & 24.31 & 250.5 & 2.00\tabularnewline
 & 0.5 & 0.05 & 1 & 22.54 & 2.06 & 10.5 & 2.15\tabularnewline
 & 0.5 & 0.05 & 5 & 112.04 & 5.37 & 50.5 & 2.22\tabularnewline
\hline 
$\left(P_{t}^{*}\right)_{t}$ & 0.1 & 0.01 & 1 & 123.77 & 13.23 & 100 & 1.24\tabularnewline
 & 0.1 & 0.01 & 5 & 718.19 & 37.09 & 500 & 1.44\tabularnewline
 & 0.1 & 0.05 & 1 & 30.26 & 2.71 & 20 & 1.51\tabularnewline
 & 0.1 & 0.05 & 5 & 135.94 & 5.20 & 100 & 1.36\tabularnewline
 & 0.5 & 0.01 & 1 & 279.61 & 29.07 & 100 & 2.80\tabularnewline
 & 0.5 & 0.01 & 5 & 1311.44 & 59.53 & 500 & 2.62\tabularnewline
 & 0.5 & 0.05 & 1 & 38.97 & 3.20 & 20 & 1.95\tabularnewline
 & 0.5 & 0.05 & 5 & 222.31 & 8.60 & 100 & 2.22\tabularnewline
\hline 
\end{tabular}
\end{table}

\subsection{Kolmogorov--Smirnov-based EWMA charts}

\label{subsec:Kolmogorov=00003D002013Smirnov-based-EWMA-ch}

We now consider a two-phase setting for EWMA charts, whereby in Phase
I one observes IID data $\mathbf{X}_{0}=\left(X_{0,i}\right)_{i=1}^{n_{0}}$
of sample size $n_{0}$, where each $X_{0,i}$ is a replicate of $X_{0}\colon\Omega\to\mathbb{R}$
with probability measure $\mathrm{P}_{X,0}$. Then, in Phase II, we
monitor data sets $\mathbf{X}_{t}=\left(X_{t,i}\right)_{i=1}^{n_{t}}$
of sizes $n_{t}\in\mathbb{N}$, corresponding to a VSS setting. Here,
the process is in control (IC) at time $t$ if the null hypothesis
that $\mathrm{P}_{X,s}=\mathrm{P}_{X,0}$ for all $s\in\left[t\right]$
holds, and out of control (OOC) otherwise. At time $t$, we test $\mathrm{H}_{0}:\mathrm{P}_{X,t}=\mathrm{P}_{X,0}$
using the classical two-sample Kolmogorov--Smirnov (KS) test statistic,
where 
\[
Z_{t}=\sup_{x\in\mathbb{R}}\left|\mathrm{F}_{0,n_{0}}\left(x\right)-\mathrm{F}_{t,n_{t}}\left(x\right)\right|,
\]
with $\mathrm{F}_{0,n_{0}}\left(x\right)=n_{0}^{-1}\sum_{i=1}^{n_{0}}\mathbf{1}_{\left(-\infty,x\right]}\left(X_{0,i}\right)$
and $\mathrm{F}_{t,n_{t}}\left(x\right)=n_{t}^{-1}\sum_{i=1}^{n_{t}}\mathbf{1}_{\left(-\infty,x\right]}\left(X_{t,i}\right)$
the empirical CDFs of $\mathbf{X}_{0}$ and $\mathbf{X}_{t}$, respectively.
The KS test statistic has a well-known finite-sample distribution
under the null hypothesis and thus we can compute a sequence of $p$-values
$\left(P_{t}\right)_{t}$ using standard implementations such as the
\texttt{ks.test}() function in \textsf{R} (see e.g., \citealp{Nikiforov1994AS288}
and cf. \citealp[Sec. 6.3]{GibbonsChakraborti2010}).

We note that the sequence $\left(P_{t}\right)_{t}$ satisfies the
super-uniformity condition (\ref{eq:superuniformity}). Thus, if we
raise an alarm whenever $P_{t}\le\alpha$, Proposition \ref{prop:kARL_superuniform}
implies that we should expect the $k$-ARL values to be bounded below
by $k/(2\alpha)+1/2$ for each $k\ge1$. To assess the tightness of
this bound, we simulate the following setting. We let $n_{0}\in\left\{ 20,50,100\right\} $,
$n_{t}=N_{t}\sim\mathrm{DiscUnif}\left[n_{0}-10,n_{0}+10\right]$,
$\alpha\in\left\{ 0.01,0.05\right\} $, $k\in\left\{ 1,5\right\} $
and take $X_{0}\sim\mathrm{N}\left(0,1\right)$. Here, we note that
the distribution of the KS test statistic under a continuous null
distribution is the same for all choices and thus any continuous distribution
on $X_{0}$ will yield the same simulation outcomes. The resulting
average times to first $k$ alarms $R_{k}$ from 100 replications
of each scenario are provided in Table \ref{tab:KS_no_EWMA_sim}. 

\begin{table}
\caption{Average $R_{k}$ versus $k$-ARL bounds from Proposition \ref{prop:kARL_superuniform}
for KS test statistic $p$-value charts $\left(P_{t}\right)_{t}$,
using 100 simulation repetitions.}
\label{tab:KS_no_EWMA_sim}

\centering{}%
\begin{tabular}{ccccccc}
\hline 
$n_{0}$  & $\alpha$  & $k$  & mean $R_{k}$  & st. err. $R_{k}$  & lower bound  & ratio\tabularnewline
\hline 
20  & 0.01  & 1  & 1318.30  & 307.79  & 50.5  & 26.10\tabularnewline
20  & 0.01  & 5  & 7822.24  & 1136.65  & 250.5  & 31.23\tabularnewline
20  & 0.05  & 1  & 68.38  & 13.19  & 10.5  & 6.51\tabularnewline
20  & 0.05  & 5  & 410.77  & 54.22  & 50.5  & 8.13\tabularnewline
50  & 0.01  & 1  & 1465.94  & 288.67  & 50.5  & 29.03\tabularnewline
50  & 0.01  & 5  & 5672.49  & 818.32  & 250.5  & 22.64\tabularnewline
50  & 0.05  & 1  & 85.55  & 15.01  & 10.5  & 8.15\tabularnewline
50  & 0.05  & 5  & 397.58  & 41.43  & 50.5  & 7.87\tabularnewline
100  & 0.01  & 1  & 1663.56  & 316.09  & 50.5  & 32.94\tabularnewline
100  & 0.01  & 5  & 4005.76  & 610.87  & 250.5  & 15.99\tabularnewline
100  & 0.05  & 1  & 68.63  & 10.41  & 10.5  & 6.54\tabularnewline
100  & 0.05  & 5  & 394.16  & 41.78  & 50.5  & 7.81\tabularnewline
\hline 
\end{tabular}
\end{table}

Next, to assess the correctness of Propositions \ref{prop:EWMA-like-charts}
and \ref{prop:EWMA-like-charts-conditional} we consider the application
of the charts $\left(\tilde{Q}_{\lambda,t}^{\left(r\right)}\right)_{t}$
and $\left(\bar{Q}_{\lambda,t}^{\left(r\right)}\right)_{t}$ from
Section \ref{subsec:EWMA-like-charts-p-avg}, using the base sequence
$\left(P_{t}\right)_{t}$. We take $n_{0}\in\left\{ 50,100,200\right\} $
and $n_{t}$ as above, and we assess the cases when $\alpha\in\left\{ 0.05,0.1\right\} $,
$k\in\left\{ 1,5\right\} $, and $X_{0}\sim\mathrm{N}\left(0,1\right)$.
For $\left(\tilde{Q}_{\lambda,t}^{\left(r\right)}\right)_{t}$ we
choose $\lambda=1/2$ and $r\in\left\{ -0.9,-0.8\right\} $, while
for $\left(\bar{Q}_{\lambda,t}^{\left(r\right)}\right)_{t}$ we consider
$\lambda\in\left\{ 0.8,0.9\right\} $ and $r=1$. The average $R_{k}$
values from 100 replications of each setting are provided in Tables
\ref{tab:KS_IC_EWMA_sim} and \ref{tab:KS_IC_EWMA_sim_part2}. 

\begin{table}
\caption{Average $R_{k}$ while IC versus $k$-ARL bounds from Proposition
\ref{prop:kARL_superuniform} for KS EWMA-like $p$-value charts $\left(\tilde{Q}_{\lambda,t}^{\left(r\right)}\right)_{t}$
and $\left(\bar{Q}_{\lambda,t}^{\left(r\right)}\right)_{t}$, using
100 simulation repetitions. This table contains results for $n_{0}=50$,
while Appendix Table \ref{tab:KS_IC_EWMA_sim_part2} contains results
for $n_{0}\in\left\{ 100,200\right\} $.}
\label{tab:KS_IC_EWMA_sim}

\centering{}%
\begin{tabular}{cccccccccc}
\hline 
Chart  & $n_{0}$  & $\alpha$  & $k$  & $\lambda$  & $r$  & mean $R_{k}$  & st. err. $R_{k}$  & lower bound  & ratio\tabularnewline
\hline 
$\tilde{Q}_{\lambda,t}^{\left(r\right)}$  & 50  & 0.05  & 1  & 0.5  & -0.9  & 32846.69  & 7587.56  & 10.5  & 3128.2562\tabularnewline
 & 50  & 0.05  & 1  & 0.5  & -0.8  & 12091.79  & 3184.22  & 10.5  & 1151.5990\tabularnewline
 & 50  & 0.05  & 5  & 0.5  & -0.9  & 75643.87  & 14807.62  & 50.5  & 1497.8984\tabularnewline
 & 50  & 0.05  & 5  & 0.5  & -0.8  & 36112.13  & 5902.91  & 50.5  & 715.0917\tabularnewline
 & 50  & 0.1  & 1  & 0.5  & -0.9  & 7334.45  & 1907.45  & 5.5  & 1333.5364\tabularnewline
 & 50  & 0.1  & 1  & 0.5  & -0.8  & 2566.96  & 588.93  & 5.5  & 466.72\tabularnewline
 & 50  & 0.1  & 5  & 0.5  & -0.9  & 35103.81  & 9423.84  & 25.5  & 1376.62\tabularnewline
 & 50  & 0.1  & 5  & 0.5  & -0.8  & 12640.32  & 2615.45  & 25.5  & 495.6988\tabularnewline
$\bar{Q}_{\lambda,t}^{\left(r\right)}$  & 50  & 0.05  & 1  & 0.9  & 1  & 2005.07  & 515.08  & 10.5  & 190.9590\tabularnewline
 & 50  & 0.05  & 1  & 0.95  & 1  & 192.31  & 27.32  & 10.5  & 18.3152\tabularnewline
 & 50  & 0.05  & 5  & 0.9  & 1  & 6851.53  & 1628.88  & 50.5  & 135.6739\tabularnewline
 & 50  & 0.05  & 5  & 0.95  & 1  & 1417.72  & 261.55  & 50.5  & 28.0737\tabularnewline
 & 50  & 0.1  & 1  & 0.9  & 1  & 111.16  & 18.43  & 5.5  & 20.2109\tabularnewline
 & 50  & 0.1  & 1  & 0.95  & 1  & 44.14  & 7.03  & 5.5  & 8.0255\tabularnewline
 & 50  & 0.1  & 5  & 0.9  & 1  & 480.42  & 73.81  & 25.5  & 18.84\tabularnewline
 & 50  & 0.1  & 5  & 0.95  & 1  & 222.61  & 24.91  & 25.5  & 8.7298\tabularnewline
\hline 
\end{tabular}
\end{table}

To obtain some insights regarding the performance of these charts
for detecting OOC processes, we will consider again taking $n_{0}\in\left\{ 50,100,200\right\} $,
with $n_{t}$ as in Table \ref{tab:KS_no_EWMA_sim}. We assess the
performance of the raw $\left(P_{t}\right)_{t}$ sequence, along with
the charts $\left(\tilde{Q}_{\lambda,t}^{\left(r\right)}\right)_{t}$
and $\left(\bar{Q}_{\lambda,t}^{\left(r\right)}\right)_{t}$. Here
we consider alarm rules based on $\alpha\in\left\{ 0.01,0.05\right\} $,
where $X_{0}\sim\mathrm{N}\left(0,1\right)$ but we take two settings
for the OOC $X_{t}$. In the first case the process is persistently
OOC in the sense that $X_{t}$ is distributed as $\mathrm{N}\left(1/2,1\right)$,
$\mathrm{N}\left(1,1\right)$, $\mathrm{N}\left(0,2\right)$, and
$\mathrm{Cauchy}$ for every $t\in\mathbb{N}$. In the second case,
the process is OOC with $X_{t}$ distributed dynamically as $\mathrm{N}\left(\mu_{t},1\right)$
or $\mathrm{N}\left(0,\sigma_{t}^{2}\right)$, where $\mu_{t}\sim\mathrm{N}\left(0,1/2\right)$
or $\mu_{t}\sim\mathrm{N}\left(0,1/4\right)$, and $\sigma_{t}^{2}\sim\chi_{1}^{2}$
or $\sigma_{t}^{2}\sim\chi_{2}^{2}$. For $\left(\tilde{Q}_{\lambda,t}^{\left(r\right)}\right)_{t}$
we choose $\lambda=1/2$ and $r\in\left\{ -0.9,-0.8\right\} $, while
for $\left(\bar{Q}_{\lambda,t}^{\left(r\right)}\right)_{t}$ we consider
$\lambda\in\left\{ 0.8,0.9\right\} $ and $r=1$. To assess the relative
performances, we compute the average time to the first $k$ alarms
$R_{k}$ for $k\in\left\{ 1,5\right\} $ from 100 replicates and report
the persistent OOC outcomes in Tables \ref{tab:KS_OC_EWMA_sim}, \ref{tab:KS_OC_EWMA_sim_part2}
and \ref{tab:KS_OC_EWMA_sim_part3}, and the dynamic OOC outcomes
in Tables \ref{tab:KS_OC_EWMA_sim_dynamic}, \ref{tab:KS_OC_EWMA_sim_dynamic_part2},
and \ref{tab:KS_OC_EWMA_sim_dynamic_part3}. 

\begin{table}
\caption{Average $R_{k}$ under persistent OOC processes for KS $p$-value
charts $\left(P_{t}\right)_{t}$, $\left(\tilde{Q}_{\lambda,t}^{\left(r\right)}\right)_{t}$
and $\left(\bar{Q}_{\lambda,t}^{\left(r\right)}\right)_{t}$, using
100 simulation repetitions. This table contains results for $n_{0}=50$,
while Appendix Tables \ref{tab:KS_OC_EWMA_sim_part2} and \ref{tab:KS_OC_EWMA_sim_part3}
contains results for $n_{0}=100$ and $n_{0}=200$, respectively.}\label{tab:KS_OC_EWMA_sim}

\centering{}{\footnotesize{}%
\begin{tabular}{cccccccccc}
\hline 
{\footnotesize OOC process} & {\footnotesize$n_{0}$} & {\footnotesize$\alpha$} & {\footnotesize Chart} & {\footnotesize$r$} & {\footnotesize$\lambda$} & {\footnotesize mean $R_{1}$} & {\footnotesize st. err. $R_{1}$} & {\footnotesize mean $R_{5}$} & {\footnotesize st. err. $R_{5}$}\tabularnewline
\hline 
{\footnotesize$\text{N}\left(1/2,1\right)$} & {\footnotesize 50} & {\footnotesize 0.01} & {\footnotesize$P_{t}$} &  &  & {\footnotesize 36.06} & {\footnotesize 23.53} & {\footnotesize 185.58} & {\footnotesize 90.08}\tabularnewline
 & {\footnotesize 50} & {\footnotesize 0.01} & {\footnotesize$\tilde{Q}_{\lambda,t}^{\left(r\right)}$} & {\footnotesize -0.9} & {\footnotesize 0.5} & {\footnotesize 692.27} & {\footnotesize 300.34} & {\footnotesize 3515.62} & {\footnotesize 1359.13}\tabularnewline
 & {\footnotesize 50} & {\footnotesize 0.01} & {\footnotesize$\tilde{Q}_{\lambda,t}^{\left(r\right)}$} & {\footnotesize -0.8} & {\footnotesize 0.5} & {\footnotesize 1707.76} & {\footnotesize 1346.76} & {\footnotesize 3517.78} & {\footnotesize 2157.76}\tabularnewline
 & {\footnotesize 50} & {\footnotesize 0.01} & {\footnotesize$\bar{Q}_{\lambda,t}^{\left(r\right)}$} & {\footnotesize 1} & {\footnotesize 0.9} & {\footnotesize 54842.44} & {\footnotesize 54471.80} & {\footnotesize 252708.58} & {\footnotesize 250810.10}\tabularnewline
 & {\footnotesize 50} & {\footnotesize 0.01} & {\footnotesize$\bar{Q}_{\lambda,t}^{\left(r\right)}$} & {\footnotesize 1} & {\footnotesize 0.95} & {\footnotesize 759.31} & {\footnotesize 488.57} & {\footnotesize 4318.98} & {\footnotesize 3164.00}\tabularnewline
 & {\footnotesize 50} & {\footnotesize 0.05} & {\footnotesize$P_{t}$} &  &  & {\footnotesize 3.25} & {\footnotesize 0.56} & {\footnotesize 13.33} & {\footnotesize 1.76}\tabularnewline
 & {\footnotesize 50} & {\footnotesize 0.05} & {\footnotesize$\tilde{Q}_{\lambda,t}^{\left(r\right)}$} & {\footnotesize -0.9} & {\footnotesize 0.5} & {\footnotesize 85.23} & {\footnotesize 37.17} & {\footnotesize 305.66} & {\footnotesize 101.19}\tabularnewline
 & {\footnotesize 50} & {\footnotesize 0.05} & {\footnotesize$\tilde{Q}_{\lambda,t}^{\left(r\right)}$} & {\footnotesize -0.8} & {\footnotesize 0.5} & {\footnotesize 125.01} & {\footnotesize 65.85} & {\footnotesize 351.70} & {\footnotesize 158.42}\tabularnewline
 & {\footnotesize 50} & {\footnotesize 0.05} & {\footnotesize$\bar{Q}_{\lambda,t}^{\left(r\right)}$} & {\footnotesize 1} & {\footnotesize 0.9} & {\footnotesize 5.78} & {\footnotesize 1.87} & {\footnotesize 27.35} & {\footnotesize 8.29}\tabularnewline
 & {\footnotesize 50} & {\footnotesize 0.05} & {\footnotesize$\bar{Q}_{\lambda,t}^{\left(r\right)}$} & {\footnotesize 1} & {\footnotesize 0.95} & {\footnotesize 5.38} & {\footnotesize 1.64} & {\footnotesize 26.12} & {\footnotesize 7.81}\tabularnewline
{\footnotesize$\text{N}(1,1)$} & {\footnotesize 50} & {\footnotesize 0.01} & {\footnotesize$P_{t}$} &  &  & {\footnotesize 1.10} & {\footnotesize 0.04} & {\footnotesize 5.27} & {\footnotesize 0.08}\tabularnewline
 & {\footnotesize 50} & {\footnotesize 0.01} & {\footnotesize$\tilde{Q}_{\lambda,t}^{\left(r\right)}$} & {\footnotesize -0.9} & {\footnotesize 0.5} & {\footnotesize 1.53} & {\footnotesize 0.22} & {\footnotesize 6.15} & {\footnotesize 0.75}\tabularnewline
 & {\footnotesize 50} & {\footnotesize 0.01} & {\footnotesize$\tilde{Q}_{\lambda,t}^{\left(r\right)}$} & {\footnotesize -0.8} & {\footnotesize 0.5} & {\footnotesize 1.24} & {\footnotesize 0.06} & {\footnotesize 5.55} & {\footnotesize 0.17}\tabularnewline
 & {\footnotesize 50} & {\footnotesize 0.01} & {\footnotesize$\bar{Q}_{\lambda,t}^{\left(r\right)}$} & {\footnotesize 1} & {\footnotesize 0.9} & {\footnotesize 1.14} & {\footnotesize 0.09} & {\footnotesize 5.62} & {\footnotesize 0.27}\tabularnewline
 & {\footnotesize 50} & {\footnotesize 0.01} & {\footnotesize$\bar{Q}_{\lambda,t}^{\left(r\right)}$} & {\footnotesize 1} & {\footnotesize 0.95} & {\footnotesize 1.07} & {\footnotesize 0.03} & {\footnotesize 5.27} & {\footnotesize 0.06}\tabularnewline
 & {\footnotesize 50} & {\footnotesize 0.05} & {\footnotesize$P_{t}$} &  &  & {\footnotesize 1.00} & {\footnotesize 0.00} & {\footnotesize 5.04} & {\footnotesize 0.02}\tabularnewline
 & {\footnotesize 50} & {\footnotesize 0.05} & {\footnotesize$\tilde{Q}_{\lambda,t}^{\left(r\right)}$} & {\footnotesize -0.9} & {\footnotesize 0.5} & {\footnotesize 1.07} & {\footnotesize 0.04} & {\footnotesize 5.15} & {\footnotesize 0.08}\tabularnewline
 & {\footnotesize 50} & {\footnotesize 0.05} & {\footnotesize$\tilde{Q}_{\lambda,t}^{\left(r\right)}$} & {\footnotesize -0.8} & {\footnotesize 0.5} & {\footnotesize 1.07} & {\footnotesize 0.04} & {\footnotesize 5.09} & {\footnotesize 0.04}\tabularnewline
 & {\footnotesize 50} & {\footnotesize 0.05} & {\footnotesize$\bar{Q}_{\lambda,t}^{\left(r\right)}$} & {\footnotesize 1} & {\footnotesize 0.9} & {\footnotesize 1.00} & {\footnotesize 0.00} & {\footnotesize 5.01} & {\footnotesize 0.01}\tabularnewline
 & {\footnotesize 50} & {\footnotesize 0.05} & {\footnotesize$\bar{Q}_{\lambda,t}^{\left(r\right)}$} & {\footnotesize 1} & {\footnotesize 0.95} & {\footnotesize 1.04} & {\footnotesize 0.02} & {\footnotesize 5.08} & {\footnotesize 0.04}\tabularnewline
{\footnotesize$\text{N}(0,2)$} & {\footnotesize 50} & {\footnotesize 0.01} & {\footnotesize$P_{t}$} &  &  & {\footnotesize 373.91} & {\footnotesize 142.99} & {\footnotesize 1197.07} & {\footnotesize 308.02}\tabularnewline
 & {\footnotesize 50} & {\footnotesize 0.01} & {\footnotesize$\tilde{Q}_{\lambda,t}^{\left(r\right)}$} & {\footnotesize -0.9} & {\footnotesize 0.5} & {\footnotesize 36498.91} & {\footnotesize 13402.09} & {\footnotesize 122395.82} & {\footnotesize 28771.52}\tabularnewline
 & {\footnotesize 50} & {\footnotesize 0.01} & {\footnotesize$\tilde{Q}_{\lambda,t}^{\left(r\right)}$} & {\footnotesize -0.8} & {\footnotesize 0.5} & {\footnotesize 13889.67} & {\footnotesize 4649.08} & {\footnotesize 104257.62} & {\footnotesize 68852.89}\tabularnewline
 & {\footnotesize 50} & {\footnotesize 0.01} & {\footnotesize$\bar{Q}_{\lambda,t}^{\left(r\right)}$} & {\footnotesize 1} & {\footnotesize 0.9} & {\footnotesize 120454.57} & {\footnotesize 59829.87} & {\footnotesize 1047349.86} & {\footnotesize 528364.60}\tabularnewline
 & {\footnotesize 50} & {\footnotesize 0.01} & {\footnotesize$\bar{Q}_{\lambda,t}^{\left(r\right)}$} & {\footnotesize 1} & {\footnotesize 0.95} & {\footnotesize 3955.17} & {\footnotesize 2125.32} & {\footnotesize 18878.65} & {\footnotesize 6016.02}\tabularnewline
 & {\footnotesize 50} & {\footnotesize 0.05} & {\footnotesize$P_{t}$} &  &  & {\footnotesize 19.31} & {\footnotesize 4.05} & {\footnotesize 89.55} & {\footnotesize 10.52}\tabularnewline
 & {\footnotesize 50} & {\footnotesize 0.05} & {\footnotesize$\tilde{Q}_{\lambda,t}^{\left(r\right)}$} & {\footnotesize -0.9} & {\footnotesize 0.5} & {\footnotesize 1549.00} & {\footnotesize 441.42} & {\footnotesize 6552.55} & {\footnotesize 1558.61}\tabularnewline
 & {\footnotesize 50} & {\footnotesize 0.05} & {\footnotesize$\tilde{Q}_{\lambda,t}^{\left(r\right)}$} & {\footnotesize -0.8} & {\footnotesize 0.5} & {\footnotesize 493.85} & {\footnotesize 81.41} & {\footnotesize 3406.35} & {\footnotesize 561.96}\tabularnewline
 & {\footnotesize 50} & {\footnotesize 0.05} & {\footnotesize$\bar{Q}_{\lambda,t}^{\left(r\right)}$} & {\footnotesize 1} & {\footnotesize 0.9} & {\footnotesize 143.39} & {\footnotesize 44.38} & {\footnotesize 520.81} & {\footnotesize 132.94}\tabularnewline
 & {\footnotesize 50} & {\footnotesize 0.05} & {\footnotesize$\bar{Q}_{\lambda,t}^{\left(r\right)}$} & {\footnotesize 1} & {\footnotesize 0.95} & {\footnotesize 35.93} & {\footnotesize 6.71} & {\footnotesize 178.34} & {\footnotesize 23.72}\tabularnewline
{\footnotesize Cauchy} & {\footnotesize 50} & {\footnotesize 0.01} & {\footnotesize$P_{t}$} &  &  & {\footnotesize 59.00} & {\footnotesize 8.73} & {\footnotesize 420.38} & {\footnotesize 55.63}\tabularnewline
 & {\footnotesize 50} & {\footnotesize 0.01} & {\footnotesize$\tilde{Q}_{\lambda,t}^{\left(r\right)}$} & {\footnotesize -0.9} & {\footnotesize 0.5} & {\footnotesize 7138.85} & {\footnotesize 1220.77} & {\footnotesize 28983.64} & {\footnotesize 4323.72}\tabularnewline
 & {\footnotesize 50} & {\footnotesize 0.01} & {\footnotesize$\tilde{Q}_{\lambda,t}^{\left(r\right)}$} & {\footnotesize -0.8} & {\footnotesize 0.5} & {\footnotesize 5510.04} & {\footnotesize 1460.69} & {\footnotesize 23253.50} & {\footnotesize 5417.67}\tabularnewline
 & {\footnotesize 50} & {\footnotesize 0.01} & {\footnotesize$\bar{Q}_{\lambda,t}^{\left(r\right)}$} & {\footnotesize 1} & {\footnotesize 0.9} & {\footnotesize 7362.64} & {\footnotesize 2945.71} & {\footnotesize 42590.51} & {\footnotesize 17517.25}\tabularnewline
 & {\footnotesize 50} & {\footnotesize 0.01} & {\footnotesize$\bar{Q}_{\lambda,t}^{\left(r\right)}$} & {\footnotesize 1} & {\footnotesize 0.95} & {\footnotesize 1877.77} & {\footnotesize 562.83} & {\footnotesize 5772.51} & {\footnotesize 1262.26}\tabularnewline
 & {\footnotesize 50} & {\footnotesize 0.05} & {\footnotesize$P_{t}$} &  &  & {\footnotesize 10.92} & {\footnotesize 1.56} & {\footnotesize 46.41} & {\footnotesize 4.18}\tabularnewline
 & {\footnotesize 50} & {\footnotesize 0.05} & {\footnotesize$\tilde{Q}_{\lambda,t}^{\left(r\right)}$} & {\footnotesize -0.9} & {\footnotesize 0.5} & {\footnotesize 704.47} & {\footnotesize 94.75} & {\footnotesize 3067.70} & {\footnotesize 397.01}\tabularnewline
 & {\footnotesize 50} & {\footnotesize 0.05} & {\footnotesize$\tilde{Q}_{\lambda,t}^{\left(r\right)}$} & {\footnotesize -0.8} & {\footnotesize 0.5} & {\footnotesize 331.04} & {\footnotesize 41.87} & {\footnotesize 1206.27} & {\footnotesize 144.66}\tabularnewline
 & {\footnotesize 50} & {\footnotesize 0.05} & {\footnotesize$\bar{Q}_{\lambda,t}^{\left(r\right)}$} & {\footnotesize 1} & {\footnotesize 0.9} & {\footnotesize 34.00} & {\footnotesize 6.69} & {\footnotesize 134.56} & {\footnotesize 19.90}\tabularnewline
 & {\footnotesize 50} & {\footnotesize 0.05} & {\footnotesize$\bar{Q}_{\lambda,t}^{\left(r\right)}$} & {\footnotesize 1} & {\footnotesize 0.95} & {\footnotesize 13.20} & {\footnotesize 2.09} & {\footnotesize 73.73} & {\footnotesize 9.92}\tabularnewline
\hline 
\end{tabular}}{\footnotesize\par}
\end{table}

\begin{table}
\caption{Average $R_{k}$ under dynamic OOC processes for KS $p$-value charts
$\left(P_{t}\right)_{t}$, $\left(\tilde{Q}_{\lambda,t}^{\left(r\right)}\right)_{t}$
and $\left(\bar{Q}_{\lambda,t}^{\left(r\right)}\right)_{t}$, using
100 simulation repetitions. This table contains results for $n_{0}=50$,
while Appendix Tables \ref{tab:KS_OC_EWMA_sim_dynamic_part2} and
\ref{tab:KS_OC_EWMA_sim_dynamic_part3} contains results for $n_{0}=100$
and $n_{0}=200$, respectively.}\label{tab:KS_OC_EWMA_sim_dynamic}

\centering{}{\footnotesize{}%
\begin{tabular}{cccccccccc}
\hline 
{\footnotesize OOC process} & {\footnotesize$n_{0}$} & {\footnotesize$\alpha$} & {\footnotesize Chart} & {\footnotesize$r$} & {\footnotesize$\lambda$} & {\footnotesize mean $R_{1}$} & {\footnotesize st. err. $R_{1}$} & {\footnotesize mean $R_{5}$} & {\footnotesize st. err. $R_{5}$}\tabularnewline
\hline 
{\footnotesize$\text{N}\left(\mu_{t},1\right)$} & {\footnotesize 50} & {\footnotesize 0.01} & {\footnotesize$P_{t}$} &  &  & {\footnotesize 2.58} & {\footnotesize 0.22} & {\footnotesize 12.39} & {\footnotesize 0.46}\tabularnewline
{\footnotesize$\mu_{t}\sim\text{N}\left(0,1/2\right)$} & {\footnotesize 50} & {\footnotesize 0.01} & {\footnotesize$\tilde{Q}_{\lambda,t}^{\left(r\right)}$} & {\footnotesize -0.9} & {\footnotesize 0.5} & {\footnotesize 4.32} & {\footnotesize 0.42} & {\footnotesize 9.14} & {\footnotesize 0.48}\tabularnewline
 & {\footnotesize 50} & {\footnotesize 0.01} & {\footnotesize$\tilde{Q}_{\lambda,t}^{\left(r\right)}$} & {\footnotesize -0.8} & {\footnotesize 0.5} & {\footnotesize 3.29} & {\footnotesize 0.30} & {\footnotesize 8.39} & {\footnotesize 0.42}\tabularnewline
 & {\footnotesize 50} & {\footnotesize 0.01} & {\footnotesize$\bar{Q}_{\lambda,t}^{\left(r\right)}$} & {\footnotesize 1} & {\footnotesize 0.9} & {\footnotesize 4.83} & {\footnotesize 0.56} & {\footnotesize 25.04} & {\footnotesize 1.50}\tabularnewline
 & {\footnotesize 50} & {\footnotesize 0.01} & {\footnotesize$\bar{Q}_{\lambda,t}^{\left(r\right)}$} & {\footnotesize 1} & {\footnotesize 0.95} & {\footnotesize 3.55} & {\footnotesize 0.39} & {\footnotesize 20.24} & {\footnotesize 0.99}\tabularnewline
 & {\footnotesize 50} & {\footnotesize 0.05} & {\footnotesize$P_{t}$} &  &  & {\footnotesize 1.91} & {\footnotesize 0.17} & {\footnotesize 9.37} & {\footnotesize 0.29}\tabularnewline
 & {\footnotesize 50} & {\footnotesize 0.05} & {\footnotesize$\tilde{Q}_{\lambda,t}^{\left(r\right)}$} & {\footnotesize -0.9} & {\footnotesize 0.5} & {\footnotesize 3.10} & {\footnotesize 0.28} & {\footnotesize 7.68} & {\footnotesize 0.36}\tabularnewline
 & {\footnotesize 50} & {\footnotesize 0.05} & {\footnotesize$\tilde{Q}_{\lambda,t}^{\left(r\right)}$} & {\footnotesize -0.8} & {\footnotesize 0.5} & {\footnotesize 2.33} & {\footnotesize 0.17} & {\footnotesize 6.76} & {\footnotesize 0.20}\tabularnewline
 & {\footnotesize 50} & {\footnotesize 0.05} & {\footnotesize$\bar{Q}_{\lambda,t}^{\left(r\right)}$} & {\footnotesize 1} & {\footnotesize 0.9} & {\footnotesize 2.22} & {\footnotesize 0.18} & {\footnotesize 11.65} & {\footnotesize 0.49}\tabularnewline
 & {\footnotesize 50} & {\footnotesize 0.05} & {\footnotesize$\bar{Q}_{\lambda,t}^{\left(r\right)}$} & {\footnotesize 1} & {\footnotesize 0.95} & {\footnotesize 1.83} & {\footnotesize 0.12} & {\footnotesize 10.65} & {\footnotesize 0.44}\tabularnewline
{\footnotesize$\text{N}\left(\mu_{t},1\right)$} & {\footnotesize 50} & {\footnotesize 0.01} & {\footnotesize$P_{t}$} &  &  & {\footnotesize 3.98} & {\footnotesize 0.33} & {\footnotesize 19.18} & {\footnotesize 0.86}\tabularnewline
{\footnotesize$\mu_{t}\sim\text{N}\left(0,1/4\right)$} & {\footnotesize 50} & {\footnotesize 0.01} & {\footnotesize$\tilde{Q}_{\lambda,t}^{\left(r\right)}$} & {\footnotesize -0.9} & {\footnotesize 0.5} & {\footnotesize 9.82} & {\footnotesize 1.48} & {\footnotesize 19.16} & {\footnotesize 1.84}\tabularnewline
 & {\footnotesize 50} & {\footnotesize 0.01} & {\footnotesize$\tilde{Q}_{\lambda,t}^{\left(r\right)}$} & {\footnotesize -0.8} & {\footnotesize 0.5} & {\footnotesize 7.26} & {\footnotesize 0.74} & {\footnotesize 17.95} & {\footnotesize 1.42}\tabularnewline
 & {\footnotesize 50} & {\footnotesize 0.01} & {\footnotesize$\bar{Q}_{\lambda,t}^{\left(r\right)}$} & {\footnotesize 1} & {\footnotesize 0.9} & {\footnotesize 10.58} & {\footnotesize 1.27} & {\footnotesize 52.62} & {\footnotesize 2.71}\tabularnewline
 & {\footnotesize 50} & {\footnotesize 0.01} & {\footnotesize$\bar{Q}_{\lambda,t}^{\left(r\right)}$} & {\footnotesize 1} & {\footnotesize 0.95} & {\footnotesize 7.66} & {\footnotesize 0.84} & {\footnotesize 39.93} & {\footnotesize 2.40}\tabularnewline
 & {\footnotesize 50} & {\footnotesize 0.05} & {\footnotesize$P_{t}$} &  &  & {\footnotesize 2.50} & {\footnotesize 0.20} & {\footnotesize 12.77} & {\footnotesize 0.43}\tabularnewline
 & {\footnotesize 50} & {\footnotesize 0.05} & {\footnotesize$\tilde{Q}_{\lambda,t}^{\left(r\right)}$} & {\footnotesize -0.9} & {\footnotesize 0.5} & {\footnotesize 5.23} & {\footnotesize 0.45} & {\footnotesize 11.70} & {\footnotesize 0.80}\tabularnewline
 & {\footnotesize 50} & {\footnotesize 0.05} & {\footnotesize$\tilde{Q}_{\lambda,t}^{\left(r\right)}$} & {\footnotesize -0.8} & {\footnotesize 0.5} & {\footnotesize 4.61} & {\footnotesize 0.49} & {\footnotesize 11.96} & {\footnotesize 0.85}\tabularnewline
 & {\footnotesize 50} & {\footnotesize 0.05} & {\footnotesize$\bar{Q}_{\lambda,t}^{\left(r\right)}$} & {\footnotesize 1} & {\footnotesize 0.9} & {\footnotesize 3.40} & {\footnotesize 0.36} & {\footnotesize 18.10} & {\footnotesize 0.83}\tabularnewline
 & {\footnotesize 50} & {\footnotesize 0.05} & {\footnotesize$\bar{Q}_{\lambda,t}^{\left(r\right)}$} & {\footnotesize 1} & {\footnotesize 0.95} & {\footnotesize 2.90} & {\footnotesize 0.28} & {\footnotesize 14.36} & {\footnotesize 0.57}\tabularnewline
{\footnotesize$\text{N}\left(\mu_{t},1\right)$} & {\footnotesize 50} & {\footnotesize 0.01} & {\footnotesize$P_{t}$} &  &  & {\footnotesize 3.61} & {\footnotesize 0.35} & {\footnotesize 18.02} & {\footnotesize 0.82}\tabularnewline
{\footnotesize$\sigma_{t}^{2}\sim\chi_{1}^{2}$} & {\footnotesize 50} & {\footnotesize 0.01} & {\footnotesize$\tilde{Q}_{\lambda,t}^{\left(r\right)}$} & {\footnotesize -0.9} & {\footnotesize 0.5} & {\footnotesize 6.17} & {\footnotesize 0.71} & {\footnotesize 13.43} & {\footnotesize 1.00}\tabularnewline
 & {\footnotesize 50} & {\footnotesize 0.01} & {\footnotesize$\tilde{Q}_{\lambda,t}^{\left(r\right)}$} & {\footnotesize -0.8} & {\footnotesize 0.5} & {\footnotesize 6.15} & {\footnotesize 0.72} & {\footnotesize 12.85} & {\footnotesize 1.16}\tabularnewline
 & {\footnotesize 50} & {\footnotesize 0.01} & {\footnotesize$\bar{Q}_{\lambda,t}^{\left(r\right)}$} & {\footnotesize 1} & {\footnotesize 0.9} & {\footnotesize 9.62} & {\footnotesize 1.07} & {\footnotesize 48.12} & {\footnotesize 3.59}\tabularnewline
 & {\footnotesize 50} & {\footnotesize 0.01} & {\footnotesize$\bar{Q}_{\lambda,t}^{\left(r\right)}$} & {\footnotesize 1} & {\footnotesize 0.95} & {\footnotesize 6.52} & {\footnotesize 0.72} & {\footnotesize 30.31} & {\footnotesize 1.75}\tabularnewline
 & {\footnotesize 50} & {\footnotesize 0.05} & {\footnotesize$P_{t}$} &  &  & {\footnotesize 2.07} & {\footnotesize 0.21} & {\footnotesize 11.68} & {\footnotesize 0.46}\tabularnewline
 & {\footnotesize 50} & {\footnotesize 0.05} & {\footnotesize$\tilde{Q}_{\lambda,t}^{\left(r\right)}$} & {\footnotesize -0.9} & {\footnotesize 0.5} & {\footnotesize 4.20} & {\footnotesize 0.40} & {\footnotesize 10.49} & {\footnotesize 0.67}\tabularnewline
 & {\footnotesize 50} & {\footnotesize 0.05} & {\footnotesize$\tilde{Q}_{\lambda,t}^{\left(r\right)}$} & {\footnotesize -0.8} & {\footnotesize 0.5} & {\footnotesize 3.60} & {\footnotesize 0.34} & {\footnotesize 8.59} & {\footnotesize 0.49}\tabularnewline
 & {\footnotesize 50} & {\footnotesize 0.05} & {\footnotesize$\bar{Q}_{\lambda,t}^{\left(r\right)}$} & {\footnotesize 1} & {\footnotesize 0.9} & {\footnotesize 3.06} & {\footnotesize 0.26} & {\footnotesize 16.18} & {\footnotesize 0.81}\tabularnewline
 & {\footnotesize 50} & {\footnotesize 0.05} & {\footnotesize$\bar{Q}_{\lambda,t}^{\left(r\right)}$} & {\footnotesize 1} & {\footnotesize 0.95} & {\footnotesize 2.48} & {\footnotesize 0.22} & {\footnotesize 13.23} & {\footnotesize 0.64}\tabularnewline
{\footnotesize$\text{N}\left(\mu_{t},1\right)$} & {\footnotesize 50} & {\footnotesize 0.01} & {\footnotesize$P_{t}$} &  &  & {\footnotesize 8.10} & {\footnotesize 0.80} & {\footnotesize 38.57} & {\footnotesize 1.91}\tabularnewline
{\footnotesize$\sigma_{t}^{2}\sim\chi_{2}^{2}$} & {\footnotesize 50} & {\footnotesize 0.01} & {\footnotesize$\tilde{Q}_{\lambda,t}^{\left(r\right)}$} & {\footnotesize -0.9} & {\footnotesize 0.5} & {\footnotesize 47.59} & {\footnotesize 5.25} & {\footnotesize 113.35} & {\footnotesize 13.21}\tabularnewline
 & {\footnotesize 50} & {\footnotesize 0.01} & {\footnotesize$\tilde{Q}_{\lambda,t}^{\left(r\right)}$} & {\footnotesize -0.8} & {\footnotesize 0.5} & {\footnotesize 30.55} & {\footnotesize 3.66} & {\footnotesize 76.98} & {\footnotesize 7.22}\tabularnewline
 & {\footnotesize 50} & {\footnotesize 0.01} & {\footnotesize$\bar{Q}_{\lambda,t}^{\left(r\right)}$} & {\footnotesize 1} & {\footnotesize 0.9} & {\footnotesize 55.59} & {\footnotesize 8.91} & {\footnotesize 254.77} & {\footnotesize 24.81}\tabularnewline
 & {\footnotesize 50} & {\footnotesize 0.01} & {\footnotesize$\bar{Q}_{\lambda,t}^{\left(r\right)}$} & {\footnotesize 1} & {\footnotesize 0.95} & {\footnotesize 21.14} & {\footnotesize 2.57} & {\footnotesize 103.50} & {\footnotesize 6.81}\tabularnewline
 & {\footnotesize 50} & {\footnotesize 0.05} & {\footnotesize$P_{t}$} &  &  & {\footnotesize 3.14} & {\footnotesize 0.24} & {\footnotesize 18.04} & {\footnotesize 0.88}\tabularnewline
 & {\footnotesize 50} & {\footnotesize 0.05} & {\footnotesize$\tilde{Q}_{\lambda,t}^{\left(r\right)}$} & {\footnotesize -0.9} & {\footnotesize 0.5} & {\footnotesize 18.31} & {\footnotesize 1.95} & {\footnotesize 39.48} & {\footnotesize 3.10}\tabularnewline
 & {\footnotesize 50} & {\footnotesize 0.05} & {\footnotesize$\tilde{Q}_{\lambda,t}^{\left(r\right)}$} & {\footnotesize -0.8} & {\footnotesize 0.5} & {\footnotesize 15.54} & {\footnotesize 1.65} & {\footnotesize 32.88} & {\footnotesize 2.91}\tabularnewline
 & {\footnotesize 50} & {\footnotesize 0.05} & {\footnotesize$\bar{Q}_{\lambda,t}^{\left(r\right)}$} & {\footnotesize 1} & {\footnotesize 0.9} & {\footnotesize 4.23} & {\footnotesize 0.45} & {\footnotesize 27.51} & {\footnotesize 1.62}\tabularnewline
 & {\footnotesize 50} & {\footnotesize 0.05} & {\footnotesize$\bar{Q}_{\lambda,t}^{\left(r\right)}$} & {\footnotesize 1} & {\footnotesize 0.95} & {\footnotesize 4.07} & {\footnotesize 0.45} & {\footnotesize 22.92} & {\footnotesize 1.29}\tabularnewline
\hline 
\end{tabular}}{\footnotesize\par}
\end{table}

Finally we note that we are not the first to consider the use of KS
tests in the SPC literature, and our work is preceded by the texts
of \citet{bakir2012nonparametric}, \citet{ross2012two}, \citet{ross2015cpm},
and \citet{khrueasom2016integrated}. In particular, we note that
\citet{ross2012two} use the asymptotic distribution of the KS statistic
to deduce a comprehensive table of time-dependent thresholds for SPC
via normalised forms of KS statistics $\left(Z_{t}\right)_{t}$ for
calibrating charts to achieve some specific ARL values. These dynamic
thresholds were then implemented in the \textsf{R} package \textsf{cpm}
in \citet{ross2015cpm}. ARL estimation is also performed in \citet{khrueasom2016integrated}
for one-sample KS charts using the fact that the empirical CDF multiplied
by $n$ is binomial pointwise, although the use of this fact is not
made theoretically rigorous. Lastly, in these previous works, there
is no treatment of EWMA or EWMA-like charts based on KS statistics,
nor are general formulas available for control of the ARL or $k$-ARL
under fixed alarm rules, for finite $n_{0}$ and $\left(n_{t}\right)_{t}$.

\subsection{The distribution of uniform EWMA charts }

\label{subsec:Uniform_EWMA_sim}

To assess the veracity of the theory from Section \ref{subsec:Uniform_EWMA_Theorems},
we simulate $10000$ replicates of the random variable $\tilde{U}_{\lambda,t}$
for $\lambda\in\left\{ 0.3,0.5,0.7\right\} $ up to times $t\in\left\{ 2,3,4\right\} $,
taking $u_{0}=1/2$ in every case. Histograms for each of the $\left(t,\lambda\right)$
combinations are plotted and displayed against the corresponding theoretical
PDF for the statistic $\tilde{U}_{\lambda,t}$ as predicted by Proposition
\ref{prop:PDF_of_Ut} in Figure \ref{fig:Plots-of-PDFs}. Furthermore,
we demonstrate that for each of the scenarios above, the statistic
$\tilde{U}_{\lambda,t}$ is left-tail super-uniform as predicted by
Proposition \ref{prop:Ut_left_super_unif} by plotting the CDFs against
that of the uniform distribution, in Figure \ref{fig:Plots-of-CDFs}. 

\begin{figure}
\begin{centering}
\includegraphics[width=15cm]{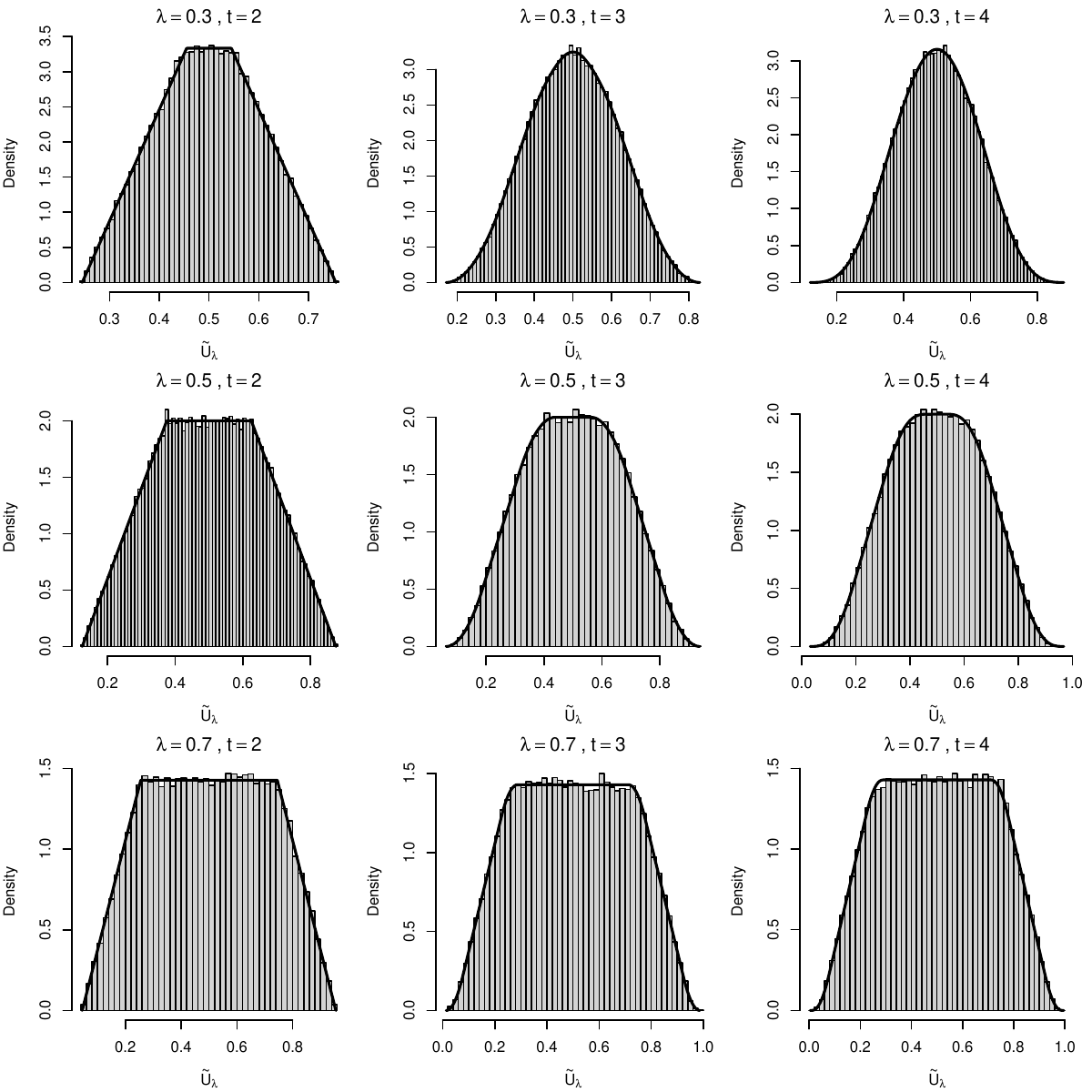}
\par\end{centering}
\caption{Plots of PDFs of the random variables $\tilde{U}_{\lambda,t}$ with
initialisation $u_{0}=1/2$, for $\lambda\in\left\{ 0.3,0.5,0.7\right\} $
and $t\in\left\{ 2,3,4\right\} $ along with histograms of 10000 replicates
of the corresponding variable.}\label{fig:Plots-of-PDFs}

\end{figure}

\subsection{Directional and coordinate localisation charts}

\label{subsec:Directional-and-coordinate-localisation-charts}

To demonstrate the results of Section \ref{subsec:Directional chart},
we consider a pair of examples. In the first case, we consider a one-phase
setting where $X_{0}\sim\mathrm{N}\left(0,\Sigma\right)$ is a $d=3$
dimensional normal random vector with $\Sigma=\left(\Sigma_{jk}\right)_{j,k\in\left[d\right]}$,
where $\Sigma_{jj}=1$ for each $j\in\left[d\right]$ and $\Sigma_{jk}=\rho>0$
for $j\ne k$. For each time point $t\in\mathbb{N}$, we observe the
OOC random variable $X_{t}\sim\mathrm{N}\left(\mu,\Sigma\right)$,
where $\mu=\left(\delta,0,-\delta\right)$ for $\delta>0$. Noting
that for each $t$ and $j$, $Z_{t,j}=X_{t,j}/\sqrt{\Sigma_{jj}}\sim\mathrm{N}\left(0,1\right)$
under the IC hypothesis, we can implement the method from Section
\ref{subsec:Directional chart} via the $p$-value sequences $\left(P_{t}^{\left(j,\square\right)}\right)_{t}$,
where $P_{t}^{\left(j,\ge\right)}=1-\Phi\left(Z_{t,j}\right)$ and
$P_{t}^{\left(j,\le\right)}=\Phi\left(Z_{t,j}\right)$, and combinations
are obtained via the Bonferroni rule. We assess the power of our directional
and coordinate localisation charts via implementation with control
limits $\alpha\in\left\{ 0.01,0.05\right\} $, and consider scenarios
$\delta\in\left\{ 0.5,1\right\} $ and $\rho\in\left\{ 0,0.5,0.9\right\} $.
To assess the power, we conduct $100$ replications of simulation
runs whereupon we compute the average run times $R$, and the number
of detected OOC directions (\#OOC; i.e., whether both OOC directions
are detected, only one, or neither) at the time of alarm. To address
the FWER guarantee of Proposition \ref{prop:Directional_correct_FWER},
we conduct $10000$ replicates of a one time point run (i.e., $t=1$)
and compute the average number of family-wise errors (FWEs) over those
runs. The results of these simulations appear in Table \ref{tab:Normal_directional}. 

\begin{table}
\caption{Average run lengths $R$, numbers of detected OOC directions at time
of alarm \#OOC, and numbers of single time point family wise errors
(at time $t=1$) for $d=3$ variable normal model tested using the
procedure from Section \ref{subsec:Directional chart} with $Z$-tests.
Average run lengths and number of true positives are reported via
100 simulation replications, while single time point family wise errors
are obtained from 10000 simulations.}\label{tab:Normal_directional}

\centering{}%
\begin{tabular}{cccccc}
\hline 
$\delta$ & $\rho$ & $\alpha$ & Mean $R$ & Mean \#OOC  & Mean FWEs\tabularnewline
\hline 
0.5 & 0 & 0.01 & 53.72 & 0.81 & 0.0044\tabularnewline
0.5 & 0 & 0.05 & 13.55 & 0.78 & 0.0230\tabularnewline
0.5 & 0.5 & 0.01 & 70.98 & 0.87 & 0.0037\tabularnewline
0.5 & 0.5 & 0.05 & 12.67 & 0.85 & 0.0201\tabularnewline
0.5 & 0.9 & 0.01 & 66.05 & 0.99 & 0.0037\tabularnewline
0.5 & 0.9 & 0.05 & 15.47 & 1.00 & 0.0245\tabularnewline
1 & 0 & 0.01 & 15.38 & 0.99 & 0.0038\tabularnewline
1 & 0 & 0.05 & 5.55 & 1.06 & 0.0198\tabularnewline
1 & 0.5 & 0.01 & 17.7 & 0.96 & 0.0037\tabularnewline
1 & 0.5 & 0.05 & 5.99 & 0.96 & 0.0220\tabularnewline
1 & 0.9 & 0.01 & 14.8 & 1.00 & 0.0047\tabularnewline
1 & 0.9 & 0.05 & 6.57 & 1.00 & 0.0259\tabularnewline
\hline 
\end{tabular}
\end{table}

In our second scenario, we consider instead that $X_{0}\sim\mathrm{Cauchy}\left(\mu_{0},\Sigma\right)$
is a $d=3$ dimensional random vector whose distribution is multivariate
Cauchy with location parameter $\mu_{0}=0$ and scale parameter $\Sigma$
(see, e.g., \citealp{KotzNadarajah2004}), where $\Sigma$ is defined
as above. At time $t\in\mathbb{N}$, we then consider $X_{t}\sim\mathrm{Cauchy}\left(\mu,\Sigma\right)$,
where $\mu$ is the centroid parameter for the multivariate Cauchy
and is defined as above.We will consider a two-phase design where
in Phase I we observe an IID sample $\mathbf{X}_{0}$ of size $n_{0}\in\left\{ 20,50,100\right\} $
and in Phase II, we observe an IID sample $\mathbf{X}_{t}$ of size
$n_{t}=N_{t}\sim\mathrm{DiscUnif}\left[n_{0}-10,n_{0}+10\right]$.
We wish to test the directional hypotheses at each time $t$ that
the median of the distribution at time $t$ for each coordinate $j\in\left[d\right]$
has experienced an upward or downward shift (i.e., $\mu_{j}-\mu_{0,j}>0$
or $\mu_{j}-\mu_{0,j}<0$) from that of the Phase I process. To this
end, we employ one-sided Mann--Whitney $U$-tests (see, e.g., \citealp[Sec. 5.7]{GibbonsChakraborti2010})
to generate our sequences of $p$-values $\left(P_{t}^{\left(j,\square\right)}\right)_{t}$.
We assess the performance of the procedure from Section \ref{subsec:Directional chart}
with these $p$-values, over the same configurations as above. To
assess the power, $100$ replications of simulation runs are used
to compute the average $R$ and \#OOC at the time of alarm. Simulation
average FWEs are computed via $10000$ replicates of a one time-point
run (i.e., $t=1$), and the results for $n_{0}=50$ are reported in
Table \ref{tab:Cauchy_50}, with the results for $n_{0}=20$ and $100$
reported in Tables \ref{tab:Cauchy_20} and \ref{tab:Cauchy_100},
respectively. 

\begin{table}
\caption{Average run lengths $R$, numbers of detected OOC directions at time
of alarm \#OOC, and numbers of single time point family wise errors
(at time $t=1$) for $d=3$ variable Cauchy model with $n_{0}=50$
tested using the procedure from Section \ref{subsec:Directional chart}
with Mann--Whitney $U$-tests. Average run lengths and number of
true positives are reported via 100 simulation replications, while
single time point family wise errors are obtained from 10000 simulations.}\label{tab:Cauchy_50}

\centering{}%
\begin{tabular}{cccccc}
\hline 
$\delta$ & $\rho$ & $\alpha$ & Mean $R$ & Mean \#OOC  & Mean FWEs\tabularnewline
\hline 
0.5 & 0 & 0.01 & 224.49 & 0.92 & 0.0023\tabularnewline
0.5 & 0 & 0.05 & 11.14 & 0.97 & 0.0195\tabularnewline
0.5 & 0.5 & 0.01 & 71.7 & 0.98 & 0.004\tabularnewline
0.5 & 0.5 & 0.05 & 9.7 & 1 & 0.0212\tabularnewline
0.5 & 0.9 & 0.01 & 33.39 & 1 & 0.0036\tabularnewline
0.5 & 0.9 & 0.05 & 6.21 & 0.99 & 0.0252\tabularnewline
1 & 0 & 0.01 & 2.86 & 1.2 & 0.0044\tabularnewline
1 & 0 & 0.05 & 1.76 & 1.47 & 0.0319\tabularnewline
1 & 0.5 & 0.01 & 1.98 & 1.13 & 0.0053\tabularnewline
1 & 0.5 & 0.05 & 1.13 & 1.36 & 0.0289\tabularnewline
1 & 0.9 & 0.01 & 1.85 & 1.06 & 0.0054\tabularnewline
1 & 0.9 & 0.05 & 1.04 & 1.3 & 0.0231\tabularnewline
\hline 
\end{tabular}
\end{table}

\subsection{Results}

In Tables~\ref{tab:Elementary_normal_sim} and~\ref{tab:AR1_Sim},
we observe that the IC bounds on the ARL and $k$-ARL from Propositions
\ref{prop:ARLnoconditions}--\ref{prop:Conditional_Super_Unif_kARL}
are faithful albeit somewhat conservative. In particular, the ratio
between the realised mean $R_{k}$ and the corresponding lower bound
can be as high as $23.89$, although Table~\ref{tab:AR1_Sim} also
shows that in some non-trivial situations the bounds can be quite
sharp, with ratios as low as $1.24$ between the realised mean ARL
and the lower bound. Furthermore, Table~\ref{tab:AR1_Sim} shows
that the bounds behave as predicted in both the (marginal) super-uniformity
and conditional super-uniformity settings.

In Section~\ref{subsec:Kolmogorov=00003D002013Smirnov-based-EWMA-ch},
we demonstrate the construction of a nontrivial EWMA-like chart for
KS tests under VSS, whose theoretical properties could not be analysed
without the techniques from this work. From Table~\ref{tab:KS_no_EWMA_sim},
we observe that the $p$-value charts without weighted averaging realise
the bounds of Propositions \ref{prop:ARLnoconditions} and \ref{prop:kARL_superuniform}.
As with the results from Section~\ref{subsec:Elementary-examples},
these bounds are realised somewhat conservatively, with ratios between
the realised mean and the bound as high as $32.94$ and as low as
$6.51$. We also observe that larger $\alpha$ levels tend to decrease
the discrepancy in the bound.

In Tables~\ref{tab:KS_IC_EWMA_sim} and~\ref{tab:KS_IC_EWMA_sim_part2},
we demonstrate that the EWMA-like $p$-value charts from Section~\ref{subsec:EWMA-like-charts-p-avg}
have correct lower bounds on the $k$-ARL. In particular, the $k$-ARL
lower bounds can be extremely conservative, even in comparison to
the non-averaged charts, with ratios as high as $\approx3500$ for
some combinations of $\lambda$ and $r$, although the ratio can be
as low as $7.63$. We further note that the charts with negative $r$
tended to produce larger $R_{k}$ values than those with positive
$r$. We emphasise that we only tested a limited number of combinations
of $\lambda$ and $r$, exactly because for some choices, even at
the modest $\alpha$ levels assessed, the resulting $R_{k}$ values
can be arbitrarily large. We note that although these charts can be
conservative with respect to IC $k$-ARL control, they are nevertheless
powerful across a range of OOC processes, as observed in Tables~\ref{tab:KS_OC_EWMA_sim}--\ref{tab:KS_OC_EWMA_sim_part3}
and~\ref{tab:KS_OC_EWMA_sim_dynamic}--\ref{tab:KS_OC_EWMA_sim_dynamic_part3},
where we assess performance under both the first alarm rule and the
first $k$ alarms rule for OOC detection. 

From Tables~\ref{tab:KS_OC_EWMA_sim}, \ref{tab:KS_OC_EWMA_sim_part2},
and~\ref{tab:KS_OC_EWMA_sim_part3}, we observe that in the persistent
OOC setting the original $p$-value chart (without averaging) typically
attains the smallest detection delays. Nevertheless, the penalty from
averaging, a slightly larger time to detection in many scenarios,
can be seen as a worthy trade in the face of the dramatic increases
in the IC $k$-ARL levels that we observe.

The fact that the EWMA-like charts raise a first alarm no earlier
than the $p$-value chart (i.e., the $k=1$ setting) at the same level
$\alpha$ comes down to the simple fact that for every $r\ne0$, $\lambda\in(0,1)$,
and $t\ge1$, 
\[
\tilde{Q}_{\lambda,t}^{(r)}\;\ge\;\min_{s\in[t]}P_{s}.
\]
However, this type of ordering does not necessarily extend to the
$k$ alarms rule. Empirically, in the persistent OOC setting we still
observe that the $p$-value chart raises an alarm before the EWMA-like
charts in every tested scenario. The situation changes when considering
the dynamic OOC settings in Tables~\ref{tab:KS_OC_EWMA_sim_dynamic},
\ref{tab:KS_OC_EWMA_sim_dynamic_part2}, and~\ref{tab:KS_OC_EWMA_sim_dynamic_part3}.
Here we observe that in many cases it is possible for the EWMA-like
charts with $r<0$ to raise, for example, a fifth alarm before the
$p$-value chart does at the same $\alpha$ level. This is quite atypical,
since in most EWMA charts the extra power from moving averages is
obtained by modifying the control limits, which require additional
calibration and simulation to set (see, e.g., \citealp{LucasSaccucci1990}
and \citealp{ReynoldsStoumbos2005ShewhartLimits}). Here, we observe
that in the dynamic setting it is possible to gain additional power
relative to the raw $p$-value chart under the $k$ alarms rule, while
simultaneously increasing the IC average time to the $k$th alarm,
without any additional overhead of control-limit calibration.

The results from Figures~\ref{fig:Plots-of-PDFs} and~\ref{fig:Plots-of-CDFs}
demonstrate directly the functional forms of the distributions arising
from uniform EWMA charts, and verify the claims of Propositions \ref{prop:PDF_of_Ut}
and~\ref{prop:Ut_left_super_unif}. Lastly, Tables~\ref{tab:Normal_directional},
\ref{tab:Cauchy_50}, \ref{tab:Cauchy_20}, and~\ref{tab:Cauchy_100}
provide evidence towards the correctness of Proposition~\ref{prop:Directional_correct_FWER}
regarding the directional and coordinate localisation charts from
Section~\ref{subsec:Directional chart}. In all cases, we observe
that the average FWE is controlled at the time of alarm, and is somewhat
conservative as would be expected from the use of Bonferroni procedures.
We also observe that at the time of alarm, usually only one direction
(and typically only one coordinate) is identified at a time, and that
power increases with the size of the OOC deviation and the degree
of correlation between the coordinates. In addition, these results
provide a demonstration of a viable localisation approach in the face
of pathological Cauchy data with strong theoretical guarantees. 

\section{Concluding remarks}

\label{sec:concluding-remarks}

In this work, we have developed a general framework for SPC based
on charting $p$-values. The key requirement underlying our IC guarantees
is the super-uniformity property~(\ref{eq:superuniformity}), which
is satisfied by any valid $p$-value construction, irrespective of
the underlying DGP and without any assumption on stochastic dependence.
Within this framework, the simple Shewhart-type rule that raises an
alarm whenever $P_{t}\le\alpha$ admits universal lower bounds on
the ARL and its natural generalisation, the $k$-ARL. In particular,
Propositions~\ref{prop:ARLnoconditions} and \ref{prop:kARL_superuniform}
provide worst-case IC guarantees under (\ref{eq:superuniformity}),
and Propositions~\ref{prop:ARL_with_conditional_indep} and~\ref{prop:Conditional_Super_Unif_kARL}
show that these bounds improve to the familiar $\alpha^{-1}$ and
$k\alpha^{-1}$ sharp rates when conditional super-uniformity~(\ref{eq:Conditional_Super_Uni})
holds. Taken together, these results permit calibration of $p$-value
charts in a manner that is valid for all DGPs for which $\mathrm{H}_{0}$
holds, and mitigate against the reliance on simulation-based ARL calibration
tied to any parametric assumptions.

Further to these ARL results, we have demonstrated that $p$-values
can support meaningful weighted averaging schemes, despite the fact
that naive averaging of $p$-values generally fails to preserve validity.
By leveraging recent results on merging functions for dependent $p$-values,
we constructed EWMA-like schemes that produce, at each time $t$,
a new $p$-value for the hypothesis that $\mathrm{H}_{0}$ holds for
all time points up to time $t$ (Propositions~\ref{prop:EWMA-like-charts}
and \ref{prop:EWMA-like-charts-conditional}). This provides a generic
route to constructing smoothed charts whose IC behaviour remains theoretically
controlled, without the additional overhead of control-limit calibration.
We also studied the distribution theory of uniform EWMA processes,
both to clarify why naive uniform EWMA charts do not generally yield
super-uniform tails, and to obtain explicit distributional formulae
and left-tail guarantees (Propositions~\ref{prop:PDF_of_Ut} and~\ref{prop:Ut_left_super_unif}).
Complementary constructions based on $e$-values and $p$-to-$e$
calibrators appear in the Appendix, providing additional flexibility
in how one may design valid averaged charts.

We have further shown that the $p$-value framework extends naturally
to the problem of directional and coordinate localisation in multivariate
monitoring. Using a $p$-value-based closed-testing perspective and
a Holm-type shortcut method, we obtained a procedure that can both
raise a global alarm and then identify the set of affected coordinates
and their corresponding OOC directions, while maintaining FWER control
at level $\alpha$ at the time of alarm. This yields a localisation
approach that is modular in the choice of coordinate-wise test, and
remains applicable in both one- and two-phase designs.

The examples and simulations in Section~\ref{sec:Example-applications}
illustrate several practical features of the methodology. First, the
IC ARL and $k$-ARL lower bounds are faithful and, in some settings,
sharp, while in others they can be conservative, which is the natural
price of working without dependence assumptions. Second, the EWMA-like
charts can trade additional IC conservativeness for OOC power across
a range of alternatives, with particularly interesting behaviour under
the $k$-alarms rule in dynamic OOC scenarios. Lastly, the localisation
results provide a concrete demonstration that strong inferential guarantees
remain attainable even in challenging multivariate settings, including
situations with heavy-tailed data. Within these constructions, we
present two novel charts: an EWMA-like $p$-value chart based on KS
statistics that can be used in two-phase designs for monitoring deviations
from the IC distribution, and a nonparametric localisation chart,
based on Mann--Whitney $U$-statistics for non-standard scenarios
such as heavy-tailed settings.

There are a number of natural directions for further work. It would
be useful to develop principled guidance for selecting the averaging
parameters (such as $\lambda$ and $r$) to balance OOC detection
and IC conservativeness, and to study whether alternative merging
functions can reduce conservativeness while retaining the same level
of robustness. On the multivariate side, it is of interest to investigate
localisation procedures that exploit dependence structure to improve
power while still providing FWER control at the time of alarm. More
broadly, the $p$-value charting perspective suggests a modular approach
to SPC: given any setting where one can construct valid $p$-values
for $\mathrm{H}_{0}$, one can immediately obtain charts with explicit
IC guarantees, and layer on additional structure, including smoothing
and localisation, without sacrificing theoretical control of ARL and
$k$-ARL. 

\subsubsection*{Data availability statement}

The data and code that support the findings of this study are openly
available at \url{https://github.com/hiendn/Pvalue_SPC}.

\bibliographystyle{apalike2}
\bibliography{bib}

\section*{Appendix}

\subsection*{EWMA charts based on $e$-values}

We say that a random variable $E:\Omega\to\left[0,\infty\right]$
is an $e$-value with respect to the probability measure $\mathrm{P}_{0}$
if $\mathrm{E}_{0}\left[E\right]\le1$. Further, a function $\gamma:\left[0,1\right]\to\mathbb{R}_{\ge0}\cup\left\{ \infty\right\} $
is a $p$-to-$e$ calibrator if, for each $p$-value $P$ satisfying
$\mathrm{P}_{0}\left(P\le\alpha\right)\le\alpha$ for every $\alpha\in\left[0,1\right]$,
it holds that $\mathrm{E}_{0}\left[\gamma\left(P\right)\right]\le1$.
The following result appears as Proposition 2.1 of \citet{Vovk2021}. 
\begin{lem}
\label{lem:p-to-e=00003D00003D00003D00003D000020calibrator-1-1} A
decreasing function $\gamma:\left[0,1\right]\to\mathbb{R}_{\ge0}\cup\left\{ \infty\right\} $
is a $p$-to-$e$ calibrator if and only if $\int_{0}^{1}\gamma\left(p\right)\,\mathrm{d}p\le1$. 
\end{lem}

We now leverage Lemma \ref{lem:p-to-e=00003D00003D00003D00003D000020calibrator-1-1}
to obtain EWMA-like charts. To do so, we first note that given a set
of $e$-values $E_{1},\dots,E_{m}:\Omega\to\left[0,\infty\right]$,
it holds that any convex combination of these $e$-values remains
an $e$-value. That is, for any set of weights $w_{1},\dots,w_{m}\in\left[0,1\right]$
such that $\sum_{t=1}^{m}w_{t}=1$, the inequality 
\[
\mathrm{E}_{0}\left[\sum_{t=1}^{m}w_{t}E_{t}\right]=\sum_{t=1}^{m}w_{t}\mathrm{E}_{0}\left[E_{t}\right]\le\sum_{t=1}^{m}w_{t}=1\text{,}
\]
holds regardless of what the dependence behaviour between the $e$-values
is. Another simple fact is that, given an $e$-value $E$, the random
variable $P=\min\left\{ 1,1/E\right\} $ is a $p$-value, since Markov's
inequality implies that, for $\alpha\in\left(0,1\right]$, 
\[
\mathrm{P}_{0}\left(P\le\alpha\right)=\mathrm{P}_{0}\left(1/E\le\alpha\right)=\mathrm{P}_{0}\left(E\ge\frac{1}{\alpha}\right)\le\alpha\,\mathrm{E}_{0}\left[E\right]\le\alpha\text{.}
\]

Putting these two pieces together, to generate a sequence of $p$-values
that averages information across the history, we first convert the
$p$-values $\left(P_{t}\right)_{t}$ into corresponding $e$-values
$\left(E_{t}\right)_{t}$, defined by $E_{t}=\gamma\left(P_{t}\right)$,
where $\gamma$ satisfies the conditions of Lemma \ref{lem:p-to-e=00003D00003D00003D00003D000020calibrator-1-1}.
Then, we generate an EWMA sequence $\left(\tilde{E}_{\lambda,t}\right)_{t}$
using these $e$-values by computing 
\[
\tilde{E}_{\lambda,t}=\lambda E_{t}+\left(1-\lambda\right)\tilde{E}_{\lambda,t-1}\text{,}
\]
for $t\ge2$, where $\tilde{E}_{\lambda,1}=E_{1}$ and $\lambda\in\left(0,1\right)$.
We then obtain a super-uniform sequence of $p$-values with time-averaged
information $\left(Q_{\lambda,t}^{\mathrm{e}}\right)_{t}$, defined
by 
\[
Q_{\lambda,t}^{\mathrm{e}}=\min\left\{ 1,1/\tilde{E}_{\lambda,t}\right\} \text{.}
\]

\begin{prop}
\label{prop:E_value_averaging_super_unif-1}For each $\lambda\in\left(0,1\right)$,
if $\left(P_{t}\right)_{t}$ satisfies \emph{(\ref{eq:superuniformity})},
then $\left(Q_{\lambda,t}^{\mathrm{e}}\right)_{t}$ is a super-uniform
sequence under $\mathrm{P}_{0}$; i.e., for each $\alpha\in\left[0,1\right]$,
\[
\mathrm{P}_{0}\left(Q_{\lambda,t}^{\mathrm{e}}\le\alpha\right)\le\alpha\quad\mathrm{for\ all\ }t\in\mathbb{N}.
\]
\end{prop}

We note that there are many choices for the $p$-to-$e$ calibrator
$\gamma$, and $P=\min\{1,1/E\}$ is not the only method for generating
a $p$-value from an $e$-value. However, the choice $P=\min\{1,1/E\}$
is admissible in the sense that it is pointwise no larger than every
other transformation of $E$ that yields a super-uniform $p$-value.
On the other hand, there is no best choice for $\gamma$, although
it is known that $\gamma$ is admissible in the sense that there is
no other calibrator that yields a valid $e$-value and is pointwise
larger than $\gamma(p)$, if and only if $\lim_{p\to0}\gamma(p)=\infty$
and $\int_{0}^{1}\gamma(p)\,\mathrm{d}p=1$. Thus, our suggested calibrator
$\gamma(p)=\beta p^{\beta-1}$, for $\beta\in(0,1)$, is an admissible
and simple choice. These optimality results can be obtained from \citet[Sec. 2]{Vovk2021}.

Lastly, we do not know of a useful method for producing time-averaged
$p$-values using the $e$-value approach that yields a sequence of
conditionally super-uniform $p$-values from an initial sequence $\left(P_{t}\right)_{t}$
satisfying (\ref{eq:Conditional_Super_Uni}). 

\subsection*{Proofs of technical results}

\subsubsection*{Proof of Proposition \ref{prop:ARLnoconditions}}

Let $\nu=\left\lfloor 1/\alpha\right\rfloor $. Then $1-\left(m-1\right)\alpha>0$,
for $m\in\left[\nu\right]$ and $1-\nu\alpha\ge0$ at $m=\nu+1$.
Thus, 
\[
\sum_{m=1}^{\infty}\bigl[1-\left(m-1\right)\alpha\bigr]_{+}=\sum_{m=1}^{\nu+1}\bigl(1-\left(m-1\right)\alpha\bigr)=\left(\nu+1\right)\left(1-\frac{\alpha\nu}{2}\right).
\]
Next, write $u=1/\alpha$. Then 
\[
\left(\nu+1\right)\left(1-\frac{\alpha\nu}{2}\right)=\left(\nu+1\right)\left(1-\frac{\nu}{2u}\right)=\frac{1}{2}\left\{ 2\left(\nu+1\right)-\frac{\nu\left(\nu+1\right)}{u}\right\} .
\]
Observe that 
\[
2u\left(\nu+1\right)-\nu\left(\nu+1\right)\ge u^{2}+u\ \text{iff}\ \left(u-\nu\right)\left(u-(\nu+1)\right)\le0,
\]
which holds for each $\nu\le u<\nu+1$ (by the definition of $\nu$).
Therefore, from (\ref{eq:ARL-expression}) we obtain the desired conclusion.

\subsubsection*{Proof of Proposition \ref{prop:ARL_with_conditional_indep}}

Let us write $I_{t}=\mathbf{1}_{\left\{ P_{t}\le\alpha\right\} }$
and note that the run length can be written as $R=\inf\left\{ t\ge1:I_{t}=1\right\} $.
Then, with 
\[
\mathbb{S}_{t}=\left\{ \sum_{s=1}^{t}I_{s}=0\right\} =\left\{ I_{1}=\dots=I_{t}=0\right\} \text{,}
\]
we can write $\left\{ R>t\right\} =\mathbb{S}_{t}$.

Since $I_{t}:\Omega\to\left\{ 0,1\right\} $, (\ref{eq:Conditional_Super_Uni})
gives $\mathrm{E}_{0}\left(I_{t}\mid{\cal F}_{t-1}\right)\le\alpha$,
and hence 
\[
\mathrm{P}_{0}\left(I_{t}=0\mid{\cal F}_{t-1}\right)=\mathrm{E}_{0}\left(1-I_{t}\mid{\cal F}_{t-1}\right)\ge1-\alpha\text{, }\mathrm{a.s.}
\]
Thus, by the tower property and since $\mathbb{S}_{t-1}\in{\cal F}_{t-1}$,
\begin{align*}
\mathrm{P}_{0}\left(\mathbb{S}_{t}\right) & =\mathrm{E}_{0}\mathbf{1}_{\mathbb{S}_{t}}=\mathrm{E}_{0}\left[\mathbf{1}_{\mathbb{S}_{t-1}}\left(1-I_{t}\right)\right]\\
 & =\mathrm{E}_{0}\left[\mathbf{1}_{\mathbb{S}_{t-1}}\mathrm{E}_{0}\left(1-I_{t}\mid{\cal F}_{t-1}\right)\right]\\
 & \ge\left(1-\alpha\right)\mathrm{E}_{0}\mathbf{1}_{\mathbb{S}_{t-1}}=\left(1-\alpha\right)\mathrm{P}_{0}\left(\mathbb{S}_{t-1}\right)\text{.}
\end{align*}
Iterated applications of this bound then yield 
\[
\mathrm{P}_{0}\left(R>t\right)=\mathrm{P}_{0}\left(\mathbb{S}_{t}\right)\ge\left(1-\alpha\right)^{t}\text{,}
\]
noting that $\mathrm{P}_{0}\left(\mathbb{S}_{0}\right)=1$. The survival
function expression for the expectation formula then gives 
\[
\mathrm{E}_{0}R=\sum_{t=0}^{\infty}\mathrm{P}_{0}\left(R>t\right)\ge\sum_{t=0}^{\infty}\left(1-\alpha\right)^{t}=\frac{1}{\alpha}\text{,}
\]
as required, by geometric summation.

\subsubsection*{Proof of Proposition \ref{prop:kARL_superuniform}}

For any $m\ge1$, by Markov's inequality, 
\[
\mathrm{P}_{0}\left(R_{k}\le m\right)=\mathrm{P}_{0}\left(\sum_{t=1}^{m}\mathbf{1}_{\{P_{t}\le\alpha\}}\ge k\right)\le\frac{\sum_{t=1}^{m}\mathrm{E}_{0}\mathbf{1}_{\{P_{t}\le\alpha\}}}{k}\le\frac{\alpha m}{k}\text{.}
\]
Thus $\mathrm{P}_{0}\left(R_{k}\ge m+1\right)\ge1-\alpha m/k$. Using
the survival function representation for the expectation, $\mathrm{E}_{0}R_{k}=\sum_{m=0}^{\infty}\mathrm{P}_{0}\left(R_{k}\ge m+1\right)$,
we obtain 
\[
\mathrm{E}_{0}R_{k}\ge\sum_{m=0}^{\nu}\left(1-\frac{\alpha m}{k}\right)=\left(\nu+1\right)\left(1-\frac{\alpha\nu}{2k}\right)\text{,}
\]
upon discarding the negative terms. The second inequality follows
by an argument analogous to the proof of Proposition \ref{prop:ARLnoconditions}.

\subsubsection*{Proof of Proposition \ref{prop:Conditional_Super_Unif_kARL}}

Recall the notation $I_{t}=\mathbf{1}_{\{P_{t}\le\alpha\}}$ and define
$M_{t}=\sum_{s=1}^{t}\bigl(I_{s}-\alpha\bigr)$, $N_{t}=\sum_{s=1}^{t}I_{s}=M_{t}+\alpha t$,
and $M_{0}=0$. Since $\mathrm{E}_{0}\left[I_{t}\mid{\cal F}_{t-1}\right]\le\alpha$,
we have $\mathrm{E}_{0}\left[M_{t}\mid{\cal F}_{t-1}\right]\le M_{t-1}$,
so $\left(M_{t}\right)_{t}$ is a supermartingale adapted to $\left({\cal F}_{t}\right)_{t}$.

Let $T_{m}=\min\left\{ R_{k},m\right\} $ be a bounded stopping time.
Note that 
\[
R_{k}=\inf\left\{ t\ge1:N_{t}\ge k\right\} ,
\]
and, because $N_{t}$ increases by at most one at each step, $N_{R_{k}}=k$
on $\{R_{k}<\infty\}$. By Doob's optional stopping theorem for bounded
stopping times (see, e.g., \citealp[Thm. 3.8]{LattimoreSzepesvari2020}),
$\mathrm{E}_{0}M_{T_{m}}\le\mathrm{E}_{0}M_{0}=0$, implying 
\begin{equation}
\mathrm{E}_{0}N_{T_{m}}\le\alpha\mathrm{E}_{0}T_{m}.\label{eq:NTmvsTm}
\end{equation}
Since $T_{m}\nearrow R_{k}$ and $T_{m}\ge0$, the monotone convergence
theorem yields $\lim_{m\to\infty}\mathrm{E}_{0}T_{m}=\mathrm{E}_{0}R_{k}$.
For the left-hand side, $N_{T_{m}}\nearrow N_{R_{k}}$ almost surely;
hence 
\[
\lim_{m\to\infty}\mathrm{E}_{0}N_{T_{m}}=\mathrm{E}_{0}N_{R_{k}}=k\,\mathrm{P}_{0}\left(R_{k}<\infty\right)+\mathrm{E}_{0}\left[\lim_{t\to\infty}N_{t}\,\mathbf{1}_{\{R_{k}=\infty\}}\right]\le k,
\]
since on $\{R_{k}=\infty\}$ we have $\lim_{t\to\infty}N_{t}\le k-1$.
If $\mathrm{P}_{0}\left(R_{k}=\infty\right)>0$, then $\mathrm{E}_{0}R_{k}=\infty$
and the desired bound $\mathrm{E}_{0}R_{k}\ge k/\alpha$ is trivial.
Otherwise, $\mathrm{P}_{0}\left(R_{k}<\infty\right)=1$ and $\lim_{m\to\infty}\mathrm{E}_{0}N_{T_{m}}=k$.
Taking limits in (\ref{eq:NTmvsTm}) gives $k\le\alpha\mathrm{E}_{0}R_{k}$,
as required.

\subsubsection*{Proof of Proposition \ref{prop:EWMA-like-charts}}

Fix $\lambda\in(0,1)$, $r>-1$ with $r\neq0$ and $t\in\mathbb{N}$.
By (\ref{eq:EWMA-no-P0-weights}) we can write 
\[
S_{\lambda,t}^{\left(r\right)}=\sum_{s=1}^{t}w_{t,s}P_{s}^{r},
\]
where the nonnegative weights $w_{t,1},\dots,w_{t,t}$ sum to $1$
and have maximum $w_{t,\max}=\max\{(1-\lambda)^{t-1},\lambda\}$.
The super-uniformity condition (\ref{eq:superuniformity}) implies
that $P_{1},\dots,P_{t}$ are all $p$-values under $\mathrm{P}_{0}$.

If $r\ge1$, then by Lemma \ref{lem:p-value=00003D00003D00003D00003D000020averaging=00003D00003D00003D00003D000020Lemma}(2),
applied to $P_{1},\dots,P_{t}$ with weights $w_{t,1},\dots,w_{t,t}$,
\[
q_{t}=\Bigl(\min\{1+r,w_{t,\max}^{-1}\}\Bigr)^{1/r}\Bigl(\sum_{s=1}^{t}w_{t,s}P_{s}^{r}\Bigr)^{1/r}=Q_{\lambda,t}^{\left(r\right)}
\]
is a valid $p$-value under $\mathrm{P}_{0}$. If $r\in(-1,1)\setminus\{0\}$,
then Lemma \ref{lem:p-value=00003D00003D00003D00003D000020averaging=00003D00003D00003D00003D000020Lemma}(1)
yields 
\[
q_{t}=\left(1+r\right)^{1/r}\Bigl(\sum_{s=1}^{t}w_{t,s}P_{s}^{r}\Bigr)^{1/r}=Q_{\lambda,t}^{\left(r\right)},
\]
which is again a valid $p$-value under $\mathrm{P}_{0}$. Therefore,
in either case 
\[
\mathrm{P}_{0}\left(Q_{\lambda,t}^{\left(r\right)}\le\alpha\right)=\mathrm{P}_{0}\left(q_{t}\le\alpha\right)\le\alpha
\]
for every $\alpha\in[0,1]$. Since $t$ was arbitrary, the result
follows.

\subsubsection*{Proof of Proposition \ref{prop:EWMA-like-charts-conditional}}

Fix $\lambda\in(0,1)$, $r\ge1$ and $t\ge1$. By construction, $S_{\lambda,t-1}^{\left(r\right)}$
is a measurable function of $(P_{1},\dots,P_{t-1})$, hence $\left(1-\lambda\right)S_{\lambda,t-1}^{\left(r\right)}$
is ${\cal F}_{t-1}$-measurable and nonnegative. Therefore 
\[
\bar{Q}_{\lambda,t}^{\left(r\right)}=\lambda^{-1/r}\Bigl\{\lambda P_{t}^{r}+\left(1-\lambda\right)S_{\lambda,t-1}^{\left(r\right)}\Bigr\}^{1/r}\ge\lambda^{-1/r}\bigl(\lambda P_{t}^{r}\bigr)^{1/r}=P_{t},
\]
so that 
\[
\left\{ \bar{Q}_{\lambda,t}^{\left(r\right)}\le\alpha\right\} \subseteq\left\{ P_{t}\le\alpha\right\} \text{, }\mathrm{a.s.}
\]
Taking conditional probabilities given ${\cal F}_{t-1}$ and using
(\ref{eq:Conditional_Super_Uni}) yields 
\[
\mathrm{P}_{0}\left(\bar{Q}_{\lambda,t}^{\left(r\right)}\le\alpha\,\big|\,{\cal F}_{t-1}\right)\le\mathrm{P}_{0}\left(P_{t}\le\alpha\,\big|\,{\cal F}_{t-1}\right)\le\alpha\text{, }\mathrm{a.s.},
\]
which proves the claim.

\subsubsection*{Proof of Proposition \ref{prop:PDF_of_Ut}}

To obtain the probability density function (PDF) of $\tilde{U}_{\lambda,t}$,
observe that for each $t\ge1$, with $a_{t,s}=\lambda\left(1-\lambda\right)^{t-s}$,
we can write 
\[
\tilde{U}_{\lambda,t}=\left(1-\lambda\right)^{t}u_{0}+\sum_{s=1}^{t}a_{t,s}U_{s}\text{.}
\]
Thus, $\tilde{U}_{\lambda,t}=\left(1-\lambda\right)^{t}u_{0}+W_{t}$,
where $W_{t}=\sum_{s=1}^{t}a_{t,s}U_{s}$ and $U_{s}\sim\mathrm{Unif}\left(0,1\right)$
are IID, so that $a_{t,s}U_{s}\sim\mathrm{Unif}\left(0,a_{t,s}\right)$.
Let 
\[
\mathrm{F}_{t,u_{0}}\left(u\right)=\mathrm{P}_{0}\left(\tilde{U}_{\lambda,t}\le u\right)=\mathrm{P}_{0}\left(W_{t}\le u-\left(1-\lambda\right)^{t}u_{0}\right)
\]
be the CDF of $\tilde{U}_{\lambda,t}$, and note that $\mathrm{F}_{t,u_{0}}\left(u\right)=\mathrm{G}_{t}\left(u-\left(1-\lambda\right)^{t}u_{0}\right)$,
where $\mathrm{G}_{t}$ is the CDF of $W_{t}$. We write the respective
PDFs of $\tilde{U}_{\lambda,t}$ and $W_{t}$ as $f_{t,u_{0}}$ and
$g_{t}$, respectively, and note that $f_{t,u_{0}}\left(u\right)=g_{t}\left(u-\left(1-\lambda\right)^{t}u_{0}\right)$.

For $t=1$, observe that $W_{1}=a_{t,1}U_{1}\sim\mathrm{Unif}\left(0,a_{t,1}\right)$,
thus 
\[
g_{1}\left(w\right)=\frac{1}{a_{t,1}}\mathbf{1}_{\left[0,a_{t,1}\right]}\left(w\right)\text{, and }\mathrm{G}_{1}\left(w\right)=\frac{\left[w\right]_{+}-\left[w-a_{t,1}\right]_{+}}{a_{t,1}}\text{.}
\]
When $t\ge2$, define $\tilde{W}_{t}=a_{t,t}U_{t}\sim\mathrm{Unif}\left(0,a_{t,t}\right)$,
independent of $W_{t-1}^{(t)}=\sum_{s=1}^{t-1}a_{t,s}U_{s}$. Then
\[
W_{t}=W_{t-1}^{(t)}+\tilde{W}_{t},
\]
and the convolution identity for the PDF of sums of independent random
variables yields 
\[
g_{t}\left(w\right)=\int_{\mathbb{R}}g_{t-1}^{(t)}\left(w-\tilde{w}\right)h_{t}\left(\tilde{w}\right)\,\mathrm{d}\tilde{w}\text{,}
\]
where $h_{t}\left(\tilde{w}\right)=a_{t,t}^{-1}\mathbf{1}_{\left[0,a_{t,t}\right]}\left(\tilde{w}\right)$
is the PDF of $\tilde{W}_{t}$ and $g_{t-1}^{(t)}$ denotes the PDF
of $W_{t-1}^{(t)}$. Hence 
\begin{equation}
g_{t}\left(w\right)=\frac{1}{a_{t,t}}\int_{0}^{a_{t,t}}g_{t-1}^{(t)}\left(w-\tilde{w}\right)\,\mathrm{d}\tilde{w}\text{.}\label{eq:convolution_identity_correct}
\end{equation}
We need the following lemma. 
\begin{lem}
\label{lem:TP_lemma-1}For each $m\in\mathbb{N}$, $a>0$, and $x\in\mathbb{R}$,
\[
\int_{0}^{a}\left[x-y\right]_{+}^{m-1}\,\mathrm{d}y=\frac{\left[x\right]_{+}^{m}-\left[x-a\right]_{+}^{m}}{m}\text{.}
\]
\end{lem}

\begin{proof}
If $x\le0$, then both sides are equal to zero. If $0<x<a$, then
\[
\int_{0}^{a}\left[x-y\right]_{+}^{m-1}\,\mathrm{d}y=\int_{0}^{x}\left(x-y\right)^{m-1}\,\mathrm{d}y=\left[-\frac{\left(x-y\right)^{m}}{m}\right]_{y=0}^{y=x}=\frac{x^{m}}{m}\text{.}
\]
Finally, if $x\ge a$, then 
\[
\int_{0}^{a}\left[x-y\right]_{+}^{m-1}\,\mathrm{d}y=\int_{0}^{a}\left(x-y\right)^{m-1}\,\mathrm{d}y=\left[-\frac{\left(x-y\right)^{m}}{m}\right]_{y=0}^{y=a}=\frac{x^{m}-\left(x-a\right)^{m}}{m}\text{.}
\]
The result follows by combining the cases. 
\end{proof}
We proceed by induction on $t$. The case $t=1$ has been verified.
For $t-1\ge1$, make the inductive hypothesis 
\[
g_{t-1}^{(t)}\left(w\right)=\frac{1}{\left(t-2\right)!\prod_{s=1}^{t-1}a_{t,s}}\sum_{\mathbb{S}\subseteq\left[t-1\right]}\left(-1\right)^{\left|\mathbb{S}\right|}\left[w-a_{\mathbb{S}}\right]_{+}^{t-2}\text{,}
\]
where $a_{\mathbb{S}}=\sum_{s\in\mathbb{S}}a_{t,s}$. Then (\ref{eq:convolution_identity_correct})
and Lemma \ref{lem:TP_lemma-1} give 
\begin{align*}
g_{t}\left(w\right) & =\frac{1}{a_{t,t}}\int_{0}^{a_{t,t}}\frac{1}{\left(t-2\right)!\prod_{s=1}^{t-1}a_{t,s}}\sum_{\mathbb{S}\subseteq\left[t-1\right]}\left(-1\right)^{\left|\mathbb{S}\right|}\left[w-\tilde{w}-a_{\mathbb{S}}\right]_{+}^{t-2}\,\mathrm{d}\tilde{w}\\
 & =\frac{1}{\left(t-1\right)!\prod_{s=1}^{t}a_{t,s}}\sum_{\mathbb{S}\subseteq\left[t-1\right]}\left(-1\right)^{\left|\mathbb{S}\right|}\left(\left[w-a_{\mathbb{S}}\right]_{+}^{t-1}-\left[w-a_{\mathbb{S}}-a_{t,t}\right]_{+}^{t-1}\right)\\
 & =\frac{1}{\left(t-1\right)!\prod_{s=1}^{t}a_{t,s}}\sum_{\mathbb{S}\subseteq\left[t\right]}\left(-1\right)^{\left|\mathbb{S}\right|}\left[w-a_{\mathbb{S}}\right]_{+}^{t-1}\text{.}
\end{align*}
Therefore, 
\begin{align*}
f_{t,u_{0}}\left(u\right) & =g_{t}\left(u-\left(1-\lambda\right)^{t}u_{0}\right)\\
 & =\frac{1}{\left(t-1\right)!\prod_{s=1}^{t}a_{t,s}}\sum_{\mathbb{S}\subseteq\left[t\right]}\left(-1\right)^{\left|\mathbb{S}\right|}\left[u-\left(1-\lambda\right)^{t}u_{0}-a_{\mathbb{S}}\right]_{+}^{t-1}\text{,}
\end{align*}
as claimed. Finally, for any $a\in\mathbb{R}$ and $m\ge1$, $(\mathrm{d}/\mathrm{d}u)\left[u-a\right]_{+}^{m}=m\left[u-a\right]_{+}^{m-1}$,
so the fundamental theorem of calculus yields 
\[
\int_{-\infty}^{u}\left[x-a\right]_{+}^{t-1}\,\mathrm{d}x=\frac{\left[u-a\right]_{+}^{t}}{t}\text{.}
\]
Applying this identity termwise to $f_{t,u_{0}}$ gives the CDF 
\[
\mathrm{F}_{t,u_{0}}\left(u\right)=\frac{1}{t!\prod_{s=1}^{t}a_{t,s}}\sum_{\mathbb{S}\subseteq\left[t\right]}\left(-1\right)^{\left|\mathbb{S}\right|}\left[u-\left(1-\lambda\right)^{t}u_{0}-a_{\mathbb{S}}\right]_{+}^{t}\text{,}
\]
which completes the proof.

\subsubsection*{Proof of Proposition \ref{prop:Ut_left_super_unif}}

Recall that, with $a_{t,s}=\lambda(1-\lambda)^{t-s}$, 
\[
\tilde{U}_{\lambda,t}=(1-\lambda)^{t}u_{0}+\sum_{s=1}^{t}a_{t,s}U_{s}\text{.}
\]
Define the centred sum $\tilde{U}_{\lambda,t}^{\mathrm{cent}}=\sum_{s=1}^{t}\{a_{t,s}U_{s}-a_{t,s}/2\}$
and $c_{t}(u_{0})=1/2+(1-\lambda)^{t}(u_{0}-1/2)$, so that 
\[
\tilde{U}_{\lambda,t}=c_{t}(u_{0})+\tilde{U}_{\lambda,t}^{\mathrm{cent}}\text{.}
\]
Since each $a_{t,s}U_{s}-a_{t,s}/2\sim\mathrm{Unif}\left(-a_{t,s}/2,a_{t,s}/2\right)$,
\[
\mathrm{supp}\bigl(\tilde{U}_{\lambda,t}^{\mathrm{cent}}\bigr)=\bigl[-\{1-(1-\lambda)^{t}\}/2,\;\{1-(1-\lambda)^{t}\}/2\bigr]\text{,}
\]
\[
\mathrm{supp}\bigl(\tilde{U}_{\lambda,t}\bigr)=\bigl[(1-\lambda)^{t}u_{0},1-(1-\lambda)^{t}(1-u_{0})\bigr]\text{.}
\]

For each $t\in\mathbb{N}$, let $f_{t}^{\mathrm{cent}}$ denote the
PDF of $\tilde{U}_{\lambda,t}^{\mathrm{cent}}$. To establish that
$f_{t}^{\mathrm{cent}}$ is even and nonincreasing on $[0,\infty)$,
fix $t$ and define the partial centred sums, with the fixed weights
$a_{t,1},\dots,a_{t,t}$,
\[
\tilde{U}_{\lambda,t,m}^{\mathrm{cent}}=\sum_{s=1}^{m}\{a_{t,s}U_{s}-a_{t,s}/2\}\text{,}\qquad m=1,\dots,t\text{,}
\]
so that $\tilde{U}_{\lambda,t,t}^{\mathrm{cent}}=\tilde{U}_{\lambda,t}^{\mathrm{cent}}$.
Let $f_{t,m}^{\mathrm{cent}}$ denote the PDF of $\tilde{U}_{\lambda,t,m}^{\mathrm{cent}}$.
For $m=1$, $f_{t,1}^{\mathrm{cent}}=1/a_{t,1}$ on $\mathrm{supp}(\tilde{U}_{\lambda,t,1}^{\mathrm{cent}})=[-a_{t,1}/2,a_{t,1}/2]$.
More generally, $\tilde{U}_{\lambda,t,m}^{\mathrm{cent}}=\tilde{U}_{\lambda,t,m-1}^{\mathrm{cent}}+V_{t,m}$,
where $V_{t,m}\sim\mathrm{Unif}\left(-a_{t,m}/2,a_{t,m}/2\right)$
is independent of $\tilde{U}_{\lambda,t,m-1}^{\mathrm{cent}}$. The
convolution identity gives 
\[
f_{t,m}^{\mathrm{cent}}(u)=\frac{1}{a_{t,m}}\int_{-a_{t,m}/2}^{a_{t,m}/2}f_{t,m-1}^{\mathrm{cent}}(u-v)\,\mathrm{d}v\text{.}
\]
Differentiating under the integral (Leibniz' rule) yields 
\[
\frac{\mathrm{d}}{\mathrm{d}u}f_{t,m}^{\mathrm{cent}}(u)=\frac{1}{a_{t,m}}\Bigl(f_{t,m-1}^{\mathrm{cent}}(u+a_{t,m}/2)-f_{t,m-1}^{\mathrm{cent}}(u-a_{t,m}/2)\Bigr)\text{.}
\]
If $f_{t,m-1}^{\mathrm{cent}}$ is nonincreasing on $[0,\infty)$
and even, then $f_{t,m}^{\mathrm{cent}}$ is nonincreasing on $[0,\infty)$:
i.e., when $u\ge a_{t,m}/2$, $u\pm a_{t,m}/2\ge0$ and nonincreasingness
gives $f_{t,m-1}^{\mathrm{cent}}(u+a_{t,m}/2)\le f_{t,m-1}^{\mathrm{cent}}(u-a_{t,m}/2)$;
when $u\in[0,a_{t,m}/2]$, evenness gives $f_{t,m-1}^{\mathrm{cent}}(u-a_{t,m}/2)=f_{t,m-1}^{\mathrm{cent}}(a_{t,m}/2-u)$
with $a_{t,m}/2-u\ge0$; and then $f_{t,m-1}^{\mathrm{cent}}(u+a_{t,m}/2)\le f_{t,m-1}^{\mathrm{cent}}(a_{t,m}/2-u)=f_{t,m-1}^{\mathrm{cent}}(u-a_{t,m}/2)$.
Moreover, evenness is preserved under convolution with the symmetric
uniform kernel, so $f_{t,m}^{\mathrm{cent}}$ is even whenever $f_{t,m-1}^{\mathrm{cent}}$
is even. Since $f_{t,1}^{\mathrm{cent}}$ is even and nonincreasing
on $[0,\infty)$, induction implies that $f_{t,t}^{\mathrm{cent}}$
is even and nonincreasing on $[0,\infty)$. Since $f_{t}^{\mathrm{cent}}$
is the PDF of $\tilde{U}_{\lambda,t}^{\mathrm{cent}}=\tilde{U}_{\lambda,t,t}^{\mathrm{cent}}$,
we conclude that $f_{t}^{\mathrm{cent}}$ is even and nonincreasing
on $[0,\infty)$.

If $g\ge0$ is nonincreasing on $[0,\bar{d}]$, then for each $d\in[0,\bar{d}]$,
\[
\frac{1}{d}\int_{0}^{d}g(u)\,\mathrm{d}u\ge\frac{1}{\bar{d}}\int_{0}^{\bar{d}}g(u)\,\mathrm{d}u\text{.}
\]
Applying this to $g=f_{t}^{\mathrm{cent}}$ with $\bar{d}=\{1-(1-\lambda)^{t}\}/2$
and using evenness, for each $d\le\{1-(1-\lambda)^{t}\}/2$, 
\[
\int_{0}^{d}f_{t}^{\mathrm{cent}}(u)\,\mathrm{d}u\ge\frac{d}{\bar{d}}\int_{0}^{\bar{d}}f_{t}^{\mathrm{cent}}(u)\,\mathrm{d}u=\frac{2d}{1-(1-\lambda)^{t}}\int_{0}^{\{1-(1-\lambda)^{t}\}/2}f_{t}^{\mathrm{cent}}(u)\,\mathrm{d}u=\frac{d}{1-(1-\lambda)^{t}}\text{,}
\]
so 
\[
\mathrm{P}_{0}\bigl(\tilde{U}_{\lambda,t}^{\mathrm{cent}}\le-d\bigr)=\frac{1}{2}-\int_{0}^{d}f_{t}^{\mathrm{cent}}(u)\,\mathrm{d}u\le\frac{1}{2}-\frac{d}{1-(1-\lambda)^{t}}\text{.}
\]

Fix $\alpha\in[0,1/2]$ and $u_{0}\ge1/2$, and set 
\[
d=c_{t}(u_{0})-\alpha=\bigl(1/2-\alpha\bigr)+(1-\lambda)^{t}\bigl(u_{0}-1/2\bigr)\ge0\text{.}
\]
Then $\mathrm{P}_{0}(\tilde{U}_{\lambda,t}\le\alpha)=\mathrm{P}_{0}(\tilde{U}_{\lambda,t}^{\mathrm{cent}}\le-d)$.
If $d>\{1-(1-\lambda)^{t}\}/2$, the probability is $0\le\alpha$
and we are done. Otherwise apply the bound above and use $1-(1-\lambda)^{t}\le1$:
\[
\mathrm{P}_{0}\bigl(\tilde{U}_{\lambda,t}\le\alpha\bigr)\le\frac{1}{2}-\frac{d}{\,1-(1-\lambda)^{t}\,}\le\frac{1}{2}-d\le\alpha\text{,}
\]
as required.

\subsubsection*{Proof of Proposition \ref{prop:Directional_correct_FWER}}

We work with fixed time $t$. Recall that $\mathrm{P}_{0}$ be a probability
measure under which the IC hypothesis $\mathrm{H}_{0}:\theta=\theta_{0}$
holds at time $t$. Then, for every $j\in[d]$ and $\square\in\{\le,\ge\}$,
the directional null $\mathrm{H}_{0}^{(j,\square)}$ is true, since
$\theta_{j}=\theta_{0,j}$ implies $\theta_{0,j}\ \square\ \theta_{j}$.
Consequently, 
\[
\bigl\{\exists s\in\mathbb{D}_{t}:\ \theta_{0,j}\ \square\ \theta_{j}\bigr\}=\bigl\{\mathbb{D}_{t}\neq\varnothing\bigr\}\text{.}
\]
By (\ref{eq:Union_bound}), for each $j\in[d]$ the two-sided $p$-value
$P_{t}^{j}$ is super-uniform under the equality hypothesis $\mathrm{H}_{0}^{j}:\theta_{j}=\theta_{0,j}$.
Applying Holm's step-down procedure at level $\alpha$ to $\{P_{t}^{j}\}_{j\in[d]}$
produces the rejection set $\mathbb{J}_{t}$. By Theorems \ref{thm:closure}
and \ref{thm:Holm's}, Holm's procedure has strong FWER control at
level $\alpha$ for the family $\{\mathrm{H}_{0}^{j}\}_{j\in[d]}$.
Since $\mathrm{H}_{0}^{j}$ is true for every $j\in[d]$ under $\mathrm{P}_{0}$,
we have 
\[
\mathrm{P}_{0}\bigl(\mathbb{J}_{t}\neq\varnothing\bigr)\le\alpha\text{.}
\]
By the reporting rule in Section~\ref{subsec:Directional chart},
the set $\mathbb{D}_{t}$ is constructed from $\mathbb{J}_{t}$ by
assigning, for each $j\in\mathbb{J}_{t}$, exactly one direction $\square\in\{\le,\ge\}$.
In particular, $\mathbb{D}_{t}\neq\varnothing$ implies $\mathbb{J}_{t}\neq\varnothing$,
therefore, 
\[
\mathrm{P}_{0}\bigl(\exists s\in\mathbb{D}_{t}:\ \theta_{0,j}\ \square\ \theta_{j}\bigr)\le\mathrm{P}_{0}\bigl(\mathbb{J}_{t}\neq\varnothing\bigr)\le\alpha\text{,}
\]
which is (\ref{eq:FWER_control-1}). 

\subsection*{Additional numerical results}

We provide further numerical results from Section \ref{sec:Example-applications}.
Tables \ref{tab:KS_IC_EWMA_sim_part2}, \ref{tab:KS_OC_EWMA_sim_part2},
\ref{tab:KS_OC_EWMA_sim_part3}, \ref{tab:KS_OC_EWMA_sim_dynamic_part2}
and \ref{tab:KS_OC_EWMA_sim_dynamic_part3} provide additional results
for Section \ref{subsec:Kolmogorov=00003D002013Smirnov-based-EWMA-ch}.
Figure \ref{fig:Plots-of-CDFs} provides additional visualisation
for Section \ref{subsec:Uniform_EWMA_sim}. Tables \ref{tab:Cauchy_20}
and \ref{tab:Cauchy_100} provide additional results for Section \ref{subsec:Directional-and-coordinate-localisation-charts}. 

\begin{table}
\caption{Average $R_{k}$ while IC versus $k$-ARL bounds from Proposition
\ref{prop:kARL_superuniform} for KS EWMA-like $p$-value charts $\left(\tilde{Q}_{\lambda,t}^{\left(r\right)}\right)_{t}$
and $\left(\bar{Q}_{\lambda,t}^{\left(r\right)}\right)_{t}$, using
100 simulation repetitions. This table contains results for $n_{0}\in\left\{ 100,200\right\} $,
while Table \ref{tab:KS_IC_EWMA_sim} contains results for $n_{0}=50$.}
\label{tab:KS_IC_EWMA_sim_part2}

\centering{}%
\begin{tabular}{cccccccccc}
\hline 
Chart  & $n_{0}$  & $\alpha$  & $k$  & $\lambda$  & $r$  & mean $R_{k}$  & st. err. $R_{k}$  & lower bound  & ratio\tabularnewline
\hline 
$\tilde{Q}_{\lambda,t}^{\left(r\right)}$  & 100  & 0.05  & 1  & 0.5  & -0.9  & 36300.60  & 8129.78  & 10.5  & 3457.2000\tabularnewline
 & 100  & 0.05  & 1  & 0.5  & -0.8  & 17899.07  & 4642.61  & 10.5  & 1704.6733\tabularnewline
 & 100  & 0.05  & 5  & 0.5  & -0.9  & 170624.82  & 37579.01  & 50.5  & 3378.7093\tabularnewline
 & 100  & 0.05  & 5  & 0.5  & -0.8  & 83157.22  & 16025.42  & 50.5  & 1646.6776\tabularnewline
 & 100  & 0.1  & 1  & 0.5  & -0.9  & 15112.70  & 3629.50  & 5.5  & 2747.7636\tabularnewline
 & 100  & 0.1  & 1  & 0.5  & -0.8  & 4915.30  & 1010.26  & 5.5  & 893.6909\tabularnewline
 & 100  & 0.1  & 5  & 0.5  & -0.9  & 38756.87  & 6254.42  & 25.5  & 1520.8576\tabularnewline
 & 100  & 0.1  & 5  & 0.5  & -0.8  & 19485.31  & 3483.18  & 25.5  & 764.1298\tabularnewline
 & 200  & 0.05  & 1  & 0.5  & -0.9  & 36438.61  & 12501.96  & 10.5  & 3470.3438\tabularnewline
 & 200  & 0.05  & 1  & 0.5  & -0.8  & 13673.33  & 3899.34  & 10.5  & 1302.2229\tabularnewline
 & 200  & 0.05  & 5  & 0.5  & -0.9  & 163773.92  & 40044.39  & 50.5  & 3243.0489\tabularnewline
 & 200  & 0.05  & 5  & 0.5  & -0.8  & 61254.49  & 9087.68  & 50.5  & 1212.9602\tabularnewline
 & 200  & 0.1  & 1  & 0.5  & -0.9  & 12153.63  & 3262.65  & 5.5  & 2209.751\tabularnewline
 & 200  & 0.1  & 1  & 0.5  & -0.8  & 3806.43  & 658.28  & 5.5  & 692.0782\tabularnewline
 & 200  & 0.1  & 5  & 0.5  & -0.9  & 34633.98  & 5038.05  & 25.5  & 1358.1953\tabularnewline
 & 200  & 0.1  & 5  & 0.5  & -0.8  & 11625.17  & 2128.99  & 25.5  & 455.8882\tabularnewline
\hline 
$\bar{Q}_{\lambda,t}^{\left(r\right)}$  & 100  & 0.05  & 1  & 0.9  & 1  & 2358.38  & 632.10  & 10.5  & 224.6076\tabularnewline
 & 100  & 0.05  & 1  & 0.95  & 1  & 584.18  & 110.88  & 10.5  & 55.6362\tabularnewline
 & 100  & 0.05  & 5  & 0.9  & 1  & 7574.50  & 1252.08  & 50.5  & 149.9901\tabularnewline
 & 100  & 0.05  & 5  & 0.95  & 1  & 2943.04  & 487.05  & 50.5  & 58.2770\tabularnewline
 & 100  & 0.1  & 1  & 0.9  & 1  & 234.15  & 47.66  & 5.5  & 42.5727\tabularnewline
 & 100  & 0.1  & 1  & 0.95  & 1  & 41.94  & 6.41  & 5.5  & 7.6255\tabularnewline
 & 100  & 0.1  & 5  & 0.9  & 1  & 765.87  & 125.16  & 25.5  & 30.0330\tabularnewline
 & 100  & 0.1  & 5  & 0.95  & 1  & 289.27  & 33.17  & 25.5  & 11.344\tabularnewline
 & 200  & 0.05  & 1  & 0.9  & 1  & 2229.21  & 607.40  & 10.5  & 212.3057\tabularnewline
 & 200  & 0.05  & 1  & 0.95  & 1  & 311.43  & 45.42  & 10.5  & 29.6590\tabularnewline
 & 200  & 0.05  & 5  & 0.9  & 1  & 9698.15  & 2254.83  & 50.5  & 192.0426\tabularnewline
 & 200  & 0.05  & 5  & 0.95  & 1  & 1703.18  & 230.54  & 50.5  & 33.7263\tabularnewline
 & 200  & 0.1  & 1  & 0.9  & 1  & 105.28  & 20.25  & 5.5  & 19.1418\tabularnewline
 & 200  & 0.1  & 1  & 0.95  & 1  & 49.17  & 8.48  & 5.5  & 8.9400\tabularnewline
 & 200  & 0.1  & 5  & 0.9  & 1  & 785.69  & 108.18  & 25.5  & 30.8114\tabularnewline
 & 200  & 0.1  & 5  & 0.95  & 1  & 218.23  & 24.79  & 25.5  & 8.5580\tabularnewline
\hline 
\end{tabular}
\end{table}

\begin{table}
\caption{Average $R_{k}$ under persistent OOC processes for KS $p$-value
charts $\left(P_{t}\right)_{t}$, $\left(\tilde{Q}_{\lambda,t}^{\left(r\right)}\right)_{t}$
and $\left(\bar{Q}_{\lambda,t}^{\left(r\right)}\right)_{t}$, using
100 simulation repetitions. This table contains results for $n_{0}=100$,
while Appendix Tables \ref{tab:KS_OC_EWMA_sim} and \ref{tab:KS_OC_EWMA_sim_part3}
contains results for $n_{0}=50$ and $n_{0}=200$, respectively.}\label{tab:KS_OC_EWMA_sim_part2}

\centering{}{\footnotesize{}%
\begin{tabular}{cccccccccc}
\hline 
{\footnotesize OOC process} & {\footnotesize$n_{0}$} & {\footnotesize$\alpha$} & {\footnotesize Chart} & {\footnotesize$r$} & {\footnotesize$\lambda$} & {\footnotesize mean $R_{1}$} & {\footnotesize st. err. $R_{1}$} & {\footnotesize mean $R_{5}$} & {\footnotesize st. err. $R_{5}$}\tabularnewline
\hline 
{\footnotesize$\text{N}\left(1/2,1\right)$} & {\footnotesize 100} & {\footnotesize 0.01} & {\footnotesize$P_{t}$} &  &  & {\footnotesize 2.27} & {\footnotesize 0.57} & {\footnotesize 10.28} & {\footnotesize 1.23}\tabularnewline
 & {\footnotesize 100} & {\footnotesize 0.01} & {\footnotesize$\tilde{Q}_{\lambda,t}^{\left(r\right)}$} & {\footnotesize -0.9} & {\footnotesize 0.5} & {\footnotesize 23.65} & {\footnotesize 6.93} & {\footnotesize 52.64} & {\footnotesize 12.82}\tabularnewline
 & {\footnotesize 100} & {\footnotesize 0.01} & {\footnotesize$\tilde{Q}_{\lambda,t}^{\left(r\right)}$} & {\footnotesize -0.8} & {\footnotesize 0.5} & {\footnotesize 13.49} & {\footnotesize 6.07} & {\footnotesize 56.24} & {\footnotesize 28.13}\tabularnewline
 & {\footnotesize 100} & {\footnotesize 0.01} & {\footnotesize$\bar{Q}_{\lambda,t}^{\left(r\right)}$} & {\footnotesize 1} & {\footnotesize 0.9} & {\footnotesize 9.46} & {\footnotesize 4.16} & {\footnotesize 49.26} & {\footnotesize 27.41}\tabularnewline
 & {\footnotesize 100} & {\footnotesize 0.01} & {\footnotesize$\bar{Q}_{\lambda,t}^{\left(r\right)}$} & {\footnotesize 1} & {\footnotesize 0.95} & {\footnotesize 3.63} & {\footnotesize 0.74} & {\footnotesize 25.29} & {\footnotesize 7.53}\tabularnewline
 & {\footnotesize 100} & {\footnotesize 0.05} & {\footnotesize$P_{t}$} &  &  & {\footnotesize 1.28} & {\footnotesize 0.08} & {\footnotesize 6.43} & {\footnotesize 0.30}\tabularnewline
 & {\footnotesize 100} & {\footnotesize 0.05} & {\footnotesize$\tilde{Q}_{\lambda,t}^{\left(r\right)}$} & {\footnotesize -0.9} & {\footnotesize 0.5} & {\footnotesize 6.71} & {\footnotesize 2.11} & {\footnotesize 28.66} & {\footnotesize 8.20}\tabularnewline
 & {\footnotesize 100} & {\footnotesize 0.05} & {\footnotesize$\tilde{Q}_{\lambda,t}^{\left(r\right)}$} & {\footnotesize -0.8} & {\footnotesize 0.5} & {\footnotesize 3.41} & {\footnotesize 0.84} & {\footnotesize 10.77} & {\footnotesize 1.98}\tabularnewline
 & {\footnotesize 100} & {\footnotesize 0.05} & {\footnotesize$\bar{Q}_{\lambda,t}^{\left(r\right)}$} & {\footnotesize 1} & {\footnotesize 0.9} & {\footnotesize 1.56} & {\footnotesize 0.13} & {\footnotesize 11.81} & {\footnotesize 3.24}\tabularnewline
 & {\footnotesize 100} & {\footnotesize 0.05} & {\footnotesize$\bar{Q}_{\lambda,t}^{\left(r\right)}$} & {\footnotesize 1} & {\footnotesize 0.95} & {\footnotesize 1.28} & {\footnotesize 0.08} & {\footnotesize 6.89} & {\footnotesize 0.50}\tabularnewline
{\footnotesize$\text{N}(1,1)$} & {\footnotesize 100} & {\footnotesize 0.01} & {\footnotesize$P_{t}$} &  &  & {\footnotesize 1.00} & {\footnotesize 0.00} & {\footnotesize 5.00} & {\footnotesize 0.00}\tabularnewline
 & {\footnotesize 100} & {\footnotesize 0.01} & {\footnotesize$\tilde{Q}_{\lambda,t}^{\left(r\right)}$} & {\footnotesize -0.9} & {\footnotesize 0.5} & {\footnotesize 1.00} & {\footnotesize 0.00} & {\footnotesize 5.00} & {\footnotesize 0.00}\tabularnewline
 & {\footnotesize 100} & {\footnotesize 0.01} & {\footnotesize$\tilde{Q}_{\lambda,t}^{\left(r\right)}$} & {\footnotesize -0.8} & {\footnotesize 0.5} & {\footnotesize 1.01} & {\footnotesize 0.01} & {\footnotesize 5.01} & {\footnotesize 0.01}\tabularnewline
 & {\footnotesize 100} & {\footnotesize 0.01} & {\footnotesize$\bar{Q}_{\lambda,t}^{\left(r\right)}$} & {\footnotesize 1} & {\footnotesize 0.9} & {\footnotesize 1.00} & {\footnotesize 0.00} & {\footnotesize 5.00} & {\footnotesize 0.00}\tabularnewline
 & {\footnotesize 100} & {\footnotesize 0.01} & {\footnotesize$\bar{Q}_{\lambda,t}^{\left(r\right)}$} & {\footnotesize 1} & {\footnotesize 0.95} & {\footnotesize 1.00} & {\footnotesize 0.00} & {\footnotesize 5.00} & {\footnotesize 0.00}\tabularnewline
 & {\footnotesize 100} & {\footnotesize 0.05} & {\footnotesize$P_{t}$} &  &  & {\footnotesize 1.00} & {\footnotesize 0.00} & {\footnotesize 5.00} & {\footnotesize 0.00}\tabularnewline
 & {\footnotesize 100} & {\footnotesize 0.05} & {\footnotesize$\tilde{Q}_{\lambda,t}^{\left(r\right)}$} & {\footnotesize -0.9} & {\footnotesize 0.5} & {\footnotesize 1.00} & {\footnotesize 0.00} & {\footnotesize 5.00} & {\footnotesize 0.00}\tabularnewline
 & {\footnotesize 100} & {\footnotesize 0.05} & {\footnotesize$\tilde{Q}_{\lambda,t}^{\left(r\right)}$} & {\footnotesize -0.8} & {\footnotesize 0.5} & {\footnotesize 1.00} & {\footnotesize 0.00} & {\footnotesize 5.00} & {\footnotesize 0.00}\tabularnewline
 & {\footnotesize 100} & {\footnotesize 0.05} & {\footnotesize$\bar{Q}_{\lambda,t}^{\left(r\right)}$} & {\footnotesize 1} & {\footnotesize 0.9} & {\footnotesize 1.00} & {\footnotesize 0.00} & {\footnotesize 5.00} & {\footnotesize 0.00}\tabularnewline
 & {\footnotesize 100} & {\footnotesize 0.05} & {\footnotesize$\bar{Q}_{\lambda,t}^{\left(r\right)}$} & {\footnotesize 1} & {\footnotesize 0.95} & {\footnotesize 1.00} & {\footnotesize 0.00} & {\footnotesize 5.00} & {\footnotesize 0.00}\tabularnewline
{\footnotesize$\text{N}(0,2)$} & {\footnotesize 100} & {\footnotesize 0.01} & {\footnotesize$P_{t}$} &  &  & {\footnotesize 73.99} & {\footnotesize 15.95} & {\footnotesize 328.15} & {\footnotesize 52.35}\tabularnewline
 & {\footnotesize 100} & {\footnotesize 0.01} & {\footnotesize$\tilde{Q}_{\lambda,t}^{\left(r\right)}$} & {\footnotesize -0.9} & {\footnotesize 0.5} & {\footnotesize 6354.41} & {\footnotesize 1947.68} & {\footnotesize 21834.59} & {\footnotesize 6659.35}\tabularnewline
 & {\footnotesize 100} & {\footnotesize 0.01} & {\footnotesize$\tilde{Q}_{\lambda,t}^{\left(r\right)}$} & {\footnotesize -0.8} & {\footnotesize 0.5} & {\footnotesize 2904.52} & {\footnotesize 815.86} & {\footnotesize 18591.21} & {\footnotesize 4835.06}\tabularnewline
 & {\footnotesize 100} & {\footnotesize 0.01} & {\footnotesize$\bar{Q}_{\lambda,t}^{\left(r\right)}$} & {\footnotesize 1} & {\footnotesize 0.9} & {\footnotesize 12367.56} & {\footnotesize 9731.56} & {\footnotesize 49943.42} & {\footnotesize 35538.48}\tabularnewline
 & {\footnotesize 100} & {\footnotesize 0.01} & {\footnotesize$\bar{Q}_{\lambda,t}^{\left(r\right)}$} & {\footnotesize 1} & {\footnotesize 0.95} & {\footnotesize 607.67} & {\footnotesize 234.54} & {\footnotesize 2346.59} & {\footnotesize 477.43}\tabularnewline
 & {\footnotesize 100} & {\footnotesize 0.05} & {\footnotesize$P_{t}$} &  &  & {\footnotesize 9.25} & {\footnotesize 1.43} & {\footnotesize 43.07} & {\footnotesize 4.30}\tabularnewline
 & {\footnotesize 100} & {\footnotesize 0.05} & {\footnotesize$\tilde{Q}_{\lambda,t}^{\left(r\right)}$} & {\footnotesize -0.9} & {\footnotesize 0.5} & {\footnotesize 797.49} & {\footnotesize 258.57} & {\footnotesize 3114.35} & {\footnotesize 731.10}\tabularnewline
 & {\footnotesize 100} & {\footnotesize 0.05} & {\footnotesize$\tilde{Q}_{\lambda,t}^{\left(r\right)}$} & {\footnotesize -0.8} & {\footnotesize 0.5} & {\footnotesize 330.37} & {\footnotesize 74.80} & {\footnotesize 1389.18} & {\footnotesize 274.72}\tabularnewline
 & {\footnotesize 100} & {\footnotesize 0.05} & {\footnotesize$\bar{Q}_{\lambda,t}^{\left(r\right)}$} & {\footnotesize 1} & {\footnotesize 0.9} & {\footnotesize 41.69} & {\footnotesize 16.66} & {\footnotesize 152.71} & {\footnotesize 38.84}\tabularnewline
 & {\footnotesize 100} & {\footnotesize 0.05} & {\footnotesize$\bar{Q}_{\lambda,t}^{\left(r\right)}$} & {\footnotesize 1} & {\footnotesize 0.95} & {\footnotesize 11.66} & {\footnotesize 1.93} & {\footnotesize 59.13} & {\footnotesize 7.83}\tabularnewline
{\footnotesize Cauchy} & {\footnotesize 100} & {\footnotesize 0.01} & {\footnotesize$P_{t}$} &  &  & {\footnotesize 15.72} & {\footnotesize 1.88} & {\footnotesize 82.13} & {\footnotesize 7.66}\tabularnewline
 & {\footnotesize 100} & {\footnotesize 0.01} & {\footnotesize$\tilde{Q}_{\lambda,t}^{\left(r\right)}$} & {\footnotesize -0.9} & {\footnotesize 0.5} & {\footnotesize 1272.56} & {\footnotesize 305.43} & {\footnotesize 3792.39} & {\footnotesize 672.21}\tabularnewline
 & {\footnotesize 100} & {\footnotesize 0.01} & {\footnotesize$\tilde{Q}_{\lambda,t}^{\left(r\right)}$} & {\footnotesize -0.8} & {\footnotesize 0.5} & {\footnotesize 585.77} & {\footnotesize 115.57} & {\footnotesize 2275.66} & {\footnotesize 384.01}\tabularnewline
 & {\footnotesize 100} & {\footnotesize 0.01} & {\footnotesize$\bar{Q}_{\lambda,t}^{\left(r\right)}$} & {\footnotesize 1} & {\footnotesize 0.9} & {\footnotesize 147.44} & {\footnotesize 37.90} & {\footnotesize 698.28} & {\footnotesize 205.49}\tabularnewline
 & {\footnotesize 100} & {\footnotesize 0.01} & {\footnotesize$\bar{Q}_{\lambda,t}^{\left(r\right)}$} & {\footnotesize 1} & {\footnotesize 0.95} & {\footnotesize 53.31} & {\footnotesize 9.13} & {\footnotesize 238.25} & {\footnotesize 30.51}\tabularnewline
 & {\footnotesize 100} & {\footnotesize 0.05} & {\footnotesize$P_{t}$} &  &  & {\footnotesize 3.52} & {\footnotesize 0.36} & {\footnotesize 16.04} & {\footnotesize 0.93}\tabularnewline
 & {\footnotesize 100} & {\footnotesize 0.05} & {\footnotesize$\tilde{Q}_{\lambda,t}^{\left(r\right)}$} & {\footnotesize -0.9} & {\footnotesize 0.5} & {\footnotesize 119.95} & {\footnotesize 16.23} & {\footnotesize 532.01} & {\footnotesize 92.31}\tabularnewline
 & {\footnotesize 100} & {\footnotesize 0.05} & {\footnotesize$\tilde{Q}_{\lambda,t}^{\left(r\right)}$} & {\footnotesize -0.8} & {\footnotesize 0.5} & {\footnotesize 96.66} & {\footnotesize 26.42} & {\footnotesize 294.28} & {\footnotesize 51.08}\tabularnewline
 & {\footnotesize 100} & {\footnotesize 0.05} & {\footnotesize$\bar{Q}_{\lambda,t}^{\left(r\right)}$} & {\footnotesize 1} & {\footnotesize 0.9} & {\footnotesize 4.05} & {\footnotesize 0.49} & {\footnotesize 20.80} & {\footnotesize 1.76}\tabularnewline
 & {\footnotesize 100} & {\footnotesize 0.05} & {\footnotesize$\bar{Q}_{\lambda,t}^{\left(r\right)}$} & {\footnotesize 1} & {\footnotesize 0.95} & {\footnotesize 3.44} & {\footnotesize 0.37} & {\footnotesize 17.55} & {\footnotesize 1.43}\tabularnewline
\hline 
\end{tabular}}{\footnotesize\par}
\end{table}

\begin{table}
\caption{Average $R_{k}$ under persistent OOC processes for KS $p$-value
charts $\left(P_{t}\right)_{t}$, $\left(\tilde{Q}_{\lambda,t}^{\left(r\right)}\right)_{t}$
and $\left(\bar{Q}_{\lambda,t}^{\left(r\right)}\right)_{t}$, using
100 simulation repetitions. This table contains results for $n_{0}=200$,
while Appendix Tables \ref{tab:KS_OC_EWMA_sim} and \ref{tab:KS_OC_EWMA_sim_part2}
contains results for $n_{0}=50$ and $n_{0}=100$, respectively.}\label{tab:KS_OC_EWMA_sim_part3}

\centering{}{\footnotesize{}%
\begin{tabular}{cccccccccc}
\hline 
{\footnotesize OOC process} & {\footnotesize$n_{0}$} & {\footnotesize$\alpha$} & {\footnotesize Chart} & {\footnotesize$r$} & {\footnotesize$\lambda$} & {\footnotesize mean $R_{1}$} & {\footnotesize st. err. $R_{1}$} & {\footnotesize mean $R_{5}$} & {\footnotesize st. err. $R_{5}$}\tabularnewline
\hline 
{\footnotesize$\text{N}\left(1/2,1\right)$} & {\footnotesize 200} & {\footnotesize 0.01} & {\footnotesize$P_{t}$} &  &  & {\footnotesize 1.04} & {\footnotesize 0.02} & {\footnotesize 5.22} & {\footnotesize 0.07}\tabularnewline
 & {\footnotesize 200} & {\footnotesize 0.01} & {\footnotesize$\tilde{Q}_{\lambda,t}^{\left(r\right)}$} & {\footnotesize -0.9} & {\footnotesize 0.5} & {\footnotesize 1.68} & {\footnotesize 0.34} & {\footnotesize 7.36} & {\footnotesize 1.61}\tabularnewline
 & {\footnotesize 200} & {\footnotesize 0.01} & {\footnotesize$\tilde{Q}_{\lambda,t}^{\left(r\right)}$} & {\footnotesize -0.8} & {\footnotesize 0.5} & {\footnotesize 1.28} & {\footnotesize 0.09} & {\footnotesize 5.44} & {\footnotesize 0.13}\tabularnewline
 & {\footnotesize 200} & {\footnotesize 0.01} & {\footnotesize$\bar{Q}_{\lambda,t}^{\left(r\right)}$} & {\footnotesize 1} & {\footnotesize 0.9} & {\footnotesize 1.07} & {\footnotesize 0.04} & {\footnotesize 5.21} & {\footnotesize 0.06}\tabularnewline
 & {\footnotesize 200} & {\footnotesize 0.01} & {\footnotesize$\bar{Q}_{\lambda,t}^{\left(r\right)}$} & {\footnotesize 1} & {\footnotesize 0.95} & {\footnotesize 1.13} & {\footnotesize 0.06} & {\footnotesize 5.36} & {\footnotesize 0.11}\tabularnewline
 & {\footnotesize 200} & {\footnotesize 0.05} & {\footnotesize$P_{t}$} &  &  & {\footnotesize 1.00} & {\footnotesize 0.00} & {\footnotesize 5.01} & {\footnotesize 0.01}\tabularnewline
 & {\footnotesize 200} & {\footnotesize 0.05} & {\footnotesize$\tilde{Q}_{\lambda,t}^{\left(r\right)}$} & {\footnotesize -0.9} & {\footnotesize 0.5} & {\footnotesize 1.22} & {\footnotesize 0.07} & {\footnotesize 5.25} & {\footnotesize 0.08}\tabularnewline
 & {\footnotesize 200} & {\footnotesize 0.05} & {\footnotesize$\tilde{Q}_{\lambda,t}^{\left(r\right)}$} & {\footnotesize -0.8} & {\footnotesize 0.5} & {\footnotesize 1.03} & {\footnotesize 0.02} & {\footnotesize 5.06} & {\footnotesize 0.03}\tabularnewline
 & {\footnotesize 200} & {\footnotesize 0.05} & {\footnotesize$\bar{Q}_{\lambda,t}^{\left(r\right)}$} & {\footnotesize 1} & {\footnotesize 0.9} & {\footnotesize 1.00} & {\footnotesize 0.00} & {\footnotesize 5.06} & {\footnotesize 0.03}\tabularnewline
 & {\footnotesize 200} & {\footnotesize 0.05} & {\footnotesize$\bar{Q}_{\lambda,t}^{\left(r\right)}$} & {\footnotesize 1} & {\footnotesize 0.95} & {\footnotesize 1.01} & {\footnotesize 0.01} & {\footnotesize 5.06} & {\footnotesize 0.02}\tabularnewline
{\footnotesize$\text{N}(1,1)$} & {\footnotesize 200} & {\footnotesize 0.01} & {\footnotesize$P_{t}$} &  &  & {\footnotesize 1.00} & {\footnotesize 0.00} & {\footnotesize 5.00} & {\footnotesize 0.00}\tabularnewline
 & {\footnotesize 200} & {\footnotesize 0.01} & {\footnotesize$\tilde{Q}_{\lambda,t}^{\left(r\right)}$} & {\footnotesize -0.9} & {\footnotesize 0.5} & {\footnotesize 1.00} & {\footnotesize 0.00} & {\footnotesize 5.00} & {\footnotesize 0.00}\tabularnewline
 & {\footnotesize 200} & {\footnotesize 0.01} & {\footnotesize$\tilde{Q}_{\lambda,t}^{\left(r\right)}$} & {\footnotesize -0.8} & {\footnotesize 0.5} & {\footnotesize 1.00} & {\footnotesize 0.00} & {\footnotesize 5.00} & {\footnotesize 0.00}\tabularnewline
 & {\footnotesize 200} & {\footnotesize 0.01} & {\footnotesize$\bar{Q}_{\lambda,t}^{\left(r\right)}$} & {\footnotesize 1} & {\footnotesize 0.9} & {\footnotesize 1.00} & {\footnotesize 0.00} & {\footnotesize 5.00} & {\footnotesize 0.00}\tabularnewline
 & {\footnotesize 200} & {\footnotesize 0.01} & {\footnotesize$\bar{Q}_{\lambda,t}^{\left(r\right)}$} & {\footnotesize 1} & {\footnotesize 0.95} & {\footnotesize 1.00} & {\footnotesize 0.00} & {\footnotesize 5.00} & {\footnotesize 0.00}\tabularnewline
 & {\footnotesize 200} & {\footnotesize 0.05} & {\footnotesize$P_{t}$} &  &  & {\footnotesize 1.00} & {\footnotesize 0.00} & {\footnotesize 5.00} & {\footnotesize 0.00}\tabularnewline
 & {\footnotesize 200} & {\footnotesize 0.05} & {\footnotesize$\tilde{Q}_{\lambda,t}^{\left(r\right)}$} & {\footnotesize -0.9} & {\footnotesize 0.5} & {\footnotesize 1.00} & {\footnotesize 0.00} & {\footnotesize 5.00} & {\footnotesize 0.00}\tabularnewline
 & {\footnotesize 200} & {\footnotesize 0.05} & {\footnotesize$\tilde{Q}_{\lambda,t}^{\left(r\right)}$} & {\footnotesize -0.8} & {\footnotesize 0.5} & {\footnotesize 1.00} & {\footnotesize 0.00} & {\footnotesize 5.00} & {\footnotesize 0.00}\tabularnewline
 & {\footnotesize 200} & {\footnotesize 0.05} & {\footnotesize$\bar{Q}_{\lambda,t}^{\left(r\right)}$} & {\footnotesize 1} & {\footnotesize 0.9} & {\footnotesize 1.00} & {\footnotesize 0.00} & {\footnotesize 5.00} & {\footnotesize 0.00}\tabularnewline
 & {\footnotesize 200} & {\footnotesize 0.05} & {\footnotesize$\bar{Q}_{\lambda,t}^{\left(r\right)}$} & {\footnotesize 1} & {\footnotesize 0.95} & {\footnotesize 1.00} & {\footnotesize 0.00} & {\footnotesize 5.00} & {\footnotesize 0.00}\tabularnewline
{\footnotesize$\text{N}(0,2)$} & {\footnotesize 200} & {\footnotesize 0.01} & {\footnotesize$P_{t}$} &  &  & {\footnotesize 16.72} & {\footnotesize 4.58} & {\footnotesize 82.08} & {\footnotesize 10.36}\tabularnewline
 & {\footnotesize 200} & {\footnotesize 0.01} & {\footnotesize$\tilde{Q}_{\lambda,t}^{\left(r\right)}$} & {\footnotesize -0.9} & {\footnotesize 0.5} & {\footnotesize 1045.80} & {\footnotesize 316.58} & {\footnotesize 2536.42} & {\footnotesize 623.14}\tabularnewline
 & {\footnotesize 200} & {\footnotesize 0.01} & {\footnotesize$\tilde{Q}_{\lambda,t}^{\left(r\right)}$} & {\footnotesize -0.8} & {\footnotesize 0.5} & {\footnotesize 499.64} & {\footnotesize 221.70} & {\footnotesize 2085.12} & {\footnotesize 887.32}\tabularnewline
 & {\footnotesize 200} & {\footnotesize 0.01} & {\footnotesize$\bar{Q}_{\lambda,t}^{\left(r\right)}$} & {\footnotesize 1} & {\footnotesize 0.9} & {\footnotesize 246.16} & {\footnotesize 106.39} & {\footnotesize 990.62} & {\footnotesize 347.18}\tabularnewline
 & {\footnotesize 200} & {\footnotesize 0.01} & {\footnotesize$\bar{Q}_{\lambda,t}^{\left(r\right)}$} & {\footnotesize 1} & {\footnotesize 0.95} & {\footnotesize 43.49} & {\footnotesize 8.84} & {\footnotesize 237.56} & {\footnotesize 57.17}\tabularnewline
 & {\footnotesize 200} & {\footnotesize 0.05} & {\footnotesize$P_{t}$} &  &  & {\footnotesize 2.40} & {\footnotesize 0.26} & {\footnotesize 13.12} & {\footnotesize 0.91}\tabularnewline
 & {\footnotesize 200} & {\footnotesize 0.05} & {\footnotesize$\tilde{Q}_{\lambda,t}^{\left(r\right)}$} & {\footnotesize -0.9} & {\footnotesize 0.5} & {\footnotesize 71.54} & {\footnotesize 15.40} & {\footnotesize 230.44} & {\footnotesize 36.36}\tabularnewline
 & {\footnotesize 200} & {\footnotesize 0.05} & {\footnotesize$\tilde{Q}_{\lambda,t}^{\left(r\right)}$} & {\footnotesize -0.8} & {\footnotesize 0.5} & {\footnotesize 49.91} & {\footnotesize 11.32} & {\footnotesize 197.65} & {\footnotesize 39.28}\tabularnewline
 & {\footnotesize 200} & {\footnotesize 0.05} & {\footnotesize$\bar{Q}_{\lambda,t}^{\left(r\right)}$} & {\footnotesize 1} & {\footnotesize 0.9} & {\footnotesize 5.75} & {\footnotesize 1.44} & {\footnotesize 25.13} & {\footnotesize 3.84}\tabularnewline
 & {\footnotesize 200} & {\footnotesize 0.05} & {\footnotesize$\bar{Q}_{\lambda,t}^{\left(r\right)}$} & {\footnotesize 1} & {\footnotesize 0.95} & {\footnotesize 3.49} & {\footnotesize 0.50} & {\footnotesize 18.81} & {\footnotesize 2.80}\tabularnewline
{\footnotesize Cauchy} & {\footnotesize 200} & {\footnotesize 0.01} & {\footnotesize$P_{t}$} &  &  & {\footnotesize 2.70} & {\footnotesize 0.23} & {\footnotesize 12.25} & {\footnotesize 0.59}\tabularnewline
 & {\footnotesize 200} & {\footnotesize 0.01} & {\footnotesize$\tilde{Q}_{\lambda,t}^{\left(r\right)}$} & {\footnotesize -0.9} & {\footnotesize 0.5} & {\footnotesize 38.28} & {\footnotesize 5.31} & {\footnotesize 122.16} & {\footnotesize 14.95}\tabularnewline
 & {\footnotesize 200} & {\footnotesize 0.01} & {\footnotesize$\tilde{Q}_{\lambda,t}^{\left(r\right)}$} & {\footnotesize -0.8} & {\footnotesize 0.5} & {\footnotesize 30.11} & {\footnotesize 7.74} & {\footnotesize 77.77} & {\footnotesize 14.46}\tabularnewline
 & {\footnotesize 200} & {\footnotesize 0.01} & {\footnotesize$\bar{Q}_{\lambda,t}^{\left(r\right)}$} & {\footnotesize 1} & {\footnotesize 0.9} & {\footnotesize 2.61} & {\footnotesize 0.25} & {\footnotesize 13.47} & {\footnotesize 0.75}\tabularnewline
 & {\footnotesize 200} & {\footnotesize 0.01} & {\footnotesize$\bar{Q}_{\lambda,t}^{\left(r\right)}$} & {\footnotesize 1} & {\footnotesize 0.95} & {\footnotesize 2.77} & {\footnotesize 0.36} & {\footnotesize 11.60} & {\footnotesize 0.70}\tabularnewline
 & {\footnotesize 200} & {\footnotesize 0.05} & {\footnotesize$P_{t}$} &  &  & {\footnotesize 1.13} & {\footnotesize 0.04} & {\footnotesize 5.68} & {\footnotesize 0.10}\tabularnewline
 & {\footnotesize 200} & {\footnotesize 0.05} & {\footnotesize$\tilde{Q}_{\lambda,t}^{\left(r\right)}$} & {\footnotesize -0.9} & {\footnotesize 0.5} & {\footnotesize 7.00} & {\footnotesize 0.92} & {\footnotesize 26.34} & {\footnotesize 2.92}\tabularnewline
 & {\footnotesize 200} & {\footnotesize 0.05} & {\footnotesize$\tilde{Q}_{\lambda,t}^{\left(r\right)}$} & {\footnotesize -0.8} & {\footnotesize 0.5} & {\footnotesize 4.20} & {\footnotesize 0.58} & {\footnotesize 12.57} & {\footnotesize 0.95}\tabularnewline
 & {\footnotesize 200} & {\footnotesize 0.05} & {\footnotesize$\bar{Q}_{\lambda,t}^{\left(r\right)}$} & {\footnotesize 1} & {\footnotesize 0.9} & {\footnotesize 1.13} & {\footnotesize 0.04} & {\footnotesize 5.70} & {\footnotesize 0.11}\tabularnewline
 & {\footnotesize 200} & {\footnotesize 0.05} & {\footnotesize$\bar{Q}_{\lambda,t}^{\left(r\right)}$} & {\footnotesize 1} & {\footnotesize 0.95} & {\footnotesize 1.14} & {\footnotesize 0.04} & {\footnotesize 5.52} & {\footnotesize 0.08}\tabularnewline
\hline 
\end{tabular}}{\footnotesize\par}
\end{table}

\begin{table}
\caption{Average $R_{k}$ under dynamic OOC processes for KS $p$-value charts
$\left(P_{t}\right)_{t}$, $\left(\tilde{Q}_{\lambda,t}^{\left(r\right)}\right)_{t}$
and $\left(\bar{Q}_{\lambda,t}^{\left(r\right)}\right)_{t}$, using
100 simulation repetitions. This table contains results for $n_{0}=100$,
while Appendix Tables \ref{tab:KS_OC_EWMA_sim_dynamic} and \ref{tab:KS_OC_EWMA_sim_dynamic_part3}
contains results for $n_{0}=50$ and $n_{0}=200$, respectively.}\label{tab:KS_OC_EWMA_sim_dynamic_part2}

\centering{}{\footnotesize{}%
\begin{tabular}{cccccccccc}
\hline 
{\footnotesize OOC process} & {\footnotesize$n_{0}$} & {\footnotesize$\alpha$} & {\footnotesize Chart} & {\footnotesize$r$} & {\footnotesize$\lambda$} & {\footnotesize mean $R_{1}$} & {\footnotesize st. err. $R_{1}$} & {\footnotesize mean $R_{5}$} & {\footnotesize st. err. $R_{5}$}\tabularnewline
\hline 
{\footnotesize$\text{N}\left(\mu_{t},1\right)$} & {\footnotesize 100} & {\footnotesize 0.01} & {\footnotesize$P_{t}$} &  &  & {\footnotesize 1.69} & {\footnotesize 0.10} & {\footnotesize 9.05} & {\footnotesize 0.26}\tabularnewline
{\footnotesize$\mu_{t}\sim\text{N}\left(0,1/2\right)$} & {\footnotesize 100} & {\footnotesize 0.01} & {\footnotesize$\tilde{Q}_{\lambda,t}^{\left(r\right)}$} & {\footnotesize -0.9} & {\footnotesize 0.5} & {\footnotesize 2.72} & {\footnotesize 0.22} & {\footnotesize 6.84} & {\footnotesize 0.23}\tabularnewline
 & {\footnotesize 100} & {\footnotesize 0.01} & {\footnotesize$\tilde{Q}_{\lambda,t}^{\left(r\right)}$} & {\footnotesize -0.8} & {\footnotesize 0.5} & {\footnotesize 2.00} & {\footnotesize 0.15} & {\footnotesize 6.21} & {\footnotesize 0.19}\tabularnewline
 & {\footnotesize 100} & {\footnotesize 0.01} & {\footnotesize$\bar{Q}_{\lambda,t}^{\left(r\right)}$} & {\footnotesize 1} & {\footnotesize 0.9} & {\footnotesize 3.12} & {\footnotesize 0.35} & {\footnotesize 14.18} & {\footnotesize 0.67}\tabularnewline
 & {\footnotesize 100} & {\footnotesize 0.01} & {\footnotesize$\bar{Q}_{\lambda,t}^{\left(r\right)}$} & {\footnotesize 1} & {\footnotesize 0.95} & {\footnotesize 2.51} & {\footnotesize 0.25} & {\footnotesize 11.33} & {\footnotesize 0.50}\tabularnewline
 & {\footnotesize 100} & {\footnotesize 0.05} & {\footnotesize$P_{t}$} &  &  & {\footnotesize 1.67} & {\footnotesize 0.10} & {\footnotesize 8.04} & {\footnotesize 0.22}\tabularnewline
 & {\footnotesize 100} & {\footnotesize 0.05} & {\footnotesize$\tilde{Q}_{\lambda,t}^{\left(r\right)}$} & {\footnotesize -0.9} & {\footnotesize 0.5} & {\footnotesize 1.99} & {\footnotesize 0.15} & {\footnotesize 6.22} & {\footnotesize 0.18}\tabularnewline
 & {\footnotesize 100} & {\footnotesize 0.05} & {\footnotesize$\tilde{Q}_{\lambda,t}^{\left(r\right)}$} & {\footnotesize -0.8} & {\footnotesize 0.5} & {\footnotesize 1.98} & {\footnotesize 0.17} & {\footnotesize 6.09} & {\footnotesize 0.18}\tabularnewline
 & {\footnotesize 100} & {\footnotesize 0.05} & {\footnotesize$\bar{Q}_{\lambda,t}^{\left(r\right)}$} & {\footnotesize 1} & {\footnotesize 0.9} & {\footnotesize 2.10} & {\footnotesize 0.16} & {\footnotesize 9.25} & {\footnotesize 0.30}\tabularnewline
 & {\footnotesize 100} & {\footnotesize 0.05} & {\footnotesize$\bar{Q}_{\lambda,t}^{\left(r\right)}$} & {\footnotesize 1} & {\footnotesize 0.95} & {\footnotesize 1.70} & {\footnotesize 0.13} & {\footnotesize 7.86} & {\footnotesize 0.22}\tabularnewline
{\footnotesize$\text{N}\left(\mu_{t},1\right)$} & {\footnotesize 100} & {\footnotesize 0.01} & {\footnotesize$P_{t}$} &  &  & {\footnotesize 2.56} & {\footnotesize 0.21} & {\footnotesize 12.48} & {\footnotesize 0.45}\tabularnewline
{\footnotesize$\mu_{t}\sim\text{N}\left(0,1/4\right)$} & {\footnotesize 100} & {\footnotesize 0.01} & {\footnotesize$\tilde{Q}_{\lambda,t}^{\left(r\right)}$} & {\footnotesize -0.9} & {\footnotesize 0.5} & {\footnotesize 3.95} & {\footnotesize 0.35} & {\footnotesize 9.39} & {\footnotesize 0.62}\tabularnewline
 & {\footnotesize 100} & {\footnotesize 0.01} & {\footnotesize$\tilde{Q}_{\lambda,t}^{\left(r\right)}$} & {\footnotesize -0.8} & {\footnotesize 0.5} & {\footnotesize 2.81} & {\footnotesize 0.23} & {\footnotesize 7.67} & {\footnotesize 0.37}\tabularnewline
 & {\footnotesize 100} & {\footnotesize 0.01} & {\footnotesize$\bar{Q}_{\lambda,t}^{\left(r\right)}$} & {\footnotesize 1} & {\footnotesize 0.9} & {\footnotesize 5.07} & {\footnotesize 0.56} & {\footnotesize 23.07} & {\footnotesize 1.18}\tabularnewline
 & {\footnotesize 100} & {\footnotesize 0.01} & {\footnotesize$\bar{Q}_{\lambda,t}^{\left(r\right)}$} & {\footnotesize 1} & {\footnotesize 0.95} & {\footnotesize 4.61} & {\footnotesize 0.46} & {\footnotesize 19.82} & {\footnotesize 0.98}\tabularnewline
 & {\footnotesize 100} & {\footnotesize 0.05} & {\footnotesize$P_{t}$} &  &  & {\footnotesize 1.89} & {\footnotesize 0.13} & {\footnotesize 9.74} & {\footnotesize 0.30}\tabularnewline
 & {\footnotesize 100} & {\footnotesize 0.05} & {\footnotesize$\tilde{Q}_{\lambda,t}^{\left(r\right)}$} & {\footnotesize -0.9} & {\footnotesize 0.5} & {\footnotesize 2.50} & {\footnotesize 0.17} & {\footnotesize 7.07} & {\footnotesize 0.26}\tabularnewline
 & {\footnotesize 100} & {\footnotesize 0.05} & {\footnotesize$\tilde{Q}_{\lambda,t}^{\left(r\right)}$} & {\footnotesize -0.8} & {\footnotesize 0.5} & {\footnotesize 2.96} & {\footnotesize 0.23} & {\footnotesize 7.34} & {\footnotesize 0.26}\tabularnewline
 & {\footnotesize 100} & {\footnotesize 0.05} & {\footnotesize$\bar{Q}_{\lambda,t}^{\left(r\right)}$} & {\footnotesize 1} & {\footnotesize 0.9} & {\footnotesize 2.08} & {\footnotesize 0.19} & {\footnotesize 11.61} & {\footnotesize 0.41}\tabularnewline
 & {\footnotesize 100} & {\footnotesize 0.05} & {\footnotesize$\bar{Q}_{\lambda,t}^{\left(r\right)}$} & {\footnotesize 1} & {\footnotesize 0.95} & {\footnotesize 2.29} & {\footnotesize 0.20} & {\footnotesize 10.01} & {\footnotesize 0.36}\tabularnewline
{\footnotesize$\text{N}\left(\mu_{t},1\right)$} & {\footnotesize 100} & {\footnotesize 0.01} & {\footnotesize$P_{t}$} &  &  & {\footnotesize 2.14} & {\footnotesize 0.16} & {\footnotesize 10.91} & {\footnotesize 0.38}\tabularnewline
{\footnotesize$\sigma_{t}^{2}\sim\chi_{1}^{2}$} & {\footnotesize 100} & {\footnotesize 0.01} & {\footnotesize$\tilde{Q}_{\lambda,t}^{\left(r\right)}$} & {\footnotesize -0.9} & {\footnotesize 0.5} & {\footnotesize 3.26} & {\footnotesize 0.27} & {\footnotesize 7.54} & {\footnotesize 0.28}\tabularnewline
 & {\footnotesize 100} & {\footnotesize 0.01} & {\footnotesize$\tilde{Q}_{\lambda,t}^{\left(r\right)}$} & {\footnotesize -0.8} & {\footnotesize 0.5} & {\footnotesize 3.04} & {\footnotesize 0.25} & {\footnotesize 7.35} & {\footnotesize 0.27}\tabularnewline
 & {\footnotesize 100} & {\footnotesize 0.01} & {\footnotesize$\bar{Q}_{\lambda,t}^{\left(r\right)}$} & {\footnotesize 1} & {\footnotesize 0.9} & {\footnotesize 4.07} & {\footnotesize 0.52} & {\footnotesize 18.76} & {\footnotesize 0.95}\tabularnewline
 & {\footnotesize 100} & {\footnotesize 0.01} & {\footnotesize$\bar{Q}_{\lambda,t}^{\left(r\right)}$} & {\footnotesize 1} & {\footnotesize 0.95} & {\footnotesize 3.51} & {\footnotesize 0.37} & {\footnotesize 16.66} & {\footnotesize 0.85}\tabularnewline
 & {\footnotesize 100} & {\footnotesize 0.05} & {\footnotesize$P_{t}$} &  &  & {\footnotesize 1.94} & {\footnotesize 0.14} & {\footnotesize 9.32} & {\footnotesize 0.30}\tabularnewline
 & {\footnotesize 100} & {\footnotesize 0.05} & {\footnotesize$\tilde{Q}_{\lambda,t}^{\left(r\right)}$} & {\footnotesize -0.9} & {\footnotesize 0.5} & {\footnotesize 3.15} & {\footnotesize 0.29} & {\footnotesize 7.54} & {\footnotesize 0.31}\tabularnewline
 & {\footnotesize 100} & {\footnotesize 0.05} & {\footnotesize$\tilde{Q}_{\lambda,t}^{\left(r\right)}$} & {\footnotesize -0.8} & {\footnotesize 0.5} & {\footnotesize 2.59} & {\footnotesize 0.17} & {\footnotesize 7.14} & {\footnotesize 0.24}\tabularnewline
 & {\footnotesize 100} & {\footnotesize 0.05} & {\footnotesize$\bar{Q}_{\lambda,t}^{\left(r\right)}$} & {\footnotesize 1} & {\footnotesize 0.9} & {\footnotesize 2.18} & {\footnotesize 0.21} & {\footnotesize 10.65} & {\footnotesize 0.50}\tabularnewline
 & {\footnotesize 100} & {\footnotesize 0.05} & {\footnotesize$\bar{Q}_{\lambda,t}^{\left(r\right)}$} & {\footnotesize 1} & {\footnotesize 0.95} & {\footnotesize 1.91} & {\footnotesize 0.14} & {\footnotesize 8.94} & {\footnotesize 0.30}\tabularnewline
{\footnotesize$\text{N}\left(\mu_{t},1\right)$} & {\footnotesize 100} & {\footnotesize 0.01} & {\footnotesize$P_{t}$} &  &  & {\footnotesize 3.41} & {\footnotesize 0.32} & {\footnotesize 18.14} & {\footnotesize 0.91}\tabularnewline
{\footnotesize$\sigma_{t}^{2}\sim\chi_{2}^{2}$} & {\footnotesize 100} & {\footnotesize 0.01} & {\footnotesize$\tilde{Q}_{\lambda,t}^{\left(r\right)}$} & {\footnotesize -0.9} & {\footnotesize 0.5} & {\footnotesize 9.01} & {\footnotesize 0.95} & {\footnotesize 18.93} & {\footnotesize 1.41}\tabularnewline
 & {\footnotesize 100} & {\footnotesize 0.01} & {\footnotesize$\tilde{Q}_{\lambda,t}^{\left(r\right)}$} & {\footnotesize -0.8} & {\footnotesize 0.5} & {\footnotesize 7.05} & {\footnotesize 0.59} & {\footnotesize 16.68} & {\footnotesize 1.02}\tabularnewline
 & {\footnotesize 100} & {\footnotesize 0.01} & {\footnotesize$\bar{Q}_{\lambda,t}^{\left(r\right)}$} & {\footnotesize 1} & {\footnotesize 0.9} & {\footnotesize 9.14} & {\footnotesize 1.20} & {\footnotesize 48.17} & {\footnotesize 3.24}\tabularnewline
 & {\footnotesize 100} & {\footnotesize 0.01} & {\footnotesize$\bar{Q}_{\lambda,t}^{\left(r\right)}$} & {\footnotesize 1} & {\footnotesize 0.95} & {\footnotesize 5.77} & {\footnotesize 0.62} & {\footnotesize 31.56} & {\footnotesize 1.75}\tabularnewline
 & {\footnotesize 100} & {\footnotesize 0.05} & {\footnotesize$P_{t}$} &  &  & {\footnotesize 2.40} & {\footnotesize 0.20} & {\footnotesize 11.22} & {\footnotesize 0.43}\tabularnewline
 & {\footnotesize 100} & {\footnotesize 0.05} & {\footnotesize$\tilde{Q}_{\lambda,t}^{\left(r\right)}$} & {\footnotesize -0.9} & {\footnotesize 0.5} & {\footnotesize 5.61} & {\footnotesize 0.46} & {\footnotesize 12.52} & {\footnotesize 0.82}\tabularnewline
 & {\footnotesize 100} & {\footnotesize 0.05} & {\footnotesize$\tilde{Q}_{\lambda,t}^{\left(r\right)}$} & {\footnotesize -0.8} & {\footnotesize 0.5} & {\footnotesize 4.28} & {\footnotesize 0.43} & {\footnotesize 10.75} & {\footnotesize 0.66}\tabularnewline
 & {\footnotesize 100} & {\footnotesize 0.05} & {\footnotesize$\bar{Q}_{\lambda,t}^{\left(r\right)}$} & {\footnotesize 1} & {\footnotesize 0.9} & {\footnotesize 2.98} & {\footnotesize 0.29} & {\footnotesize 15.34} & {\footnotesize 0.77}\tabularnewline
 & {\footnotesize 100} & {\footnotesize 0.05} & {\footnotesize$\bar{Q}_{\lambda,t}^{\left(r\right)}$} & {\footnotesize 1} & {\footnotesize 0.95} & {\footnotesize 2.41} & {\footnotesize 0.23} & {\footnotesize 11.92} & {\footnotesize 0.48}\tabularnewline
\hline 
\end{tabular}}{\footnotesize\par}
\end{table}

\begin{table}
\caption{Average $R_{k}$ under dynamic OOC processes for KS $p$-value charts
$\left(P_{t}\right)_{t}$, $\left(\tilde{Q}_{\lambda,t}^{\left(r\right)}\right)_{t}$
and $\left(\bar{Q}_{\lambda,t}^{\left(r\right)}\right)_{t}$, using
100 simulation repetitions. This table contains results for $n_{0}=200$,
while Appendix Tables \ref{tab:KS_OC_EWMA_sim_dynamic} and \ref{tab:KS_OC_EWMA_sim_dynamic_part2}
contains results for $n_{0}=50$ and $n_{0}=100$, respectively.}\label{tab:KS_OC_EWMA_sim_dynamic_part3}

\centering{}{\footnotesize{}%
\begin{tabular}{cccccccccc}
\hline 
{\footnotesize OOC process} & {\footnotesize$n_{0}$} & {\footnotesize$\alpha$} & {\footnotesize Chart} & {\footnotesize$r$} & {\footnotesize$\lambda$} & {\footnotesize mean $R_{1}$} & {\footnotesize st. err. $R_{1}$} & {\footnotesize mean $R_{5}$} & {\footnotesize st. err. $R_{5}$}\tabularnewline
\hline 
{\footnotesize$\text{N}\left(\mu_{t},1\right)$} & {\footnotesize 200} & {\footnotesize 0.01} & {\footnotesize$P_{t}$} &  &  & {\footnotesize 1.49} & {\footnotesize 0.07} & {\footnotesize 7.69} & {\footnotesize 0.20}\tabularnewline
{\footnotesize$\mu_{t}\sim\text{N}\left(0,1/2\right)$} & {\footnotesize 200} & {\footnotesize 0.01} & {\footnotesize$\tilde{Q}_{\lambda,t}^{\left(r\right)}$} & {\footnotesize -0.9} & {\footnotesize 0.5} & {\footnotesize 1.65} & {\footnotesize 0.12} & {\footnotesize 5.67} & {\footnotesize 0.12}\tabularnewline
 & {\footnotesize 200} & {\footnotesize 0.01} & {\footnotesize$\tilde{Q}_{\lambda,t}^{\left(r\right)}$} & {\footnotesize -0.8} & {\footnotesize 0.5} & {\footnotesize 1.84} & {\footnotesize 0.13} & {\footnotesize 5.88} & {\footnotesize 0.13}\tabularnewline
 & {\footnotesize 200} & {\footnotesize 0.01} & {\footnotesize$\bar{Q}_{\lambda,t}^{\left(r\right)}$} & {\footnotesize 1} & {\footnotesize 0.9} & {\footnotesize 1.88} & {\footnotesize 0.18} & {\footnotesize 9.65} & {\footnotesize 0.41}\tabularnewline
 & {\footnotesize 200} & {\footnotesize 0.01} & {\footnotesize$\bar{Q}_{\lambda,t}^{\left(r\right)}$} & {\footnotesize 1} & {\footnotesize 0.95} & {\footnotesize 1.71} & {\footnotesize 0.14} & {\footnotesize 8.65} & {\footnotesize 0.29}\tabularnewline
 & {\footnotesize 200} & {\footnotesize 0.05} & {\footnotesize$P_{t}$} &  &  & {\footnotesize 1.42} & {\footnotesize 0.09} & {\footnotesize 6.61} & {\footnotesize 0.16}\tabularnewline
 & {\footnotesize 200} & {\footnotesize 0.05} & {\footnotesize$\tilde{Q}_{\lambda,t}^{\left(r\right)}$} & {\footnotesize -0.9} & {\footnotesize 0.5} & {\footnotesize 1.62} & {\footnotesize 0.10} & {\footnotesize 5.68} & {\footnotesize 0.11}\tabularnewline
 & {\footnotesize 200} & {\footnotesize 0.05} & {\footnotesize$\tilde{Q}_{\lambda,t}^{\left(r\right)}$} & {\footnotesize -0.8} & {\footnotesize 0.5} & {\footnotesize 1.68} & {\footnotesize 0.10} & {\footnotesize 5.70} & {\footnotesize 0.10}\tabularnewline
 & {\footnotesize 200} & {\footnotesize 0.05} & {\footnotesize$\bar{Q}_{\lambda,t}^{\left(r\right)}$} & {\footnotesize 1} & {\footnotesize 0.9} & {\footnotesize 1.44} & {\footnotesize 0.09} & {\footnotesize 7.58} & {\footnotesize 0.23}\tabularnewline
 & {\footnotesize 200} & {\footnotesize 0.05} & {\footnotesize$\bar{Q}_{\lambda,t}^{\left(r\right)}$} & {\footnotesize 1} & {\footnotesize 0.95} & {\footnotesize 1.38} & {\footnotesize 0.07} & {\footnotesize 6.61} & {\footnotesize 0.16}\tabularnewline
{\footnotesize$\text{N}\left(\mu_{t},1\right)$} & {\footnotesize 200} & {\footnotesize 0.01} & {\footnotesize$P_{t}$} &  &  & {\footnotesize 1.82} & {\footnotesize 0.11} & {\footnotesize 9.48} & {\footnotesize 0.29}\tabularnewline
{\footnotesize$\mu_{t}\sim\text{N}\left(0,1/4\right)$} & {\footnotesize 200} & {\footnotesize 0.01} & {\footnotesize$\tilde{Q}_{\lambda,t}^{\left(r\right)}$} & {\footnotesize -0.9} & {\footnotesize 0.5} & {\footnotesize 2.37} & {\footnotesize 0.18} & {\footnotesize 6.61} & {\footnotesize 0.20}\tabularnewline
 & {\footnotesize 200} & {\footnotesize 0.01} & {\footnotesize$\tilde{Q}_{\lambda,t}^{\left(r\right)}$} & {\footnotesize -0.8} & {\footnotesize 0.5} & {\footnotesize 2.27} & {\footnotesize 0.17} & {\footnotesize 6.61} & {\footnotesize 0.19}\tabularnewline
 & {\footnotesize 200} & {\footnotesize 0.01} & {\footnotesize$\bar{Q}_{\lambda,t}^{\left(r\right)}$} & {\footnotesize 1} & {\footnotesize 0.9} & {\footnotesize 2.99} & {\footnotesize 0.31} & {\footnotesize 12.85} & {\footnotesize 0.62}\tabularnewline
 & {\footnotesize 200} & {\footnotesize 0.01} & {\footnotesize$\bar{Q}_{\lambda,t}^{\left(r\right)}$} & {\footnotesize 1} & {\footnotesize 0.95} & {\footnotesize 2.49} & {\footnotesize 0.22} & {\footnotesize 12.28} & {\footnotesize 0.53}\tabularnewline
 & {\footnotesize 200} & {\footnotesize 0.05} & {\footnotesize$P_{t}$} &  &  & {\footnotesize 1.40} & {\footnotesize 0.08} & {\footnotesize 7.94} & {\footnotesize 0.23}\tabularnewline
 & {\footnotesize 200} & {\footnotesize 0.05} & {\footnotesize$\tilde{Q}_{\lambda,t}^{\left(r\right)}$} & {\footnotesize -0.9} & {\footnotesize 0.5} & {\footnotesize 2.15} & {\footnotesize 0.17} & {\footnotesize 6.29} & {\footnotesize 0.18}\tabularnewline
 & {\footnotesize 200} & {\footnotesize 0.05} & {\footnotesize$\tilde{Q}_{\lambda,t}^{\left(r\right)}$} & {\footnotesize -0.8} & {\footnotesize 0.5} & {\footnotesize 1.92} & {\footnotesize 0.13} & {\footnotesize 6.05} & {\footnotesize 0.15}\tabularnewline
 & {\footnotesize 200} & {\footnotesize 0.05} & {\footnotesize$\bar{Q}_{\lambda,t}^{\left(r\right)}$} & {\footnotesize 1} & {\footnotesize 0.9} & {\footnotesize 1.85} & {\footnotesize 0.15} & {\footnotesize 8.84} & {\footnotesize 0.27}\tabularnewline
 & {\footnotesize 200} & {\footnotesize 0.05} & {\footnotesize$\bar{Q}_{\lambda,t}^{\left(r\right)}$} & {\footnotesize 1} & {\footnotesize 0.95} & {\footnotesize 1.52} & {\footnotesize 0.10} & {\footnotesize 7.92} & {\footnotesize 0.22}\tabularnewline
{\footnotesize$\text{N}\left(\mu_{t},1\right)$} & {\footnotesize 200} & {\footnotesize 0.01} & {\footnotesize$P_{t}$} &  &  & {\footnotesize 1.68} & {\footnotesize 0.10} & {\footnotesize 8.81} & {\footnotesize 0.28}\tabularnewline
{\footnotesize$\sigma_{t}^{2}\sim\chi_{1}^{2}$} & {\footnotesize 200} & {\footnotesize 0.01} & {\footnotesize$\tilde{Q}_{\lambda,t}^{\left(r\right)}$} & {\footnotesize -0.9} & {\footnotesize 0.5} & {\footnotesize 2.00} & {\footnotesize 0.16} & {\footnotesize 6.11} & {\footnotesize 0.17}\tabularnewline
 & {\footnotesize 200} & {\footnotesize 0.01} & {\footnotesize$\tilde{Q}_{\lambda,t}^{\left(r\right)}$} & {\footnotesize -0.8} & {\footnotesize 0.5} & {\footnotesize 2.13} & {\footnotesize 0.17} & {\footnotesize 6.40} & {\footnotesize 0.22}\tabularnewline
 & {\footnotesize 200} & {\footnotesize 0.01} & {\footnotesize$\bar{Q}_{\lambda,t}^{\left(r\right)}$} & {\footnotesize 1} & {\footnotesize 0.9} & {\footnotesize 2.54} & {\footnotesize 0.23} & {\footnotesize 11.91} & {\footnotesize 0.56}\tabularnewline
 & {\footnotesize 200} & {\footnotesize 0.01} & {\footnotesize$\bar{Q}_{\lambda,t}^{\left(r\right)}$} & {\footnotesize 1} & {\footnotesize 0.95} & {\footnotesize 2.06} & {\footnotesize 0.19} & {\footnotesize 9.76} & {\footnotesize 0.38}\tabularnewline
 & {\footnotesize 200} & {\footnotesize 0.05} & {\footnotesize$P_{t}$} &  &  & {\footnotesize 1.44} & {\footnotesize 0.07} & {\footnotesize 7.35} & {\footnotesize 0.18}\tabularnewline
 & {\footnotesize 200} & {\footnotesize 0.05} & {\footnotesize$\tilde{Q}_{\lambda,t}^{\left(r\right)}$} & {\footnotesize -0.9} & {\footnotesize 0.5} & {\footnotesize 1.80} & {\footnotesize 0.12} & {\footnotesize 5.93} & {\footnotesize 0.14}\tabularnewline
 & {\footnotesize 200} & {\footnotesize 0.05} & {\footnotesize$\tilde{Q}_{\lambda,t}^{\left(r\right)}$} & {\footnotesize -0.8} & {\footnotesize 0.5} & {\footnotesize 1.75} & {\footnotesize 0.13} & {\footnotesize 5.84} & {\footnotesize 0.14}\tabularnewline
 & {\footnotesize 200} & {\footnotesize 0.05} & {\footnotesize$\bar{Q}_{\lambda,t}^{\left(r\right)}$} & {\footnotesize 1} & {\footnotesize 0.9} & {\footnotesize 1.54} & {\footnotesize 0.11} & {\footnotesize 8.04} & {\footnotesize 0.27}\tabularnewline
 & {\footnotesize 200} & {\footnotesize 0.05} & {\footnotesize$\bar{Q}_{\lambda,t}^{\left(r\right)}$} & {\footnotesize 1} & {\footnotesize 0.95} & {\footnotesize 1.56} & {\footnotesize 0.08} & {\footnotesize 7.57} & {\footnotesize 0.20}\tabularnewline
{\footnotesize$\text{N}\left(\mu_{t},1\right)$} & {\footnotesize 200} & {\footnotesize 0.01} & {\footnotesize$P_{t}$} &  &  & {\footnotesize 2.22} & {\footnotesize 0.16} & {\footnotesize 10.53} & {\footnotesize 0.31}\tabularnewline
{\footnotesize$\sigma_{t}^{2}\sim\chi_{2}^{2}$} & {\footnotesize 200} & {\footnotesize 0.01} & {\footnotesize$\tilde{Q}_{\lambda,t}^{\left(r\right)}$} & {\footnotesize -0.9} & {\footnotesize 0.5} & {\footnotesize 3.36} & {\footnotesize 0.28} & {\footnotesize 8.59} & {\footnotesize 0.40}\tabularnewline
 & {\footnotesize 200} & {\footnotesize 0.01} & {\footnotesize$\tilde{Q}_{\lambda,t}^{\left(r\right)}$} & {\footnotesize -0.8} & {\footnotesize 0.5} & {\footnotesize 3.65} & {\footnotesize 0.32} & {\footnotesize 8.52} & {\footnotesize 0.42}\tabularnewline
 & {\footnotesize 200} & {\footnotesize 0.01} & {\footnotesize$\bar{Q}_{\lambda,t}^{\left(r\right)}$} & {\footnotesize 1} & {\footnotesize 0.9} & {\footnotesize 3.57} & {\footnotesize 0.32} & {\footnotesize 16.76} & {\footnotesize 1.04}\tabularnewline
 & {\footnotesize 200} & {\footnotesize 0.01} & {\footnotesize$\bar{Q}_{\lambda,t}^{\left(r\right)}$} & {\footnotesize 1} & {\footnotesize 0.95} & {\footnotesize 3.05} & {\footnotesize 0.35} & {\footnotesize 16.25} & {\footnotesize 0.86}\tabularnewline
 & {\footnotesize 200} & {\footnotesize 0.05} & {\footnotesize$P_{t}$} &  &  & {\footnotesize 1.59} & {\footnotesize 0.10} & {\footnotesize 8.24} & {\footnotesize 0.26}\tabularnewline
 & {\footnotesize 200} & {\footnotesize 0.05} & {\footnotesize$\tilde{Q}_{\lambda,t}^{\left(r\right)}$} & {\footnotesize -0.9} & {\footnotesize 0.5} & {\footnotesize 2.63} & {\footnotesize 0.22} & {\footnotesize 7.11} & {\footnotesize 0.26}\tabularnewline
 & {\footnotesize 200} & {\footnotesize 0.05} & {\footnotesize$\tilde{Q}_{\lambda,t}^{\left(r\right)}$} & {\footnotesize -0.8} & {\footnotesize 0.5} & {\footnotesize 2.53} & {\footnotesize 0.19} & {\footnotesize 6.86} & {\footnotesize 0.20}\tabularnewline
 & {\footnotesize 200} & {\footnotesize 0.05} & {\footnotesize$\bar{Q}_{\lambda,t}^{\left(r\right)}$} & {\footnotesize 1} & {\footnotesize 0.9} & {\footnotesize 2.04} & {\footnotesize 0.18} & {\footnotesize 10.08} & {\footnotesize 0.45}\tabularnewline
 & {\footnotesize 200} & {\footnotesize 0.05} & {\footnotesize$\bar{Q}_{\lambda,t}^{\left(r\right)}$} & {\footnotesize 1} & {\footnotesize 0.95} & {\footnotesize 1.82} & {\footnotesize 0.13} & {\footnotesize 8.51} & {\footnotesize 0.31}\tabularnewline
\hline 
\end{tabular}}{\footnotesize\par}
\end{table}

\begin{figure}
\begin{centering}
\includegraphics[width=15cm]{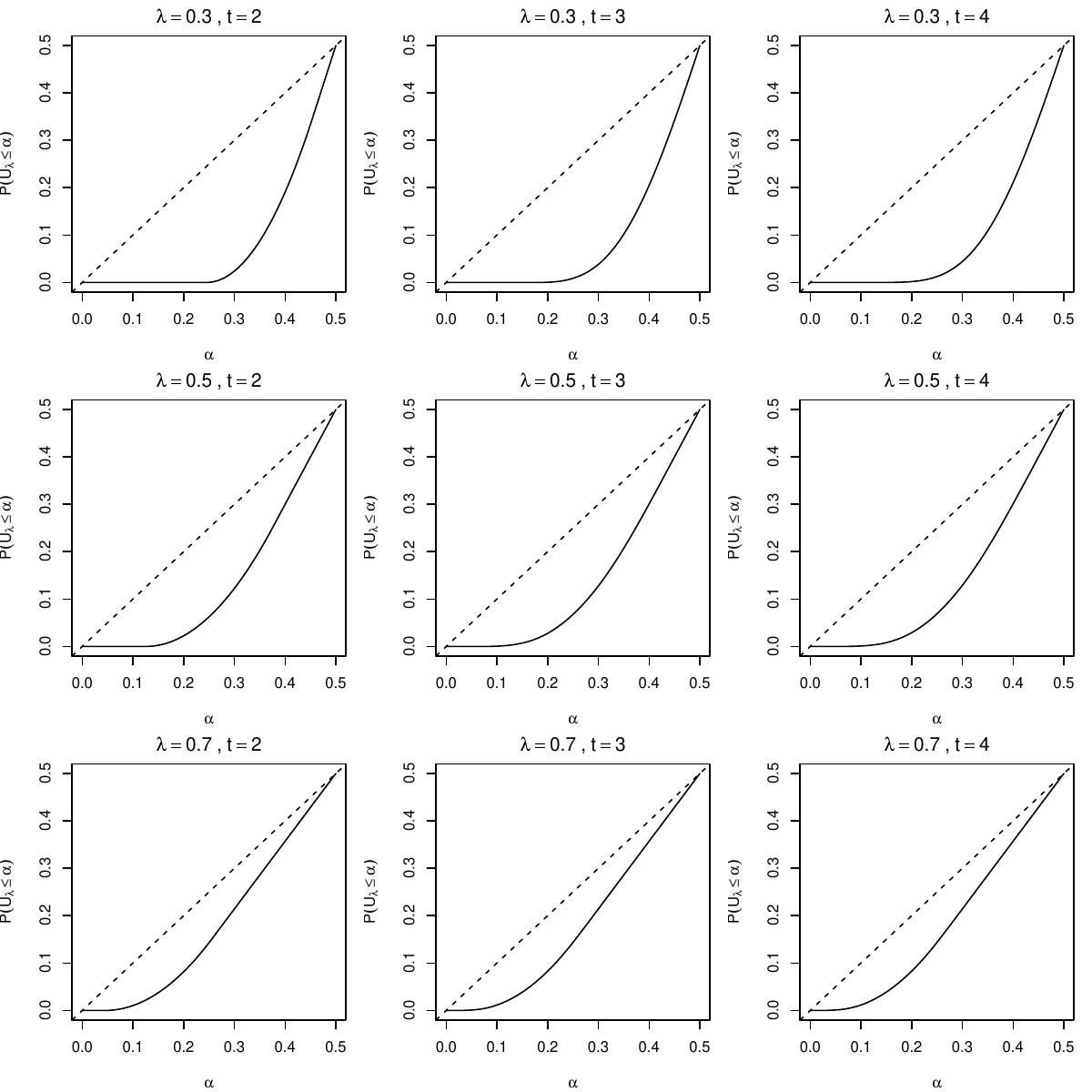}
\par\end{centering}
\caption{Plots of CDFs of the random variables $\tilde{U}_{\lambda,t}$ with
initialisation $u_{0}=1/2$, i.e., $F\left(\alpha\right)=\mathrm{P}_{0}\left(\tilde{U}_{\lambda,t}\le\alpha\right)$,
for $\lambda\in\left\{ 0.3,0.5,0.7\right\} $ and $t\in\left\{ 2,3,4\right\} $
(solid line) along with the CDF of $U\sim\mathrm{Unif}\left[0,1\right]$
(dashed line) for $\alpha\le1/2$.}\label{fig:Plots-of-CDFs}
\end{figure}

\begin{table}
\caption{Average run lengths $R$, numbers of detected OOC directions at time
of alarm \#OOC, and numbers of single time point family wise errors
(at time $t=1$) for $d=3$ variable Cauchy model with $n_{0}=20$
tested using the procedure from Section \ref{subsec:Directional chart}
with Mann--Whitney $U$-tests. Average run lengths and number of
true positives are reported via 100 simulation replications, while
single time point family wise errors are obtained from 10000 simulations.}\label{tab:Cauchy_20}

\centering{}%
\begin{tabular}{cccccc}
\hline 
$\delta$ & $\rho$ & $\alpha$ & Mean $R$ & Mean \#OOC & Mean FWEs\tabularnewline
\hline 
0.5 & 0 & 0.01 & 397 & 0.78 & 0.0031\tabularnewline
0.5 & 0 & 0.05 & 90.17 & 0.73 & 0.0184\tabularnewline
0.5 & 0.5 & 0.01 & 1762.03 & 0.83 & 0.0041\tabularnewline
0.5 & 0.5 & 0.05 & 37.99 & 0.9 & 0.0193\tabularnewline
0.5 & 0.9 & 0.01 & 2963.58 & 0.94 & 0.0051\tabularnewline
0.5 & 0.9 & 0.05 & 82.84 & 0.99 & 0.0295\tabularnewline
1 & 0 & 0.01 & 68.74 & 1.01 & 0.0038\tabularnewline
1 & 0 & 0.05 & 8.07 & 1.05 & 0.0211\tabularnewline
1 & 0.5 & 0.01 & 1052.82 & 0.97 & 0.0045\tabularnewline
1 & 0.5 & 0.05 & 4.77 & 1.03 & 0.0217\tabularnewline
1 & 0.9 & 0.01 & 44.74 & 1 & 0.0054\tabularnewline
1 & 0.9 & 0.05 & 6.05 & 1.01 & 0.023\tabularnewline
\hline 
\end{tabular}
\end{table}

\begin{table}
\caption{Average run lengths $R$, numbers of detected OOC directions at time
of alarm \#OOC, and numbers of single time point family wise errors
(at time $t=1$) for $d=3$ variable Cauchy model with $n_{0}=100$
tested using the procedure from Section \ref{subsec:Directional chart}
with Mann--Whitney $U$-tests. Average run lengths and number of
true positives are reported via 100 simulation replications, while
single time point family wise errors are obtained from 10000 simulations.}\label{tab:Cauchy_100}

\centering{}%
\begin{tabular}{cccccc}
\hline 
$\delta$ & $\rho$ & $\alpha$ & Mean $R$ & Mean \#OOC & Mean FWEs\tabularnewline
\hline 
0.5 & 0 & 0.01 & 17.45 & 1.06 & 0.0042\tabularnewline
0.5 & 0 & 0.05 & 2.65 & 1.12 & 0.0259\tabularnewline
0.5 & 0.5 & 0.01 & 9.9 & 1.01 & 0.0055\tabularnewline
0.5 & 0.5 & 0.05 & 2.69 & 1.12 & 0.0238\tabularnewline
0.5 & 0.9 & 0.01 & 6.58 & 1 & 0.0052\tabularnewline
0.5 & 0.9 & 0.05 & 2.23 & 1.04 & 0.0298\tabularnewline
1 & 0 & 0.01 & 1.1 & 1.67 & 0.0078\tabularnewline
1 & 0 & 0.05 & 1 & 1.81 & 0.0468\tabularnewline
1 & 0.5 & 0.01 & 1.03 & 1.68 & 0.0073\tabularnewline
1 & 0.5 & 0.05 & 1 & 1.78 & 0.0427\tabularnewline
1 & 0.9 & 0.01 & 1 & 1.64 & 0.0059\tabularnewline
1 & 0.9 & 0.05 & 1 & 1.84 & 0.0326\tabularnewline
\hline 
\end{tabular}
\end{table}

\end{document}